\documentclass{amsart}
\usepackage{amssymb,graphicx,hyperref}
\usepackage[utf8]{inputenc}

\let\frontmatter\relax
\def\mainmatter{\def\baselinestretch{1.1}\normalfont}
\def\backmatter{\def\baselinestretch{1}\normalfont}
\let\Subsection\subsection
\let\backslash\smallsetminus
\let\setminus\smallsetminus

\let\leq\leqslant   
\let\geq\geqslant
\let\emptyset\varnothing
\let\tilde\widetilde
\let\hat\widehat

\let\ldots\dots
\def\to{\mathchoice{\longrightarrow}{\rightarrow}{\rightarrow}{\rightarrow}}
\newcommand{\mto}{\mathchoice{\longmapsto}{\mapsto}{\mapsto}{\mapsto}}
\newcommand{\eprint}[1]{\href{http://arxiv.org/abs/#1}{\texttt{arXiv\string:\allowbreak#1}}}
\newcommand{\initial}[1]{#1}
\newcommand{\firstname}[1]{#1}
\newcommand{\lastname}[1]{#1}
\newcommand{\sfrac}[2]{{#1}/{#2}}
\newcommand{\Psfrac}[2]{(\sfrac{#1}{#2})}
\newcommand{\psfrac}[2]{\sfrac{(#1)}{#2}}
\newcommand{\spfrac}[2]{\sfrac{#1}{(#2)}}
\newcommand{\ppfrac}[2]{\sfrac{(#1)}{(#2)}}
\newcommand{\bmt}{\boldsymbol{t}}
\newcommand{\bmz}{\boldsymbol{z}}
\newcommand{\bmD}{\boldsymbol{D}}
\newcommand{\ie}{i.e.,\ }
\let\div\divv
\DeclareMathOperator{\Div}{Div}
\DeclareMathOperator{\Res}{Res}
\DeclareMathOperator{\Sym}{Sym}

\theoremstyle{plain}
\newtheorem{thm}{Theorem}[section]
\newtheorem{cor}[thm]{Corollary}
\newtheorem{lem}[thm]{Lemma}

\theoremstyle{definition}
\newtheorem{defin}[thm]{Definition}
\newtheorem{rem}[thm]{Remark}
\newtheorem{exa}[thm]{Example}

\numberwithin{equation}{section}

\begin{document}
\frontmatter
\title{Reconstructing WKB from topological~recursion}

\author[\initial{V.} \lastname{Bouchard}]{\firstname{Vincent} \lastname{Bouchard}}
\address{Department of Mathematical \& Statistical Sciences,
University of Alberta, 632 CAB\\
Edmonton, Alberta, Canada T6G 2G1}
\email{vincent.bouchard@ualberta.ca}
\urladdr{https://sites.ualberta.ca/~vbouchar/}

\author[\initial{B.} \lastname{Eynard}]{\firstname{Bertrand} \lastname{Eynard}}
\address{Institut de Physique Théorique, CEA Saclay\\
91191 Gif-sur-Yvette cedex, France}
\email{bertrand.eynard@cea.fr}
\urladdr{http://ipht.cea.fr/Pisp/33/bertrand.eynard.html}

\thanks{V.B. acknowledges the support of the Natural Sciences and Engineering Research Council of Canada. B.E. acknowledges support from Centre de Recherches Mathématiques de Montréal, a FQRNT
grant from the Québec government, and Piotr Su\l kowski and the ERC starting grant Fields-Knots.\\
Preprint numbers: CRM-3354(2016) and IPHT:T16-056.}



\begin{abstract}
We prove that the topological recursion reconstructs the WKB expansion of a quantum curve for all spectral curves whose Newton polygons have no interior point (and that are smooth as affine curves). This includes nearly all previously known cases in the literature, and many more; in particular, it includes many quantum curves of order greater than two. We also explore the connection between the choice of ordering in the quantization of the spectral curve and the choice of integration divisor to reconstruct the WKB expansion.
\end{abstract}

\maketitle
\tableofcontents
\mainmatter

\section{Introduction}

The topological recursion originally introduced in \cite{E1,CEO,EO,EO2} is now understood as being a rather universal formalism that reconstructs generating functions for various enumerative invariants from the data of a spectral curve. While it originated in the context of matrix models, it has now been shown to be closely related to other fundamental structures in enumerative geometry, such as Virasoro constraints, Frobenius structures, Givental formalism and cohomological field theories \cite{ACNP, ACNP2,DBOSS,DBNOPS,Mi,Or,KZ}. This explains, in part, why the topological recursion appears in so many different algebro-geometric context.

Another connection between the topological recursion and fundamental mathematical structures has been studied in recent years. With the intuition coming from determinantal formulae in the matrix model realm \cite{BeE, BeE2}, it has been conjectured that the topological recursion reconstructs the WKB asymptotic solution of Schrödinger-like ordinary differential equations, known as \emph{quantum curves}. More precisely, the claim is that there exists a Schrödinger-like ordinary differential operator, which is a \emph{quantization} of the original spectral curve (which is why it is called a quantum curve), and whose WKB asymptotic solution is reconstructed by the topological recursion applied to this spectral curve. This claim \cite{EO, BeE, BeE2, BorE1} has been verified for a small number of genus zero spectral curves, in various algebro-geometric contexts \cite{Al,BeE,BSLM,DDM,DM,DN,DuM,DuM2,DBMNPS,LMS,MSS,MS,Safnuk}. In the context of knot theory, this claim provides a constructive approach to the well known AJ-conjecture \cite{Ga}, which has been studied in a number of papers \cite{DFM,BoE,DFM,FGS,FGS2,GKS,GS}. The quantum curve connection also appears in the context of mirror symmetry for toric Calabi-Yau threefolds, in which the topological recursion reconstructs the mirror B-model theory \cite{BKMP,EO3,FLZ,Ma}. In fact, very interesting recent work on quantum curves in this context has appeared in \cite{CGM,GHM, GKMR, KM, KMZ, KP, Ma2}.

\subsection{Topological recursion and wave-function}

Let us be a little more explicit on the connection between topological recursion and WKB. Let us start with the topological recursion. The starting point is a spectral curve. For the purpose of this paper, a spectral curve will mean a triple $(\Sigma, x,y)$ where $\Sigma$ is a Torelli marked compact Riemann surface and $x$ and $y$ are meromorphic functions on $\Sigma$, such that the zeroes of $dx$ do not coincide with the zeroes of $dy$. Then $x$ and $y$ must satisfy an irreducible polynomial equation
\begin{equation}
P(x,y) = 0.
\end{equation}
For most of the paper, we will restrict ourselves to the case where the affine curve defined by $\{(x,y)\, | \, P(x,y) = 0\} \subset \mathbb{C}^2$ is such that its Newton polygon has no interior point, and that it is smooth as an affine curve (its projectivization may not be smooth though) -- we call such curves \emph{admissible}. In particular, admissible curves all have genus zero, and the Torelli marking is irrelevant.

Out of this spectral data, the topological recursion produces an infinite tower of meromorphic differentials $W_{g,n}(z_1,\ldots,z_n)$ on $\Sigma^n$.
In \cite{EO}, it was proposed to construct a ``wave-function'' as
\begin{multline}\label{eq:wkb}
\psi(z) =\exp\biggl(\sum_{g=0}^\infty \sum_{n=1}^\infty \frac{\hbar^{2g+n-2}}{n!}\\
\int_a^z \cdots \int_a^z \Bigl(W_{g,n}(z_1,\ldots,z_n) - \delta_{g,0}\delta_{n,2}\, \frac{dx(z_1) dx(z_2)}{(x(z_1) - x(z_2))^2} \Bigr) \biggr),
\end{multline}
where $a \in \Sigma$ is a choice of base point for integration of the meromorphic differentials. In \cite{EO,BorE1} $a$ was chosen as a pole of $x$, and if $x$ is of some degree $d$, there can be $d$ choices for $a$.\footnote{In this paper, we will assume that $\Sigma$ has genus zero, so integration is unambiguously defined for all $W_{g,n}$ with $2g-2+n \geq 0$ since they are residueless. The integral of $W_{0,1}$ may need to be regularized, but this will play no role in the following.} In fact, we will generalize the definition of the wave-function slightly by allowing more general integration divisors.
In \cite{EO,BorE1,EM1} it was argued that when the spectral curve has genus $>0$, the definition \eqref{eq:wkb} should be completed with some appropriate theta functions, which are necessary for instance to match knot polynomials \cite{BoE}.

\subsection{Quantum curve}

Then the question is whether there exists a \emph{quantum curve}, that is, a quantization $\hat{P}(\hat{x}, \hat{y}; \hbar)$ of the spectral curve $P(x,y) = 0$ that kills the wave-function:
\begin{equation}\label{eq:ode}
\hat{P}(\hat{x}, \hat{y}; \hbar) \psi = 0.
\end{equation}
What do we mean by quantum curve? To quantize the spectral curve with defining equation $P(x,y) = 0$, we map the commutative variables $(x,y)$ to non-commutative operators
\begin{equation}
\hat{x} = x, \quad \hat{y} = \hbar \frac{d}{dx},
\end{equation}
satisfying the commutation relation $[ \hat{y},\hat{x} ] = \hbar$. This turns the polynomial $P(x,y)$ into a rank $d$ linear differential operator. Of course, the process is not unique, because of ordering ambiguities. In fact, ordering is a rather involved issue here; we could add any term of the form $p(x,y) (yx-xy)^n$, with $p(x,y)$ and arbitrary polynomial of $x$ and $y$, to the polynomial $P(x,y)$ without modifying it, since $yx=xy$. However, after quantization, these terms give rise to corrections to the differential operator of the form $p(\hat{x}, \hat{y}) \hbar^n$. Hence we need to be a bit more general in our definition of quantum curves.

\begin{defin}\label{d:QC}
A \emph{quantum curve} $\hat{P}$ of a spectral curve $C$ is a rank $d$ linear differential operator in $x$, such that, after normal ordering (that is bringing all the $\hat{x}$'s to the left of the $\hat{y}$'s), it takes the form
\begin{equation}
\hat{P}(\hat{x}, \hat{y}; \hbar) = P(\hat{x}, \hat{y}) + \sum_{n \geq 1} \hbar^n P_n(\hat{x}, \hat{y}),
\end{equation}
where the leading order term $P(\hat{x}, \hat{y})$ recovers the polynomial equation of the original spectral curve (normal ordered), and the $P_n(\hat{x}, \hat{y})$ are differential operators (normal ordered) in $x$ of rank at most $d-1$.

We say that a spectral curve is \emph{simple} if there is only a finite number of $\hbar$ corrections. That is, it is simple if there exists a positive integer $N$ such that $P_n(\hat{x}, \hat{y}) = 0$ for all $n>N$.
\end{defin}

As noted above, the quantization process $P \mto \hat{P}$ is certainly not unique. However, the reverse process $\hat{P} \mto P$ is unique. Given a quantum curve $\hat{P}(\hat{x}, \hat{y}; \hbar)$, it uniquely defines an irreducible polynomial equation $P(\hat{x}, \hat{y})$ at leading order in $\hbar$, hence an associated spectral curve.

With this definition under our belt, we can now ask whether there exists a quantum curve $\hat{P}(\hat{x}, \hat{y}; \hbar)$ that kills the wave-function $\psi$. In other words, the question is whether the asymptotic series in $\hbar$ given by \eqref{eq:wkb} reconstructs the WKB expansion of some ordinary differential equation $\hat{P}(\hat{x}, \hat{y}; \hbar) \psi = 0$, where $\hat{P}(\hat{x}, \hat{y}; \hbar)$ is a quantization of the original spectral curve according to the definition above. It is clear that $\int_a^z W_{0,1}$ is the leading order term of the WKB asymptotic solution for any quantum curve associated to a given spectral curve; that is because in the topological recursion formalism $W_{0,1} = y dx$, hence it is straightforward to show that
\begin{equation}
\left[ e^{-\frac{1}{\hbar} \int_{a}^{z} W_{0,1}} \hat{P} e^{\frac{1}{\hbar} \int_{a}^{z} W_{0,1}} \right]_{\hbar = 0} = 0
\end{equation}
for any quantization of the spectral curve. Thus the question is whether the higher order terms in the $\hbar$ series in \eqref{eq:wkb} provide the full WKB asymptotic solution to a quantum curve.

The motivation for asking this question comes from matrix models. The topological recursion was originally introduced to solve Hermitian matrix models. However, it now lives a life of its own, beyond matrix models; it can be applied to any spectral curve to compute a sequence of $W_{g,n}$, and a corresponding $\psi$. It is thus natural to ask whether the mathematical structures known to be present in the context of matrix models can also be generalized to the broader context of applicability of the topological recursion.

For instance, it is well known that the partition function of a Hermitian matrix model is a particular example of a tau-function \cite{Mehta}. It can then be argued that the appropriate Schlesinger transform constructs the wave-function \eqref{eq:wkb} (for genus zero spectral curves), and that this wave-function should be a Baker-Akhiezer function for some isomonodromic integrable system. This implies that it should satisfy some Sato and Hirota equations, and that it should be annihilated by a quantum curve in the sense above.

The question then is whether this web of interconnections remains valid in the broader context of the topological recursion. As explained in \cite{BorE1}, from the topological recursion there is a natural candidate for a tau-function. For genus zero spectral curves, the wave-function \eqref{eq:wkb} is the Schlesinger transform of this conjectural tau-function. It then follows that, conjecturally, it should be annihilated by a quantum curve. This claim was proved to order $O(\hbar^3)$ in \cite{BorE1} for arbitrary spectral curves of any genus (more precisely, as mentioned above, the wave-function \eqref{eq:wkb} is only appropriate for genus zero curves; for higher genus spectral curves it must be appropriately completed with theta functions \cite{BorE1,EM1}). Our aim is to study whether this claim is true at all orders in $\hbar$.

\subsection{Our main result}

The goal of this paper is to answer this question affirmatively for a large class of spectral curves. More precisely, we prove that there exists a quantum curve for all admissible spectral curves, that is, for all spectral curves whose Newton polygons have no interior point and that are smooth as affine curves. Moreover, these quantum curves are all simple.

The class of admissible spectral curves considered in this paper includes all of the genus zero quantum curves that have already been studied in the literature (to our knowledge), and many more. It includes many quantum curves of rank greater than two.

We also study the question of whether the quantization is unique. The answer to this question turns out to be very interesting; we find (as explained in section 2.3.1 of \cite{BorE1}) a very explicit dependence between the form of the quantum curve and the choice of integration divisor to reconstruct the asymptotic expansion \eqref{eq:wkb} from the meromorphic differentials produced by the topological recursion. Different choices of integration divisors, for the same spectral curve, give rise to different quantum curves that are all quantizations of the original spectral curve; they generally differ by some choice of ordering in the quantization. We study this explicitly in many examples.

While the class of spectral curves that we study in this paper is quite large, it would be interesting to investigate whether our proof can be generalized to even more spectral curves: for instance, genus zero curves whose Newton polygons have interior points, or spectral curves in $\mathbb{C} \times \mathbb{C}^*$ or $\mathbb{C}^* \times \mathbb{C}^*$, or higher genus spectral curves. We hope to report on this in the near future.

\subsection{Outline and strategy}

To prove the existence of quantum curves for such a large class of spectral curves requires quite a few steps. We start in Section 2 by reviewing the geometry of spectral curves and their corresponding Newton polygons, and we define what we mean by admissible spectral curves.
Then, in Section 3 we reformulate the topological recursion in a ``global way'', which involves summing over sheets instead of local deck transformations near the ramification points. Such a global formulation of the topological recursion was first introduced in \cite{BHLMR, BE:2012}. But here, we push the calculations further and reformulate it in a different way, which is, for our purpose, more useful. Our main result in this section is Theorem \ref{t:ref}, which provides a neat and simple formulation of the topological recursion.

We would like to remark here that for the topological recursion to reconstruct the WKB expansion in general, we need to evaluate residues in the topological recursion at \emph{all ramification points} of the $x$-projection, not only zeros of $dx$. For most spectral curves, ramification points that are not zeros of $dx$ (\ie poles of $x$ of order $2$ or more) do not contribute to the residues, but this is not true for all curves. This is an important point that had been missed in the previous literature on topological recursion.

The next step in the program is to evaluate the residues in the topological recursion of Theorem \ref{t:ref}, which we do in Section 4. To achieve this, we propose a detailed pole analysis of the integrand. Our main result is Theorem \ref{t:pole}, which gives an explicit expression for some objects $\sfrac{p_0(z) Q^{(m)}_{g,n+1}(z; \bmz)}{dx(z)^m}$.
This is in fact perhaps the most important theorem in the paper. In practice, what it does is reconstruct a sort of loop equation from the topological recursion, from which we will be able to reconstruct the quantum curves.

Finally in section 5 we reconstruct the quantum curves. We start from the expressions in Theorem \ref{t:pole}. We define a procedure to integrate them for arbitrary integration divisors on $\Sigma$, and use it to obtain a partial differential system. We sum over $g$ and $n$ with appropriate powers of $\hbar$, and then we ``principal specialize'', meaning that we set all variables to be equal in an appropriate way. The system then becomes a system of non-linear first-order ordinary differential equations. We finally use a ``Riccati trick'' to transform this system into a system of linear first-order differential equations for some objects that we call $\psi_k$, $k=1,\ldots, r$, where $r$ is the degree of $P(x,y)$ in $y$. Those $\psi_k$ are constructed out of the wave-function introduced in \eqref{eq:wkb}. The main result of this section is Theorem \ref{thm:eqpsiD}, which presents this system of linear first-order differential equations.
This method is a generalization for rank $\geq 2$ of the method introduced in \cite{BeE, BeE2,DuM,DuM2}.

In section 5.3 we study special choices of integration divisors. If, in \eqref{eq:wkb}, we integrate from $a$ to $z$, with $a$ a pole of the function $x$, then the system simplifies very nicely. In Lemma \ref{l:QC} we then show that it can be rewritten as an order $r$ ordinary differential equation for the wave-function $\psi$ of \eqref{eq:wkb}: the quantum curve!

In section 6 we study many examples explicitly. We reproduce all genus zero quantum curves obtained in the literature (to our knowledge), and construct many more, to show that our result is not only general but also concretely applicable. All the examples presented in section 6 have also been checked numerically to a few non--trivial orders in $\hbar$ in Mathematica.

In section 7 we study the case of the $r$-Airy curve, $y^r-x= 0$, in a bit more detail, focusing on its enumerative meaning. This section is somewhat independent from the rest of the paper. We explain how the meromorphic differentials constructed by the topological recursion for the $r$-Airy curve are generating functions for $r$-spin intersection numbers; however, we postpone an explicit proof from matrix models to a future publication \cite{BEf}. This result was first announced in \cite{BEAIM}, and has now also been proved using a different approach in \cite{DBNOPS}. We give explicit calculations from the topological recursion that reproduce known $r$-spin intersection numbers.

\subsubsection*{Acknowledgments}
We would like to thank N.~Do, O.~Dumitrescu, O.~Marchal, M.~Mari\~no, M.~Mulase, and N.~Orantin for interesting discussions. We would also like to thank the referees for insightful comments. We would like to thank the \emph{Centre de recherches mathématiques} at Université de Montréal, and the thematic 2012-2013 semester on ``Moduli Spaces and their Invariants in Mathematical Physics'' where this work was initiated, and the \emph{American Institute of Mathematics}, where parts of this work was completed.

\section{The geometry}

In this section we introduce the geometric context for the topological recursion.

\Subsection{Spectral curves}

\begin{defin}\label{d:sc}
A \emph{spectral curve} is a triple $(\Sigma,x,y)$ where $\Sigma$ is a Torelli marked genus $\hat{g}$ compact Riemann surface\footnote{A Torelli marked compact Riemann surface $\Sigma$ is a genus $\hat{g}$ Riemann surface $\Sigma$ with a choice of symplectic basis of cycles $(A_1, \ldots, A_{\hat{g}}, B_1, \ldots, B_{\hat{g}}) \in H_1(\Sigma,\mathbb{Z})$.} and $x$ and $y$ are meromorphic functions on $\Sigma$, such that the zeroes of $dx$ do not coincide with the zeroes of $dy$.
\end{defin}

\begin{rem}
The definition of spectral curves can (and for many applications must) be generalized, but this restricted definition is sufficient for the purpose of this paper.
\end{rem}

Since we assume that $x$ and $y$ are meromorphic functions on $\Sigma$, this means that they must satisfy an absolutely irreducible equation of the form
\begin{align}\label{eq:sc}
P(x,y) = p_0(x) y^r + p_1(x) y^{r-1} +\cdots + p_{r-1}(x) y + p_r(x) = \sum_{i=0}^r p_{r-i}(x) y^{i} =0,
\end{align}
where the $p_i(x)$ are polynomials of $x$. Therefore, we can also see our spectral curve as being given by an irreducible affine algebraic curve \eqref{eq:sc} in $\mathbb{C}^2$, which we will call~$\Sigma_0$: in this case $\Sigma$ is the normalization of $\Sigma_0$. We will call the \emph{punctures} the poles of $x(z)$ and $y(z)$.

We will be interested in the branched covering $\pi : \Sigma \to \mathbb{P}^1$ given by the meromorphic function $x$. This branched covering agrees with the projection $\pi_0: \Sigma_0 \to \mathbb{C}$ on the $x$-axis away from the singularities and the points over $x=\infty$. We denote by $R$ the set of ramification points of $\pi$. The ramification points of $\pi$ are either at zeros of~$dx$ or at poles of $x$ of order $\geq 2$.

\subsection{Newton polygons}

Let us rewrite the defining equation \eqref{eq:sc} of the spectral curve $\Sigma$ as
\begin{equation}
P(x,y) = \sum_{(i,j) \in A} \alpha_{i,j} x^i y^j = 0,
\end{equation}
where $A \subset \mathbb{N}^2$ is the set of pairs of indices $(i,j)$ such that $\alpha_{i,j} \neq 0$ for $(i,j) \in A$.

\begin{defin}
The \emph{Newton polygon} $\triangle$ of is the convex hull of the set $A$.
\end{defin}

An example of a Newton polygon is given in Figure \ref{f:newton1}.

\begin{figure}[htb]
\begin{center}
\includegraphics[width=0.07\textwidth]{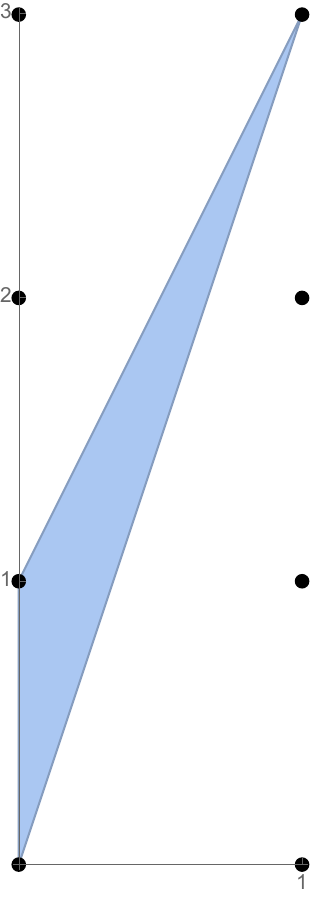}
\caption{The Newton polygon of the curve $x y^3 + y + 1 = 0$.}
\label{f:newton1}
\end{center}
\end{figure}

For $m=0,\ldots,r$, we define:
\begin{equation}\label{eq:sup}
\alpha_m = \inf\{a~|~(a,m) \in \triangle \}, \quad \beta_m = \sup\{a ~|~(a,m) \in \triangle \}.
\end{equation}
Clearly, the number of integral points of the Newton polygon $\triangle$ is then given by
\begin{equation}
\text{\# interior integral points of $\triangle$} = \sum_{i=1}^{r-1} \left(\lceil \beta_i \rceil - \lfloor \alpha_i \rfloor - 1 \right).
\end{equation}

An important result is Baker's formula \cite{Baker}:
\begin{thm}\label{t:baker}
The genus $\hat{g}$ of $\Sigma$ satisfies the inequality
\begin{equation}
\hat{g} \leq \sum_{i=1}^{r-1} \left(\lceil \beta_i \rceil - \lfloor \alpha_i \rfloor - 1 \right).
\end{equation}
The right-hand-side is of course equal to the number of interior points of the Newton polygon $\triangle$.
\end{thm}
See for instance \cite{Be} for a short proof of this result.

Another interesting result is the following. First, some notation. Given a meromorphic function $f$ on $\Sigma$, we denote by $\div(f)$ the divisor of $f$, by $\div_0(f)$ the divisor of zeros, and by $\div_\infty(f)$ the divisor of poles.

\begin{defin}\label{d:Pm}
For $m=2,\ldots,r$, we define the following meromorphic functions on $\Sigma$:
\begin{equation}\label{eq:Pm}
P_m(x,y) = \sum_{k=1}^{m-1} p_{m-1-k}(x)\,y^{k}.
\end{equation}
\end{defin}

Then, as shown in \cite{Be}, we get:
\begin{lem}[\cite{Be}]\label{l:ab}
For $m=2,\ldots,r$,
\begin{equation}
\div(P_m) \geq \alpha_{r-m+1} \div_0(x) - \beta_{r-m+1} \div_\infty(x).
\end{equation}
\end{lem}

This particular lemma will be very useful for us.

\Subsection{Admissible spectral curves}

\begin{defin}\label{d:admissible}
We say that a spectral curve is \emph{admissible} if:
\begin{enumerate}
\item
its Newton polygon $\triangle$ has no interior point;
\item if the origin $(x,y) = (0,0) \in \mathbb{C}^2$ is on the curve $\{P(x,y)=0 \} \subset \mathbb{C}^2$, then the curve is smooth at this point.
\end{enumerate}
The first condition is equivalent to
\begin{equation}
\lceil \beta_i \rceil - \lfloor \alpha_i \rfloor = 1, \quad i=1,\ldots,r-1.
\end{equation}
From Baker's formula, Theorem \ref{t:baker}, it follows that admissible spectral curves have genus zero.
\end{defin}

\begin{rem}
It follows that an admissible curve must be smooth as an affine curve. In fact, its projectivization can only have singularities at $(1{:}0{:}0)$ and $(0{:}1{:}0)$. This is because, as explained in \cite{BP}, the genus of a curve is exactly equal to the number of interior points if and only if the singularities of its projectivization are all among $(0{:}0{:}1)$, $(0{:}1{:}0)$ and $(1{:}0{:}0)$, and a certain non-degeneracy condition is satisfied. Since we also impose that the curve is smooth at the origin, it follows that it cannot have singularities anywhere else in $\mathbb{C}^2$.
\end{rem}

\begin{exa}
As an example of admissible curves, we note that all curves that are linear in $x$, \ie of the form $P(x,y) = A(y)+xB(y)=0$, with $A(y)$ and $B(y)$ polynomials in $y$, are admissible. Indeed, it is easy to see that for any such curve, the Newton polygon has no interior point, and all curves of that form are smooth as affine curves. Therefore they are admissible.
\end{exa}

We will see many interesting examples of curves of that form in Section \ref{s:examples}. But admissible curves certainly do not have to be linear in $x$; we will also study many examples that do not fall into this class.\footnote{We note that admissible curves are not too difficult to classify. They are either:
\begin{enumerate}
\item linear in $x$;
\item with Newton polygon given by the convex hull of $\{(0,0), (2,0), (0,2)\}$;
\item such that they can be obtained from the previous cases by a transformation of the form
\begin{equation}
(x,y) \mto (x^a y^b, x^c y^d), \quad \text{with} \quad ad-bc = 1,
\end{equation}
combined with overall rescaling by powers of $x$ and $y$ to get an irreducible polynomial equation.
\end{enumerate}
}
\subsection{More definitions}

Let us now go back to general spectral curves.
We introduce the following notation.

\begin{defin}\label{d:tauz}
We introduce
\begin{equation}
\tau(z) = \pi^{-1}(\pi(z)), \quad \tau'(z) = \tau(z) \setminus \{z\}.
\end{equation}
$\tau(z): \Sigma \to \Sym^r(\Sigma)$ is an analytic map that takes a point $p \in \Sigma$ to the set of preimages (with multiplicity) of the inverse image of its projection with respect to the branched covering $\pi : \Sigma \to \mathbb{P}^1$. Similarly, $\tau'(z): \Sigma \to \Sym^{r-1}(\Sigma)$ is also an analytic map that takes a point $p \in \Sigma$ to the set of preimages of the inverse image of its projection, minus the original point itself (with multiplicity 1).
\end{defin}

We now introduce two objects that are canonically defined on a compact Riemann surface $\Sigma$ with a symplectic basis of cycles for $H_1(\Sigma,\mathbb{Z}$).

\begin{defin}\label{d:w}
Let $a,b \in \Sigma$. The \emph{canonical differential of the third kind} $\omega^{a-b}(z)$ is a meromorphic one-form on $\Sigma$ such that:
\begin{itemize}
\item it is holomorphic away from $z=a$ and $z=b$;
\item it has a simple pole at $z=a$ with residue $+1$;
\item it has a simple pole at $z=b$ with residue $-1$;
\item it is normalized on $A$-cycles:
\begin{equation}
\oint_{z \in A_i} \omega^{a-b}(z) = 0, \quad \text{for $i=1,\ldots,\hat{g}$.}
\end{equation}
\end{itemize}
\end{defin}

\begin{defin}\label{d:B}
The \emph{canonical bilinear differential of the second kind} $B(z_1,z_2)$ is the unique bilinear differential on $\Sigma^2$ satisfying the conditions:
\begin{itemize}
\item it is symmetric, $B(z_1, z_2) = B(z_2,z_1)$;
\item it has its only pole, which is double, along the diagonal $z_1 = z_2$, with leading order term (in any local coordinate $z$)
\begin{equation}
B(z_1,z_2) \underset{z_1 \to z_2}{\to} \frac{d z_1 d z_2}{(z_1-z_2)^2} + \cdots;
\end{equation}
\item it is normalized on $A$-cycles:
\begin{equation}
\oint_{z_1 \in A_i} B(z_1, z_2) = 0, \quad \text{for $i=1,\ldots,\hat{g}$}.
\end{equation}
\end{itemize}
\end{defin}

\begin{rem}
It follows from the definition that
\begin{equation}
B(z_1, z_2) = d_1 \omega^{z_1-b}(z_2).
\end{equation}
Equivalently,
\begin{equation}
\omega^{a - b}(z) = \int_{z_1=b}^a B(z_1, z),
\end{equation}
where the integral is taken over the unique homology chain with boundary $[a]-[b]$ that doesn't intersect the homology basis.
\end{rem}

When $\Sigma$ has genus $0$, both objects have very explicit expressions:
\begin{align}
\omega^{a-b}(z) &= d z \Bigl(\frac{1}{z-a} - \frac{1}{z-b} \Bigr),\\
B(z_1,z_2) &= \frac{d z_1 d z_2}{(z_1-z_2)^2}.
\end{align}

\section{Topological recursion}

Let us now introduce the topological recursion formalism, which was first proposed in \cite{E1,CEO,EO,EO2}.

\subsection{Definition of the topological recursion}

Let $(\Sigma,x,y)$ be a spectral curve. The topological recursion constructs an infinite tower of symmetric meromorphic differentials $W_{g,n}(z_1,\ldots,z_n)$ on $\Sigma^n$.
To this end, we need to define the recursive structure that appears in the topological recursion.

Let us first introduce some notation:
\begin{defin}
Let $A \subseteq_k B$ if $A \subseteq B$ and $|A| = k$.
\end{defin}

\begin{defin}
Let $B$ be a set and $J_i$'s be subsets of $B$. The symbol $\uplus$ means disjoint union, \ie the notation
$J_1 \uplus J_2 \uplus \dots \uplus J_k=B$ means that the $J_i$s are all pairwise disjoint and their union is $B$.
\end{defin}

\begin{defin}
Let $\mathcal{S}(\bmt)$ be the set of set partitions of an ensemble $\bmt$.
\end{defin}

We now define the recursive structure:

\begin{defin}\label{d:RW}
Let $\{W_{g,n+1} \}$ be an arbitrary collection of symmetric meromorphic differentials on $\Sigma^n$, with $g \geq 0$, $n \geq 0$. Let $k \geq 1$. Define $\bmt= \{t_1, \ldots, t_k\}$ and $\bmz = \{z_1, \ldots, z_n\}$, where the $t_i$'s and $z_i$'s are copies of the coordinate on the Riemann surface $\Sigma$.

Then we define:
\begin{equation}\label{e:RW}
\mathcal{R}^{(k)}W_{g,n+1}(\bmt; \bmz) = \sum_{\mu \in \mathcal{S}(\bmt)} \sum_{\uplus_{i=1}^{\ell(\mu)} J_i = \bmz} \sum'_{\sum_{i=1}^{\ell(\mu)}g_i = g + \ell(\mu) - k} \biggl(\prod_{i=1}^{\ell(\mu)} W_{g_i, |\mu_i| + |J_i|} (\mu_i, J_i) \biggr).
\end{equation}
The first summation is over set partitions of $\bmt$; $\ell(\mu)$ is the number of subsets in the set partition $\mu$. The third summation is over all $\ell(\mu)$-tuple of non-negative integers $(g_1, \ldots, g_{\ell(\mu)})$ such that $g_1 +\cdots + g_{\ell(\mu)} = g + \ell(\mu) - k $. The prime over the third summation indicates that we exclude all terms that include contributions from $W_{0,1}$; more precisely, we exclude the cases with $(g_i, |\mu_i| + |J_i|) = (0,1)$ for some $i$. We also define
\begin{equation}
\mathcal{R}^{(0)} W_{g,n+1}(\bmz) = \delta_{g,0} \delta_{n,0},
\end{equation}
where $\delta_{i,j}$ is the Kronecker delta symbol.
\end{defin}

\begin{rem}
This recursive structure appeared in \cite{BHLMR, BE:2012}. It can be understood pictorially as encoding all possible ways of removing a sphere with $k+1$ marked boundaries from a genus $g$ Riemann surface with $n+1$ marked boundaries. We refer the reader to \cite{BHLMR, BE:2012} for a discussion of the geometric interpretation of the recursive structure in terms of degenerations of Riemann surfaces.
\end{rem}

\begin{rem}
Note that $\mathcal{R}^{(k)}W_{g,n+1}(\bmt; \bmz)$ is symmetric under internal permutations of both the $t$-variables and the $z$-variables.
\end{rem}

\begin{exa}\label{ex:R}
Given the ubiquity of the recursive structure defined in Definition~\ref{d:RW}, let us give a few examples to make it more explicit.

We first consider the case $k=1$. For $g=0$ and $n=0$, we simply get
\begin{equation}
\mathcal{R}^{(1)} W_{0,1}(t; \emptyset) = 0
\end{equation}
because of the prime in the summation on the right-hand-side of \eqref{e:RW}. For $g \geq 0$, $n \geq 0$ and $(g,n) \neq (0,0)$, we get
\begin{equation}
\mathcal{R}^{(1)} W_{g,n+1} (t; \bmz) = W_{g,n+1}(t, \bmz).
\end{equation}

Consider now the case $k=2$. For $g \geq 0$, $n \geq 0$, we get that
\begin{multline}
\mathcal{R}^{(2)} W_{g,n+1}(t_1, t_2; \bmz)\\[-5pt]
= W_{g-1,n+2}(t_1, t_2, \bmz) + \sum'_{\substack{J_1 \uplus J_2 = \bmz\\ g_1 + g_2 = g}} W_{g_1, |J_1|+1}(t_1, J_1) W_{g_2, |J_2|+1}(t_2, J_2),
\end{multline}
where it is understood that $W_{g,n}$'s with negative $g$ vanish. This is the original recursive structure considered by Eynard and Orantin.
\end{exa}

With this under our belt, we are ready to define the topological recursion.

\begin{defin}\label{d:TR}
Let $(\Sigma,x,y)$ be a spectral curve, with $\pi:\Sigma \to \mathbb{P}^1$ a degree~$r$ branched covering given by the meromorphic function $x$, and $R \subset \Sigma$ the set of ramification points of $\pi$.

We first define
\begin{equation}
W_{0,1}(z) = y(z) dx(z), \quad W_{0,2}(z_1,z_2) = B(z_1,z_2),
\end{equation}
with $B(z_1,z_2)$ defined in Definition \ref{d:B}.

Let $\bmz = \{z_1,\ldots,z_n\}\in \Sigma^n$. Recall the set $\tau'(z)$ defined in Definition \ref{d:tauz}.
For $n \geq 0$, $g \geq 0$ and $2g-2+n \geq 0$, we uniquely construct symmetric meromorphic differentials $W_{g,n}$ on $\Sigma^n$ with poles along $R$ via the \emph{topological recursion}:
\begin{multline}\label{eq:TR}
W_{g,n+1}(z_0, \bmz)\\[-5pt]
= \sum_{a \in R} \underset{z = a}{\Res} \biggl(\sum_{k=1}^{r-1} \sum_{\beta(z) \subseteq_k \tau'(z)} (-1)^{k+1} \frac{\omega^{z-\alpha}(z_0)}{E^{(k)}(z; \beta(z))}\, \mathcal{R}^{(k+1)} W_{g,n+1}(z,\beta(z); \bmz) \biggr),
\end{multline}
where
\begin{equation}
E^{(k)}(z; t_1, \ldots, t_k) = \prod_{i=1}^k (W_{0,1}(z) - W_{0,1}(t_i)),
\end{equation}
with the $t_i$'s copies of the coordinate on the Riemann surface $\Sigma$.
The first summation in \eqref{eq:TR} is over all ramification points in $R$, and the second and third summations together mean that we are summing over all subsets of $\tau'(z)$. $\alpha$ is an arbitrary base point on $\Sigma$ which is not in $R$.
\end{defin}

\begin{rem}
Note that this recursion was called ``global topological recursion'' in~\cite{BE:2012}; from now on we will simply refer to it as ``topological recursion''. It was shown in \cite{BE:2012} that it is indeed equivalent to the usual local formulation of the topological recursion \cite{EO,EO2} when the ramification points are all simple. But here we do not need to assume simplicity of the ramification points.
\end{rem}

\begin{rem}
It is important to note here that it is not clear \emph{a priori} that Definition~\ref{d:TR} even makes sense. Indeed, for the recursive structure introduced in Definition~\ref{d:RW} to be well defined, the differentials $W_{g,n}$ must be symmetric. Hence for the topological recursion proposed in \eqref{eq:TR} to make sense, we must show by induction that the $W_{g,n}$ thus constructed are indeed symmetric. This was proven in \cite{EO} (Theorem~4.6) for the original topological recursion, and it was shown in \cite{BE:2012} (see Section 4) that symmetry also holds for the $W_{g,n}$ constructed from the global topological recursion presented above. In fact, one could also formulate a proof of symmetry directly from the global topological recursion above along the same lines as the proof in \cite{EO}.

It is also important that \eqref{eq:TR} is independent of the choice of base point $\alpha$. This is easy to see by induction. Let $W_{g,n}^{(\alpha)}$ be the differentials constructed with base point~$\alpha$, and $W_{g,n}^{(\gamma)}$ the differentials constructed with base point $\gamma \neq \alpha$. For both cases the initial conditions of the recursion ($W_{0,1}$ and $W_{0,2}$) are the same. Now assume that $W_{g',n'}^{(\alpha)} = W_{g',n'}^{(\gamma)}$ for all $g',n'$ such that $2g'-2+n' < 2g-2+n$.
Then
\begin{multline}
W_{g,n+1}^{(\alpha)} (z_0, \bmz)- W_{g,n+1}^{(\gamma)} (z_0, \bmz)\\
=\omega^{\gamma-\alpha}(z_0) \sum_{a \in R} \underset{z = a}{\Res} \biggl(\sum_{k=1}^{r-1} \sum_{\beta(z) \subseteq_k \tau'(z)}\frac{(-1)^{k+1}}{E^{(k)}(z; \beta(z))}\, \mathcal{R}^{(k+1)} W_{g,n+1}(z,\beta(z); \bmz) \biggr).\label{eq:proofind}
\end{multline}
The right-hand-side has simple poles at $z_0 = \gamma$ and $z_0 = \alpha$, which is a contradiction, since the left-hand-side can only have poles at $z_0 = a$ for $a \in R$ (see below). Thus both the left-hand-side and the right-hand-side must be zero, and we conclude that $W_{g,n+1}^{(\alpha)}= W_{g,n+1}^{(\gamma)} $, that is, \eqref{eq:TR} is independent of the choice of base point $\alpha$.

The $W_{g,n}$ also satisfy various other properties. For instance, it can be shown that the $W_{g,n}$ only have poles along $R$, with no residues, and that they are normalized over $A$-cycles:
\begin{equation}
\oint_{z_0 \in A_k} W_{g,n+1}(z_0, \bmz) = 0, \quad k=1, \ldots, \hat{g}.
\end{equation}
Again, this was proved in \cite{EO} (see Theorems 4.2 and 4.3) for the original topological recursion, and shown to hold for the global version in \cite{BE:2012} (Section 4). Other properties of the $W_{g,n}$ are also discussed there.
\end{rem}

\begin{rem}
In the standard formulation of the topological recursion \cite{EO,EO2,BE:2012}, the $W_{g,n}$ are constructed by summing only over residues at the zeros of $dx$, instead of all ramification points of the branched covering $\pi: \Sigma \to \mathbb{P}^1$ given by the meromorphic function $x$. In other words, the poles of order $\geq 2$ of $x$ are generally not included in the sum. In most cases, this does not matter, since these poles would yield zero residues hence would not change the $W_{g,n}$. However, for some curves, they do contribute, and in fact they are necessary in order to obtain the quantum curve later on. More precisely, the fundamental result established below in Lemma \ref{l:PP} only holds when all ramification points are included in $R$ --- otherwise the meromorphic one-forms studied in Lemma \ref{l:PP} can have extra poles at the poles of $x$ of order $\geq 2$. Therefore, \emph{we must sum over residues at all ramification points in $R$ to construct the $W_{g,n}$, not only the zeros of $dx$}.
\end{rem}

\subsection{Rewriting the topological recursion}

In this subsection we rewrite the topological recursion in a different, and for our purposes nicer, way. First, we need to introduce a few more objects.

\begin{defin}
Using the same notation as in Definition \ref{d:RW}, for $g,n \geq 0$ and $k \geq 1$ we define
\begin{equation}\label{eq:EW}
\mathcal{E}^{(k)}W_{g,n+1} (\bmt;\bmz) = \!\!\sum_{\mu \in \mathcal{S}(\bmt)} \sum_{\uplus_{i=1}^{\ell(\mu)} J_i = \bmz} \sum_{\sum_{i=1}^{\ell(\mu)}g_i = g + \ell(\mu) - k} \!\!\biggl(\prod_{i=1}^{\ell(\mu)} W_{g_i, |\mu_i| + |J_i|} (\mu_i, J_i) \biggr).
\end{equation}
The main difference with Definition \ref{d:RW} is that we have removed the prime in the third summation. Therefore, the summation includes contributions from $W_{0,1}$. We also define
\begin{equation}
\mathcal{E}^{(0)} W_{g,n+1}(\bmz) = \delta_{g,0} \delta_{n,0}.
\end{equation}
\end{defin}

\begin{rem}
Similarly to $\mathcal{R}^{(k)}W_{g,n+1}(\bmt; \bmz)$, it can be understood pictorially as encoding all possible ways of removing a sphere with $k+1$ marked boundaries from a genus $g$ Riemann surface with $n+1$ marked boundaries, but for $\mathcal{E}^{(k)}W_{g,n+1}(\bmt; \bmz)$ we may also cut out discs.
\end{rem}

\begin{rem}
Note that $\mathcal{E}^{(k)}W_{g,n+1}(\bmt; \bmz)$ is also symmetric under internal permutations of both the $t$-variables and the $z$-variables.
\end{rem}

\begin{exa}
Let us give a few examples of this structure. First, for $k=1$, $g \geq 0$, $n \geq 0$ we get as in Example \ref{ex:R}:
\begin{equation}
\mathcal{E}^{(1)} W_{g,n+1} (t; \bmz) = W_{g,n+1}(t, \bmz),
\end{equation}
but now this includes the case $(g,n) = (0,0)$ as well since there is no prime in the summation on the right-hand-side of \eqref{eq:EW}. Similarly, for $k=2$, $g \geq 0$, $n \geq 0$, we get
\begin{multline}
\mathcal{E}^{(2)} W_{g,n+1}(t_1, t_2; \bmz)\\[-3pt]
= W_{g-1,n+2}(t_1, t_2, \bmz) + \sum_{\substack{J_1 \uplus J_2 = \bmz\\ g_1 + g_2 = g}} W_{g_1, |J_1|+1}(t_1, J_1) W_{g_2, |J_2|+1}(t_2, J_2),
\end{multline}
without the prime in the summation.

Another special case of interest is for $g=0$, $n=0$ and arbitrary $k \geq 1$. Then we get
\begin{equation}
\mathcal{E}^{(k)} W_{0,1}(\bmt; \emptyset) = \prod_{i=1}^k W_{0,1}(t_i).
\end{equation}
Similarly, for $g=0$, $n=1$ and arbitrary $k \geq 1$, we get
\begin{equation}
\mathcal{E}^{(k)} W_{0,2}(\bmt; z) = \sum_{j=1}^k \biggl(W_{0,2}(t_j, z) \prod_{\substack{i=1\\ i \neq j}}^k W_{0,1}(t_i) \biggr).
\end{equation}
\end{exa}

Let us now study a few properties of the $\mathcal{E}^{(k)}W_{g,n+1} (\bmt;\bmz)$ that will be useful later on.
We first recall Lemma 1 of \cite{BE:2012}, whose proof was purely combinatorial (note that $\mathcal{R}^{(k)} W_{g,n+1}$ was denoted by $\mathcal{W}^g_{k,n}$ there):

\begin{lem}[{\cite[Lem.\,1]{BE:2012}}]\label{l:curly1} For all $g,n \geq0$ and $k \geq 1$,
\begin{multline}
\mathcal{R}^{(k)} W_{g,n+1}(\bmt; \bmz) = \mathcal{R}^{(k-1)} W_{g-1,n+2} (\bmt \backslash \{t_k\}; \bmz, t_k)\\[-3pt]
+\sum_{J_1 \uplus J_2 = \bmz} \sum'_{g_1 + g_2 = g} \bigl(\mathcal{R}^{(k-1)} W_{g_1, |J_1|+1}(\bmt \backslash \{t_k\}; J_1) \bigr)W_{g_2,|J_2|+1}(t_k, J_2).
\end{multline}
The prime over the summation means that we do not include the case $(g_2, J_2) = (0, \emptyset)$.
\end{lem}

Similarly, we can prove the following lemma:

\begin{lem}\label{l:Ek} For all $g,n,k \geq0$,
\begin{multline}
\mathcal{E}^{(k)} W_{g,n+1}(\bmt; \bmz) = \mathcal{E}^{(k-1)} W_{g-1,n+2} (\bmt \backslash \{t_k\}; \bmz, t_k)\\+\sum_{J_1 \uplus J_2 = \bmz} \sum_{g_1 + g_2 = g} \bigl(\mathcal{E}^{(k-1)} W_{g_1, |J_1|+1}(\bmt \backslash \{t_k\}; J_1) \bigr) W_{g_2,|J_2|+1}(t_k, J_2).
\end{multline}
Note that unlike in Lemma \ref{l:curly1}, the summation is unprimed, that is, it includes the case $(g_2, J_2) = (0,\emptyset)$.
\end{lem}

\begin{proof}
The proof is exactly the same as for Lemma 1 in \cite{BE:2012}. By definition, the LHS is given by
\begin{equation}
\mathcal{E}^{(k)} W_{g,n+1}(\bmt; \bmz) = \sum_{\mu \in \mathcal{S}(\bmt)} \sum_{\uplus_{i=1}^{\ell(\mu)} J_i =\!\! \bmz} \sum_{\sum_{i=1}^{\ell(\mu)}g_i = g + \ell(\mu) - k} \!\!\biggl(\prod_{i=1}^{\ell(\mu)} W_{g_i, |\mu_i| + |J_i|} (\mu_i, J_i) \biggr).
\end{equation}
The first term on the RHS takes care of all partitions that do not include a subset of cardinality one of $\bmt$ containing only $t_k$. These are then taken care by the second term on the RHS. The summation in the second term on the RHS is unprimed, because terms with $W_{0,1}(t_k)$ (that is $(g_2,J_2) = (0,\emptyset)$) are included in the definition of $\mathcal{E}^{(k)} W_{g,n+1}(\bmt; \bmz) $.
\end{proof}

The relation between $\mathcal{R}^{(k)}W_{g,n+1}(\bmt; \bmz)$ and $\mathcal{E}^{(k)}W_{g,n+1}(\bmt; \bmz)$ can be expressed explicitly. We obtain
\begin{lem}\label{l:RE} For all $g,n,k \geq 0$,
\begin{equation}
\begin{aligned}
&\hspace*{-1.5cm}\mathcal{E}^{(k)}W_{g,n+1} (\bmt;\bmz)\\[-5pt]
&= \mathcal{R}^{(k)}W_{g,n+1}(\bmt; \bmz) +
\sum_{i=1}^{k} \sum_{\beta \subseteq_i \bmt} \mathcal{E}^{(i)} W_{0,1}(\beta) \mathcal{R}^{(k-i)}W_{g,n+1}(\bmt \setminus \beta; \bmz)\\[-5pt]
&= \sum_{i=0}^{k} \sum_{\beta \subseteq_i \bmt} \mathcal{E}^{(i)} W_{0,1}(\beta) \mathcal{R}^{(k-i)}W_{g,n+1}(\bmt \setminus \beta; \bmz),
\end{aligned}
\end{equation}
where in the second line we used the fact that $\mathcal{E}^{(0)} W_{0,1} = 1$.
\end{lem}
\begin{proof}
This is just a combinatorial rewriting, taking the $W_{0,1}(t)$ contributions out:
\begin{multline}
\mathcal{E}^{(k)}W_{g,n+1} (\bmt;\bmz)\\
= \mathcal{R}^{(k)}W_{g,n+1}(\bmt; \bmz) + \sum_{i=1}^k W_{0,1}(t_i) \mathcal{R}^{(k-1)} W_{g,n+1} (\bmt \backslash \{t_i\}; \bmz)\\
\shoveright{+ \sum_{1=i_1< i_2 \leq k} W_{0,1}(t_{i_1}) W_{0,1}(t_{i_2}) \mathcal{R}^{(k-2)} W_{g,n+1} (\bmt \backslash \{t_{i_1}, t_{i_2}\}; \bmz) +\cdots}\\
\shoveright{+\sum_{1=i_1< i_2< \ldots < i_{k-2} \leq k} W_{0,1}(t_{i_1}) \cdots W_{0,1}(t_{i_{k-2}}) \mathcal{R}^{(2)} W_{g,n+1} (\bmt \backslash \{t_{i_1}, \ldots, t_{i_{k-2}}\}; \bmz)}\\
+ \sum_{1=i_1< i_2< \ldots < i_{k-1} \leq k} W_{0,1}(t_{i_1}) \cdots W_{0,1}(t_{i_{k-1}}) W_{g,n+1} (\bmt \backslash \{t_{i_1}, \ldots, t_{i_{k-1}}\}, \bmz).
\end{multline}
Noting that $\mathcal{E}^{(i)} W_{0,1}(t_1, \ldots, t_i) = W_{0,1}(t_1) \cdots W_{0,1}(t_i)$, and that $\mathcal{R}^{(1)} W_{g,n+1} (t; \bmz) = W_{g,n+1}(t, \bmz)$, we get the statement of the lemma.
\end{proof}

We also define:

\begin{defin}\label{d:Pg}
Using the notation of Definition \ref{d:tauz}, we define, for $g,n,k \geq 0$:
\begin{equation}\label{eq:Pk}
Q^{(k)}_{g,n+1} (z; \bmz) = \sum_{\beta(z) \subseteq_k \tau(z)} \mathcal{E}^{(k)} W_{g,n+1}(\beta(z); \bmz).
\end{equation}
\end{defin}

\begin{rem}
Since \eqref{eq:Pk} is invariant under permutations of the preimages in~$\tau(z)$, $Q^{(k)}_{g,n+1} (z; \bmz) $ is in fact the pullback (in the variable $z$) of a globally defined meromorphic $k$-differential on the base of the branched covering $\pi: \Sigma \to \mathbb{P}^1$. In other words, we can write
\begin{equation}
Q^{(k)}_{g,n+1} (z; \bmz) = f(x(z), z_1, \ldots, z_n) dx(z)^k dz_1 \cdots dz_2,
\end{equation}
for some meromorphic function $f(x, z_1, \ldots, z_n)$.
\end{rem}

\begin{rem}
Note that for any spectral curve of degree $r$, and for all $(g,n)$,
\begin{equation}
Q^{(k)}_{g,n+1}(z; \bmz) = 0, \quad \text{for all $k > r$.}
\end{equation}
This is because we are summing over subsets of $\tau(z)$, hence there is no subset of $\tau(z)$ of cardinality $k > r$ since $\tau(z)$ has $r$ elements.
\end{rem}

\begin{rem}
As for the $k=0$ case, note that from the definition we obtain:
\begin{equation}
Q^{(0)}_{g,n+1}(\bmz) = \delta_{g,0} \delta_{n,0},
\end{equation}
\end{rem}

\begin{exa}\label{e:Q00}
The $Q^{(k)}_{0,1}(z)$ are particularly simple. From the definition, we have
\begin{equation}
Q^{(k)}_{0,1}(z) = \sum_{\beta(z) \subseteq_k \tau(z)} \prod_{i=1}^k W_{0,1}(\beta_{i}(z)),
\end{equation}
where the product on the RHS is over all elements of a given subset $\beta(z)$. It then follows that
\begin{equation}
Q^{(k)}_{0,1}(z) = (-1)^k \,\frac{p_k(z)}{p_0(z)} \,dx(z)^k,
\end{equation}
where the $p_k$ are the coefficients in the defining equation for the spectral curve in Definition \ref{d:sc}. Here and henceforth, we will abuse notation slightly and write $p_i(z) := p_i(x(z))$ to unclutter equations.
\end{exa}

\begin{exa}\label{ex:Q1}
The $Q^{(1)}_{g,n+1}$ are also easy to understand. From the definition, they are simply the pullback of the pushforward of the correlation functions $W_{g,n+1}$ with respect to the branched covering $\pi: \Sigma \to \mathbb{P}^1$ (in the $z$-variable):
\begin{equation}
Q^{(1)}_{g,n+1}(z;\bmz) = \pi^* \pi_* W_{g,n+1} (z, \bmz) = \sum_{i=1}^r W_{g,n+1}(\tau_i(z), \bmz).
\end{equation}
where we wrote $\tau(z)=\{\tau_1(z),\dots,\tau_r(z)\}$ (the labeling doesn't matter).
Indeed, the pushforward $\pi_*$ means that we are summing over all preimages in $\tau(z)$ to get a well defined meromorphic differential on the base $\mathbb{P}^1$, and then we pull it back to $\Sigma$.
\end{exa}

It then follows from Theorem 4.4 of \cite{EO} and Example \ref{e:Q00} above that:
\begin{lem}[\cite{EO}]\label{l:P1}
For $2g-2+n\geq 0$,
\begin{equation}
Q^{(1)}_{g,n+1}(z; \bmz) = 0.
\end{equation}
For the unstable cases, we have
\begin{equation}
Q^{(1)}_{0,1}(z) = - \frac{p_1(z)}{p_0(z)}\, dx(z),
\end{equation}
where $p_1(x)$ is the coefficient of $y^{r-1}$ in the defining equation \eqref{eq:sc} for the spectral curve, and
\begin{equation}
Q^{(1)}_{0,2}(z; z_1) = \pi^* B(x, x_1) = \frac{dx(z) dx (z_1)}{(x(z) - x(z_1))^2},
\end{equation}
by which we mean that we are pulling back to $\Sigma$ (in both variables) the canonical bilinear differential $B(x,x_1)$ on the base $\mathbb{P}^1$ of the branched covering $\pi : \Sigma \to \mathbb{P}^1$, which has the expression above since $\mathbb{P}^1$ has genus $0$.
\end{lem}

\begin{proof}
The $(g,n) = (0,0)$ case follows directly from the definition of spectral curves, see Example \ref{e:Q00} above. The $(g,n) = (0,1)$ case is a well known property of the canonical bilinear differential (see for instance \cite[Eq.\,(A-1)]{EO}):
\begin{equation}\label{eq:bsh}
Q^{(1)}_{0,2}(z; z_1) = \sum_{i=1}^r B(\tau_i(z), z_1) = \frac{dx(z) dx (z_1)}{(x(z) - x(z_1))^2}.
\end{equation}
As for $2g-2+n \geq 0$, the statement was proven in Theorem 4.4 of \cite{EO}. To be precise, there the statement was proven for the original topological recursion, but it is easy to see, following arguments similar to those in Section 4 of \cite{BE:2012}, that it also holds for the global topological recursion.
\end{proof}

We are now ready to rewrite the topological recursion in a different way.

\begin{thm}\label{t:ref}
The topological recursion in Definition \ref{d:TR} is equivalent to the following equation (for $2g-2+n \geq 0$) :
\begin{equation}
0 = \sum_{a \in R} \underset{z = a}{\Res} \left(\omega^{z-\alpha}(z_0) Q_{g,n+1}(z; \bmz) \right),
\end{equation}
where the $1$-form $Q_{g,n+1}(z; \bmz)$ is defined in terms of the $Q^{(k)}_{g,n+1}(z; \bmz)$ from Definition~\ref{d:Pg} as follows:
\begin{equation}
Q_{g,n+1}(z; \bmz) := \frac{dx(z)}{\sfrac{\partial P}{\partial y}(z)} \biggl(p_0(z) \sum_{k=1}^r (-1)^k y(z)^{r-k}\, \frac{Q^{(k)}_{g,n+1}(z; \bmz)}{dx(z)^k}\biggr),
\end{equation}
with $P(x,y)=0$ the equation of the spectral curve introduced in \eqref{eq:sc}.

\end{thm}

\begin{rem}
Before we prove the theorem, let us remark that the beauty of this formulation of the topological recursion is the following.
First notice that $Q_{g,n+1}(z; \bmz)$ is a meromorphic $1$-form of $z\in\Sigma$, with poles possibly on $R$, at coinciding points, and possibly at the poles of $W_{0,1}$, and/or at the zeros of $\partial P/\partial y$ (which are $R$ for a smooth affine curve).

If $\Sigma$ is genus zero, $\omega^{z-\alpha}(z_0)$ is also meromorphic and the integrand is a globally defined meromorphic differential on $\Sigma$ in $z$. If $\Sigma$ is higher genus, it is not quite globally defined (since $\omega^{z-\alpha}(z_0)$ is not a well defined meromorphic function of $z$), but it is a meromorphic differential on the fundamental domain. Therefore, in both cases we can replace the sum over residues by a single contour integral surrounding all ramification points in $R$. This point of view is quite powerful to study various global properties of the topological recursion.
\end{rem}

\begin{proof}
We start with the topological recursion in Definition \ref{d:TR}. For $2g-2+n \geq 0$,
\begin{multline}
W_{g,n+1}(z_0, \bmz)\\
= \sum_{a \in R} \underset{z = a}{\Res} \biggl(\sum_{k=1}^{r-1} \sum_{\beta(z) \subseteq_k \tau'(z)}(-1)^{k+1} \frac{\omega^{z-\alpha}(z_0)}{E^{(k)}(z; \beta(z))} \mathcal{R}^{(k+1)} W_{g,n+1}(z,\beta(z); \bmz) \biggr).
\end{multline}
We put all terms on a common denominator to get
\begin{multline}
W_{g,n+1}(z_0, \bmz) = \sum_{a \in R} \underset{z = a}{\Res} \biggl(\frac{p_0(z)\omega^{z-\alpha}(z_0)}{\sfrac{\partial P}{\partial y}(z)\, dx(z)^{r-1}}\\[-5pt]
\times \sum_{\rho(z) \uplus \beta(z) = \tau'(z)} (-1)^{|\beta|+1} E^{(|\rho |)}(z; \rho(z)) \mathcal{R}^{(|\beta|+1)} W_{g,n+1}(z,\beta(z); \bmz) \biggr).
\end{multline}

Now we can in fact replace the second sum by a sum over all non-empty disjoint subsets $\rho(z), \beta(z) \subset \tau(z)$ such that $\rho(z) \uplus \beta(z) = \tau(z)$, instead of $\tau'(z)$. This is because in the equation above, $z$ is already an entry in $ \mathcal{R}^{(|\beta|+1)} W_{g,n+1}$, so all subsets $\beta(z) \subset\nobreak \tau(z)$ that include $z$ are already taken into account. As for the subsets $\rho(z) \subset\nobreak \tau(z)$ that include $z$, they will not contribute to the sum, since for those $E^{(|\rho |)}(z; \rho(z))$ will vanish. However, since in the previous sum the subsets $\beta(z) \subset \tau'(z)$ had to be non-empty, in our new sum the subsets $\beta(z) \subset \tau(z)$ must have cardinality at least two. Therefore we get:
\begin{multline}
W_{g,n+1}(z_0, \bmz) = \sum_{a \in R} \underset{z = a}{\Res} \biggl(\frac{p_0(z)\omega^{z-\alpha}(z_0)}{\sfrac{\partial P}{\partial y}(z) dx(z)^{r-1}}\\[-5pt]
\times \sum_{\substack{\rho(z) \uplus \beta(z) = \tau(z)\\ |\beta(z)| \geq 2}} (-1)^{|\beta|} E^{(|\rho |)}(z;\rho(z)) \mathcal{R}^{(|\beta|)} W_{g,n+1}(\beta(z); \bmz) \biggr).
\end{multline}
Now we would like the terms with $|\beta(z)| = 1$ to be included in the sum as well. But this is exactly what happens if we bring the term on the LHS to the RHS. More precisely, for $2g-2+n \geq 0$, we can write
\begin{multline}
W_{g,n+1} (z_0, \bmz) = - \underset{z = z_0}{\Res}\ \omega^{z-\alpha}(z_0) W_{g,n+1}(z, \bmz)\\
\hspace*{-1cm}\shoveleft{= \sum_{a \in R} \underset{z = a}{\Res}\ \omega^{z-\alpha}(z_0) W_{g,n+1}(z, \bmz)}\\[-5pt]
\hspace*{-5mm}\shoveleft{+ \frac{1}{2\pi i}\sum_{i=1}^{\hat{g}} \biggl(\oint_{z \in A_i} B(z,z_0) \oint_{z \in B_i} W_{g,n+1}(z,\bmz)
 - \oint_{z \in B_i} B(z,z_0) \oint_{z \in A_i} W_{g,n+1}(z,\bmz) \biggr)}\\
\hspace*{-1cm}\shoveleft{= \sum_{a \in R} \underset{z = a}{\Res} \left(\omega^{z-\alpha}(z_0) W_{g,n+1}(z, \bmz) \right)}\\
\hspace*{-1cm}\shoveleft{= \sum_{a \in R} \underset{z = a}{\Res} \biggl(\frac{p_0(z)\omega^{z-\alpha}(z_0)}{\sfrac{\partial P}{\partial y}(z) dx(z)^{r-1}} E^{(r-1)}(z; \tau'(z))W_{g,n+1}(z,\bmz) \biggr)}\\
\hspace*{-1cm}\shoveleft{= \sum_{a \in R} \underset{z = a}{\Res} \biggl(\frac{p_0(z)\omega^{z-\alpha}(z_0)}{\sfrac{\partial P}{\partial y}(z) dx(z)^{r-1}} \sum_{\beta(z) \subset_1 \tau(z)} E^{(r-1)}(z; \tau(z) \setminus \beta(z))W_{g,n+1}(\beta(z),\bmz) \biggr),}
\end{multline}
where for the second equality we used Riemann's bilinear identity to pick up residues at the other poles ($a \in R$) of the integrand. Then we used the fact that $B(z,z_0)$ and the $W_{g,n}$ are normalized on $A$-cycles to show that the contour integrals vanish.

We then move this term to the RHS, and we end up with the equation:
\begin{multline}
0= \sum_{a \in R} \underset{z = a}{\Res} \biggl(\frac{p_0(z) \omega^{z-\alpha}(z_0)}{\sfrac{\partial P}{\partial y}(z) dx(z)^{r-1}}\\[-5pt]
\times \sum_{\rho(z) \uplus \beta(z) = \tau(z)} (-1)^{|\beta|} E^{(|\rho |)}(z;\rho(z)) \mathcal{R}^{(|\beta|)} W_{g,n+1}(\beta(z); \bmz) \biggr),
\label{eq:paa}
\end{multline}
where the summation now includes subsets $\beta(z) \subset \tau(z)$ of cardinality one.

Then, recall that
\begin{equation}
\begin{aligned}
E^{(|\rho |)}(z; \rho(z)) &= \prod_{i=1}^{|\rho|} \left(W_{0,1}(z) - W_{0,1}(\rho_i(z)) \right)\\
&=
\sum_{j=0}^{|\rho|} (-1)^j\, W_{0,1}(z)^{|\rho|-j}\,\, \sum_{\gamma(z)\subset_j \rho(z)} {\mathcal E}^{(j)} W_{0,1}(\gamma(z)),
\end{aligned}
\end{equation}
where we wrote the elements of $\rho(z)=\{\rho_1(z),\dots,\rho_{|\rho|}(z)\}$.

What we need to do now is collect terms in the second summation of \eqref{eq:paa} order by order in $W_{0,1}(z)$.
Writing $j=|\gamma|$ and $k=|\beta|+|\gamma|$, we have $|\rho|-j=r-|\beta|-|\gamma|=r-k$, that gives
\begin{equation}
\begin{aligned}
&\hspace*{-10mm}\sum_{\rho(z) \uplus \beta(z) = \tau(z)}\hspace*{-5mm} (-1)^{|\beta|} E^{(|\rho |)}(z;\rho(z)) \mathcal{R}^{(|\beta|)} W_{g,n+1}(\beta(z); \bmz)\\
&\hspace*{5mm}= \sum_k (-1)^k\, W_{0,1}(z)^{r-k}\hspace*{-5mm} \sum_{\gamma(z) \uplus \beta(z) \subset_k \tau(z)} \hspace*{-5mm}\mathcal E^{(|\gamma |)}W_{0,1}(\gamma(z)) \mathcal{R}^{(|\beta|)} W_{g,n+1}(\beta(z); \bmz)\\
&\hspace*{5mm}= \sum_k (-1)^k\, W_{0,1}(z)^{r-k}\hspace*{-5mm} \sum_{\rho(z)\subset_k \tau(z)}\hspace*{-1mm} \mathcal E^{(k)} W_{g,n+1}(\rho(z); \bmz)\\
&\hspace*{5mm}= \sum_k (-1)^k\, W_{0,1}(z)^{r-k}\, Q^{(k)}_{g,n+1}(z; \bmz)\\
&\hspace*{5mm}= \frac{\sfrac{\partial P}{\partial y}(z) dx(z)^{r-1}}{p_0(z)}\,Q_{g,n+1}(z; \bmz),
\end{aligned}
\end{equation}
where the second equality is lemma \ref{l:RE}.
We get the Theorem.
\end{proof}

\section{Pole analysis}
\label{s:pole}

For the next few sections, we now assume that our spectral curve $(\Sigma,x,y)$ is admissible, according to Definition \ref{d:admissible}.

In this section what we do is get rid of the residue in the topological recursion as presented in Theorem \ref{t:ref}. More precisely, what we will show is that Theorem \ref{t:ref} implies a nice formula for the expressions
\begin{equation}
\frac{p_0(z) Q^{(m)}_{g,n+1}(z; \bmz)}{dx(z)^m}.
\end{equation}
This result is what will give rise to the quantum curve in the next section.

\subsection{The $U^{(k)}_{g,n}$}

To proceed further we introduce the following objects:
\begin{defin}\label{d:Ug}
Using the notation of Definition \ref{d:tauz}, we define:
\begin{equation}\label{eq:Uk}
U^{(k)}_{g,n+1} (z; \bmz) = \sum_{\beta(z) \subseteq_k \tau'(z)} \mathcal{E}^{(k)} W_{g,n+1}(\beta(z); \bmz).
\end{equation}
For $k=0$, $g \geq 0$ and $n \geq 0$, we define
\begin{equation}
U^{(0)}_{g,n+1}(\bmz) = \delta_{g,0} \delta_{n,0}.
\end{equation}
\end{defin}

\begin{rem}
The difference between the $U^{(k)}_{g,n}$ and the $Q^{(k)}_{g,n}$ is that the latter sums over all preimages $\tau(z)$, while in the former we are only summing over the preimages in $\tau'(z) = \tau(z) \setminus \{z\}$. Thus, while $Q^{(k)}_{g,n}$ is the pullback (in $z$) of a $k$-differential on the base, $U^{(k)}_{g,n}$ is an honest $k$-differential in $z$ on $\Sigma$.
\end{rem}

\begin{rem}\label{r:Ur}
Note that for any spectral curve of degree $r$, and for all $(g,n)$,
\begin{equation}
U^{(k)}_{g,n+1}(z; \bmz) = 0, \quad \text{for all $k \geq r$.}
\end{equation}
This is because we are summing over subsets of $\tau'(z) = \tau(z) \setminus \{z\}$, hence there is no subset of $\tau'(z)$ of cardinality $k \geq r$ since $\tau'(z)$ has $r-1$ elements.
\end{rem}

\begin{rem}\label{rem:44}
Notice that for any $Y$ one has
\begin{equation}
\begin{aligned}
\frac{P(x(z),Y)}{Y-y(z)} &= p_0(z)\prod_{q\in \tau'(z)} (Y-y(q))=p_0(z) \prod_{i=1}^{r-1} (Y-y(\tau_i(z)))\\
&=p_0(z) \sum_{k=0}^{r-1} (-1)^k Y^{r-1-k} \sum_{\beta\subset_k \tau'(z)} \prod_{q\in \beta} y(q)\\
&=p_0(z) \sum_{k=0}^{r-1} (-1)^k Y^{r-1-k}\, \frac{U^{(k)}_{0,1}(z)}{dx(z)^k}.
\end{aligned}
\end{equation}
In particular if we choose $Y=y(z)=y(\tau_0(z))$ we get
\begin{equation}
\frac{\partial P}{\partial y}(z)= p_0(z) \sum_{k=0}^{r-1} (-1)^k y(z)^{r-1-k}\, \frac{U^{(k)}_{0,1}(z)}{dx(z)^k},
\end{equation}
while for $Y=y(\tau_i(z))$ with $i\neq 0$, we get
\begin{equation}
0 = p_0(z) \sum_{k=0}^{r-1}(-1)^k y(\tau_i(z))^{r-1-k}\, \frac{U^{(k)}_{0,1}(z)}{dx(z)^k}.
\end{equation}
Conversely,
\begin{equation}
\begin{aligned}\label{eq:up}
p_0(z) U_{0,1}^{(m)}(z)
&=(-1)^m dx(z)^m\sum_{k=0}^m p_{m-k}(z)\,y(z)^{k}\\
&= (-1)^m dx(z)^m \left(P_{m+1}(x(z), y(z)) + p_m(z) \right),
\end{aligned}
\end{equation}
where $P_{m+1}(x,y)$ was defined in Definition \ref{d:Pm}.
\end{rem}

The $U^{(k)}_{g,n}$ are closely related to the $Q^{(k)}_{g,n}$. We find:

\begin{lem}\label{l:PU} For all $g,n,k \geq 0$,
\begin{multline}
Q^{(k)}_{g,n+1}(z; \bmz) =U^{(k)}_{g,n+1}(z; \bmz)+U^{(k-1)}_{g-1,n+2}(z;\bmz,z)\\
+ \sum_{J_1 \uplus J_2 = \bmz} \sum_{g_1 + g_2 = g} U^{(k-1)}_{g_1, |J_1|+1}(z; J_1) W_{g_2,|J_2|+1}(z, J_2).
\end{multline}
\end{lem}

\begin{proof}
The $k=0$ case is obvious by definition. Let us focus on $k \geq 1$.
By definition, we have:
\begin{equation}
\begin{aligned}
Q^{(k)}_{g,n+1}(z; \bmz) &=\sum_{\beta(z) \subseteq_k \tau(z)} \mathcal{E}^{(k)} W_{g,n+1}(\beta(z); \bmz)\\
&= U^{(k)}_{g,n+1}(z; \bmz)+ \sum_{\rho(z) \subseteq_{k-1} \tau'(z)} \mathcal{E}^{(k)} W_{g,n+1}(z,\rho(z); \bmz).
\end{aligned}
\end{equation}
By Lemma \ref{l:Ek}, we have:
\begin{multline}
\mathcal{E}^{(k)} W_{g,n+1}(\bmt; \bmz) = \mathcal{E}^{(k-1)} W_{g-1,n+2} (\bmt \backslash \{t_k\}; \bmz, t_k)\\+\sum_{J_1 \uplus J_2 = \bmz} \sum_{g_1 + g_2 = g} \mathcal{E}^{(k-1)} W_{g_1, |J_1|+1}(\bmt \backslash \{t_k\}; J_1) W_{g_2,|J_2|+1}(t_k, J_2).
\end{multline}
Therefore
\begin{multline}
\sum_{\rho(z) \subseteq_{k-1} \tau'(z)} \mathcal{E}^{(k)} W_{g,n+1}(z,\rho(z); \bmz) = U^{(k-1)}_{g-1,n+2}(z; \bmz,z)\\[-5pt]
+ \sum_{J_1 \uplus J_2 = \bmz} \sum_{g_1 + g_2 = g} U^{(k-1)}_{g_1, |J_1|+1}(z; J_1) W_{g_2,|J_2|+1}(z, J_2),
\end{multline}
and the lemma is proved.
\end{proof}

\begin{cor}\label{c:Pk} For all $g,n,k \geq 0$,
\begin{multline}
Q^{(k)}_{g,n+1}(z; \bmz) =U^{(k)}_{g,n+1}(z; \bmz)+U^{(k-1)}_{g-1,n+2}(z;\bmz,z)\\
\shoveright{- \sum_{J_1 \uplus J_2 = \bmz} \sum_{g_1 + g_2 = g} U^{(k-1)}_{g_1, |J_1|+1}(z; J_1) U^{(1)}_{g_2,|J_2|+1}(z, J_2)
- \frac{p_1(z)}{p_0(z)}\, dx(z) U^{(k-1)}_{g, n+1}(z; \bmz)}\\[-5pt]
+ \sum_{i=1}^n \frac{dx(z) dx(z_i)}{(x(z)-x(z_i))^2}\, U^{(k-1)}_{g,n} (z; \bmz \setminus \{z_i\}).
\end{multline}
\end{cor}

\begin{proof}
This follows directly from Lemma \ref{l:PU}, Lemma \ref{l:P1} and the fact that
\begin{equation}
Q^{(1)}_{g,n+1}(z; \bmz) = W_{g,n+1}(z, \bmz) + U^{(1)}_{g,n+1}(z; \bmz).\qedhere
\end{equation}
\end{proof}

\subsection{Pole analysis}

We now prove the following lemma:

\begin{lem}\label{l:PP}
Consider the topological recursion presented in Theorem \ref{t:ref}. Then, for $2g-2+n \geq 0$, the meromorphic one-forms (in $z$)
\begin{equation}\label{eq:PP}
Q_{g,n+1}(z;\bmz) = \frac{dx(z)}{\sfrac{\partial P}{\partial y}(z)} \biggl(p_0(z)\sum_{k=1}^r (-1)^k y(z)^{r-k}\, \frac{Q^{(k)}_{g,n+1}(z;\bmz)}{dx(z)^k} \biggr)
\end{equation}
can only have poles at coinciding points, that is, at $z \in \tau(z_i)$, for $i=1,\ldots,n$.
\end{lem}

\begin{proof}
Let us first prove that they do not have poles on $R$.
We start with the topological recursion in Theorem \ref{t:ref}:
\begin{equation}\label{eq:p1}
0 = \sum_{i=1}^m \underset{z=a_i}{\Res}\ \omega^{z - \alpha}(z_0)Q_{g,n+1}(z; \bmz).
\end{equation}
Assume that $Q_{g,n+1}(z; \bmz)$ has a pole of some order $m+1\geq 1$ at a ramification point $a\in R$, \ie in some local coordinate near $a$ we may write
\begin{equation}
Q_{g,n+1}(z; \bmz) \sim \frac{dz}{(z-a)^{m+1}}\,S_{g,n+1}(\bmz) \,\, (1+\mathcal{O}(z-a)),
\end{equation}
where we assume $S_{g,n+1}(\bmz)\neq 0$.
We also have the Taylor expansion at $a$ of $\omega^{z-\alpha}(z_0)$ in the same local coordinate
\begin{equation}
\omega^{z-\alpha}(z_0) \sim \sum_{k=0}^\infty (z-a)^k \xi_{a,k}(z_0;\alpha),
\end{equation}
where $\xi_{a,k}(z_0;\alpha)$ is a meromorphic $1$-form of $z_0$ that is analytical everywhere but at~$a$, where it has a pole of order $k+1$. Using the same local coordinate near $a$ we have the Laurent expansion
\begin{equation}
\xi_{a,k}(z_0) \sim \frac{dz_0}{(z_0-a)^{k+1}}\,\,(1+\mathcal{O}(z_0-a)).
\end{equation}
Writing\enlargethispage{.5\baselineskip}%
\begin{equation}
\begin{aligned}
0 &= \sum_{b\in R} \underset{z = b}{\Res}\ \omega^{z-\alpha}(z_0) \,Q_{g,n+1}(z;\bmz)\\[-3pt]
&= \underset{z = a}{\Res}\ \omega^{z-\alpha}(z_0) \,Q_{g,n+1}(z;\bmz)
+\sum_{b\neq a} \underset{z = b}{\Res}\ \omega^{z-\alpha}(z_0) \,Q_{g,n+1}(z;\bmz),
\end{aligned}
\end{equation}
the sum of residues over $b\neq a$ may produce some linear combination of the $\xi_{b,k}(z_0)$s which don't have poles at $z_0\to a$.
The only term that can have a pole at $z_0=a$ comes from the residue at $z\to a$, and therefore we have
\begin{equation}
\begin{aligned}
\text{terms holomorphic at }z_0\to a &= \underset{z = a}{\Res} \omega^{z-\alpha}(z_0) \,Q_{g,n+1}(z;\bmz)\\
&= S_{g,n+1}(\bmz) \,\,\xi_{a,m}(z_0) (1+\mathcal{O}(z_0-a)),
\end{aligned}
\end{equation}
which is a contradiction since the right hand side is not analytical at $z_0\to a$ whenever $S_{g,n+1}(\bmz)\neq 0$.
This shows that $Q_{g,n+1}(z; \bmz)$ cannot have poles on $R$.

Now we need to check that they do not have poles elsewhere, except perhaps at the points $z \in \tau(z_i)$, for $i=1,\ldots,n$. Where could other poles come from? Assuming that the curve is smooth as an affine curve, it is clear from \eqref{eq:PP} that the only possible poles are at coinciding points $z \in \tau(z_i)$, for $i=1,\ldots,n$, or at punctures. So we need to show that $Q_{g,n+1}(z; \bmz) $ has no pole at the punctures.
But this follows directly by noticing that we could have included these punctures in the sum over residues in the definition of the topological recursion \eqref{eq:TR}, since the integrand cannot have poles at the punctures that are not in $R$, so those would not contribute to the residues. But then, by the same argument as above, we conclude that $Q_{g,n+1}(z, \bmz)$ is holomorphic at the punctures.

Therefore $Q_{g,n+1}(z, \bmz)$ can only have poles at coinciding points.
\end{proof}

To go further, we need to strengthen this result. What we are really interested in is not the one-forms $Q_{g,n+1}(z; \bmz)$, but rather the expressions
\begin{equation}
\frac{p_0(z) Q^{(m)}_{g,n+1}(z; \bmz)}{dx(z)^m}
\end{equation}
for each $m=1, \ldots, r$.
Let us consider the cases $(g,n) = (0,0)$ and $(g,n)=(0,1)$ first.

\begin{lem}\label{l:00}
For $m=0,\ldots,r$,
\begin{equation}
\frac{p_0(z) Q_{0,1}^{(m)}(z)}{dx(z)^m} = (-1)^m p_m(z).
\end{equation}
\end{lem}

\begin{proof}
This follows from Example \eqref{e:Q00}.
\end{proof}

\begin{lem}\label{l:01}
For $m=1,\ldots,r$,
\begin{multline}
\frac{p_0(z) Q_{0,2}^{(m)}(z; z_1)}{x(z)^{\lfloor\alpha_{r-m+1} \rfloor} dx(z)^m} =d_{z_1} \biggl(\frac{1}{x(z)-x(z_1)} \Bigl(\frac{U^{(m-1)}_{0,1}(z_1)}{dx(z_1)^{m-1}}\, \frac{p_0(z_1)}{x(z_1)^{\lfloor \alpha_{r-m+1} \rfloor}}\\
+ (-1)^{m-1} \Bigl(\frac{p_{m-1}(z)}{x(z)^{\lfloor\alpha_{r-m+1} \rfloor}} - \frac{p_{m-1}(z_1)}{x(z_1)^{\lfloor\alpha_{r-m+1} \rfloor}} \Bigr) \Bigr) \biggr),
\end{multline}
where $\alpha_m$ was defined in \eqref{eq:sup}.
\end{lem}

\begin{proof}
First, the case $m=1$ is straightforward. Since $U_{0,1}^{(0)}(z_1) = 1$, the statement is simply that
\begin{equation}
\frac{p_0(z) Q_{0,2}^{(1)}(z; z_1)}{dx(z)} = p_{0}(z) d_{z_1} \Bigl(\frac{1}{x(z)-x(z_1)} \Bigr)= \frac{p_0(z) dx(z_1)}{(x(z)-x(z_1))^2},
\end{equation}
which is indeed correct since
\begin{equation}
Q_{0,2}^{(1)}(z;z_1) = \frac{dx(z) dx(z_1)}{(x(z)-x(z_1))^2}.
\end{equation}

For $m=2,\ldots,r$, we have
\begin{multline}
\frac{p_0(z) Q_{0,2}^{(m)}(z)}{dx(z)^m} = \sum_{k=0}^{r-1} \frac{B(\tau_k(z), z_1)}{dx(z)} \frac{U^{(m-1)}_{0,1}(\tau_k(z)) p_0(z)}{dx(z)^{m-1}}\\= (-1)^{m-1} \biggl(\sum_{k=0}^{r-1} \frac{B(\tau_k(z), z_1)}{dx(z)} P_m(x(z),y(\tau_k(z))) + \sum_{k=0}^{r-1} \frac{B(\tau_k(z), z_1)}{dx(z)} p_{m-1}(z) \biggr),
\end{multline}
where the second equality follows from \eqref{eq:up} (recall that $P_m(x,y)$ was defined in Definition \ref{d:Pm}). The second term is easy to evaluate. Since
\begin{equation}
\sum_{k=0}^{r-1} \frac{B(\tau_k(z), z_1)}{dx(z)} = \frac{dx(z_1)}{(x(z)-x(z_1))^2},
\end{equation}
we get
\begin{equation}
\begin{aligned}
\sum_{k=0}^{r-1} \frac{B(\tau_k(z), z_1)}{dx(z)} p_{m-1}(z) &= p_{m-1}(z) \frac{dx(z_1)}{(x(z)-x(z_1))^2}\\
&= p_{m-1}(z) d_{z_1} \Bigl(\frac{1}{x(z)-x(z_1)} \Bigr).
\end{aligned}
\end{equation}

As for the first term, we rewrite it as follows:
\begin{equation}
\begin{aligned}
\sum_{k=0}^{r-1} \frac{B(\tau_k(z), z_1)}{dx(z)}& \,P_m(x(z),y(\tau_k(z)))\\[-8pt]
&= \sum_{k=0}^{r-1} \underset{z' = \tau_k(z)}{\Res} \frac{B(z',z_1)}{x(z')-x(z)}\, P_m(x(z'), y(z'))\\
&= \sum_{k=0}^{r-1} \underset{z' = \tau_k(z)}{\Res} \frac{B(z',z_1)}{x(z')-x(z)}\, \frac{x(z)^{\lfloor\alpha_{r-m+1} \rfloor} P_m(x(z'), y(z'))}{x(z')^{\lfloor \alpha_{r-m+1} \rfloor}},
\end{aligned}
\end{equation}
where $\alpha_m$ was defined in \eqref{eq:sup}.

Now recall from Lemma \ref{l:ab} that for $m=2,\ldots,r$,
\begin{equation}
\div(P_m) \geq \alpha_{r-m+1} \div_0(x) - \beta_{r-m+1} \div_\infty(x).
\end{equation}
Moreover, the Newton polygon of an admissible spectral curve has no interior point, hence, for all $m=2,\ldots,r$,
\begin{equation}
\lceil \beta_{r-m+1} \rceil - \lfloor \alpha_{r-m+1} \rfloor = 1.
\end{equation}
Therefore,
\begin{equation}\label{eq:limitex}
\begin{aligned}
&\hspace*{-7mm}\div \Bigl(\frac{P_m}{x^{\lfloor \alpha_{r-m+1} \rfloor}} \Bigr)\\
&\geq \left(\alpha_{r-m+1} -\lfloor \alpha_{r-m+1} \rfloor \right) \div_0 (x) - \left(\beta_{r-m+1} - \lceil \beta_{r-m+1} \rceil + 1 \right) \div_\infty (x)\\
&\geq - \div_\infty (x).
\end{aligned}
\end{equation}
It then follows that the only poles of the expression
\begin{equation}
\frac{B(z',z_1)}{x(z')-x(z)}\, \frac{x(z)^{\lfloor\alpha_{r-m+1} \rfloor} P_m(x(z'), y(z'))}{x(z')^{\lfloor \alpha_{r-m+1} \rfloor}}
\end{equation}
in $z'$ are at $z' = \tau_k(z)$ and at $z'=z_1$. Therefore, we obtain
\begin{equation}
\begin{aligned}
\sum_{k=0}^{r-1} \frac{B(\tau_k(z), z_1)}{dx(z)} &P_m(x(z),y(\tau_k(z)))\\[-8pt]
&= \underset{z' = z_1}{\Res} \frac{B(z',z_1)}{x(z)-x(z')} \frac{x(z)^{\lfloor\alpha_{r-m+1} \rfloor} P_m(x(z'), y(z'))}{x(z')^{\lfloor \alpha_{r-m+1} \rfloor}}\\
&=x(z)^{\lfloor\alpha_{r-m+1} \rfloor} d_{z_1} \left(\frac{1}{x(z)-x(z_1)} \frac{P_m(x(z_1), y(z_1))}{x(z_1)^{\lfloor \alpha_{r-m+1} \rfloor}} \right).
\end{aligned}
\end{equation}
Putting this together, we get
\begin{equation}
\begin{aligned}
&\hspace*{-1.5cm}\frac{p_0(z) Q_{0,2}^{(m)}(z)}{dx(z)^m}\\[-8pt]
&\hspace*{-.5cm}= (-1)^{m-1} d_{z_1} \left(\frac{x(z)^{\lfloor\alpha_{r-m+1} \rfloor}}{x(z)-x(z_1)} \frac{P_m(x(z_1), y(z_1))}{x(z_1)^{\lfloor \alpha_{r-m+1} \rfloor}} + \frac{p_{m-1}(z)}{x(z)-x(z_1)}\right)\\
&\hspace*{-.5cm}= x(z)^{\lfloor\alpha_{r-m+1} \rfloor} \biggl[ d_{z_1} \biggl(\frac{1}{x(z)-x(z_1)} \Bigl(\frac{U^{(m-1)}_{0,1}(z_1)}{dx(z_1)^{m-1}} \frac{p_0(z_1)}{x(z_1)^{\lfloor \alpha_{r-m+1} \rfloor}}\\
&\hspace*{3cm}+ (-1)^{m-1} \Bigl(\frac{p_{m-1}(z)}{x(z)^{\lfloor\alpha_{r-m+1} \rfloor}} - \frac{p_{m-1}(z_1)}{x(z_1)^{\lfloor\alpha_{r-m+1} \rfloor}} \Bigr) \Bigr) \biggr) \biggr].
\end{aligned}
\vspace*{-2\baselineskip}%
\end{equation}
\end{proof}

Now we consider the general case $(g,n) \neq (0,0), (0,1)$.
Let us first prove the following useful lemma.
\begin{lem}\label{l:inversion}
Consider a $r$-differential
\begin{equation}
Q(z) = dx(z)^r \sum_{k=1}^{r} (-1)^k y(z)^{r-k} \frac{Q_k(x(z))}{dx(z)^k},
\end{equation}
where the $Q_k(x(z))$ are $k$-differentials pullbacked from the base. Then
\begin{equation}
Q_k(x(z)) = - p_0(z) dx(z)^k \sum_{i=0}^{r-1} \Bigl(\frac{1}{\sfrac{\partial P}{\partial y}(\tau_i(z))}\, \frac{Q(\tau_i(z))}{dx(z)^r}\frac{U_{0,1}^{(k-1)} (\tau_i(z))}{dx(z)^{k-1}} \Bigr).
\end{equation}
where we denoted $\tau(z) = \{\tau_0(z),\dots,\tau_{r-1}(z)\}$ (the labeling is irrelevant).
\end{lem}

\begin{proof}
This is an application of the Lagrange interpolating polynomial.

Let us denote $\tau(z) = \{\tau_0(z),\dots,\tau_{r-1}(z)\}$, where we chose our labeling so that $z=\tau_0(z)$.

Recall from Remark \ref{rem:44} that
\begin{equation}
\frac{\partial P}{\partial y}(z)= p_0(z) \sum_{k=0}^{r-1} (-1)^k y(z)^{r-1-k}\, \frac{U^{(k)}_{0,1}(z)}{dx(z)^k},
\end{equation}
and
\begin{equation}
0 = p_0(z) \sum_{k=0}^{r-1}(-1)^k y(\tau_i(z))^{r-1-k}\, \frac{U^{(k)}_{0,1}(z)}{dx(z)^k}.
\end{equation}
In other words
\begin{equation}
p_0(z) \sum_{k=0}^{r-1} (-1)^k y(z)^{r-k-1}\, \frac{U_{0,1}^{(k)} (\tau_i(z))}{dx(z)^k}
= \frac{\partial P}{\partial y}(z) \delta_{i,0}.
\end{equation}

Then we can compute:
\begin{equation}
\begin{aligned}
- dx(z)^r &\sum_{k=1}^{r} (-1)^k y(z)^{r-k} p_0(z) \sum_{i=0}^{r-1} \Bigl(\frac{1}{\sfrac{\partial P}{\partial y}(\tau_i(z))}\, \frac{Q(\tau_i(z))}{dx(z)^r}\,\frac{U_{0,1}^{(k-1)} (\tau_i(z))}{dx(z)^{k-1}} \Bigr)\\
&= \sum_{i=0}^{r-1} \frac{Q(\tau_i(z))}{\sfrac{\partial P}{\partial y}(\tau_i(z))}\,p_0(z) \sum_{k=0}^{r-1} (-1)^k y(z)^{r-k-1}\, \frac{U_{0,1}^{(k)} (\tau_i(z))}{dx(z)^{k}}\\
&= \sum_{i=0}^{r-1} \frac{Q(\tau_i(z))}{\sfrac{\partial P}{\partial y}(\tau_i(z))}\, \frac{\partial P}{\partial y}(z)\, \delta_{i,0}= Q(z).
\end{aligned}
\end{equation}
Hence
\begin{equation}
Q_k(x(z)) = - p_0(z) dx(z)^k \sum_{i=0}^{r-1} \Bigl(\frac{1}{\sfrac{\partial P}{\partial y}(\tau_i(z))}\, \frac{Q(\tau_i(z))}{dx(z)^r}\,\frac{U_{0,1}^{(k-1)} (\tau_i(z))}{dx(z)^{k-1}} \Bigr).\qedhere
\end{equation}
\end{proof}

As an application of this lemma, we get a relation between the $Q^{(m)}_{g,n+1}(z; \bmz)$ and the $Q_{g,n+1}(z; \bmz)$:

\begin{cor}\label{l:L}
For $2g-2+n \geq 0$ and $m=1,\ldots,r$,
\begin{equation}
Q^{(m)}_{g,n+1}(z; \bmz)= - \sum_{k=0}^{r-1} \bigl(Q_{g,n+1}(\tau_k(z); \bmz) U_{0,1}^{(m-1)}(\tau_k(z)) \bigr).
\end{equation}
\end{cor}

\begin{proof}
Recall that
\begin{equation}
Q_{g,n+1}(z;\bmz) = \frac{dx(z)}{\sfrac{\partial P}{\partial y}(z)}\, p_0(z)\sum_{k=1}^r (-1)^k y(z)^{r-k} \,\frac{Q^{(k)}_{g,n+1}(z;\bmz)}{dx(z)^k}.
\end{equation}
Then we simply apply Lemma \ref{l:inversion} to get the statement of the corollary.
\end{proof}

Finally we can prove the following theorem on the expressions
\begin{equation}
\frac{p_0(z) Q^{(m)}_{g,n+1}(z; \bmz)}{dx(z)^m}.
\end{equation}

\begin{thm}\label{t:pole}
For $2g-2+n \geq 0$, $m=1,\ldots, r$,
\begin{multline}
\frac{p_0(z) Q_{g,n+1}^{(m)}(z; \bmz)}{x(z)^{\lfloor\alpha_{r-m+1} \rfloor}dx(z)^m}\\
=\sum_{i=1}^n d_{z_i} \biggl(\frac{1}{x(z)-x(z_i)} \Bigl(\frac{U^{(m-1)}_{g,n}(z_i; \bmz \setminus \{z_i \})}{dx(z_i)^{m-1}}\, \frac{p_0(z_i)}{x(z_i)^{\lfloor \alpha_{r-m+1} \rfloor}} \Bigr) \biggl).
\label{eq:tpole}
\end{multline}
For $(g,n) = (0,1)$, $m=1, \ldots, r$,
\begin{multline}
\frac{p_0(z) Q_{0,2}^{(m)}(z; z_1)}{x(z)^{\lfloor\alpha_{r-m+1} \rfloor} dx(z)^m} =d_{z_1} \biggl(\frac{1}{x(z)-x(z_1)} \Bigl(\frac{U^{(m-1)}_{0,1}(z_1)}{dx(z_1)^{m-1}}\, \frac{p_0(z_1)}{x(z_1)^{\lfloor \alpha_{r-m+1} \rfloor}}\\
+ (-1)^{m-1} \Bigl(\frac{p_{m-1}(z)}{x(z)^{\lfloor\alpha_{r-m+1} \rfloor}} - \frac{p_{m-1}(z_1)}{x(z_1)^{\lfloor\alpha_{r-m+1} \rfloor}} \Bigr) \Bigr) \biggr),
\end{multline}
while for $(g,n) = (0,0)$, $m=0,\ldots,r$,
\begin{equation}
\frac{p_0(z) Q_{0,1}^{(m)}(z)}{dx(z)^m} = (-1)^m p_m(z).
\end{equation}
\end{thm}

\begin{proof}
The $(g,n) =(0,0)$ and $(g,n) = (0,1)$ statements were proved in Lemmas \ref{l:00} and \ref{l:01}. Let us then focus on $2g-2+n \geq 0$.

The proof follows along similar lines to the proof of Lemma \ref{l:01}.

First, the case $m=1$ is trivial, since both the left-hand-side and the right-hand-side are zero.

So we consider $m=2,\ldots,r$. We have
\begin{equation}
\begin{aligned}
\frac{p_0(z) Q_{g,n+1}^{(m)}(z; \bmz)}{dx(z)^m} &= - \sum_{k=0}^{r-1} \frac{Q_{g,n+1}(\tau_k(z); \bmz)}{dx(z)} \,\frac{U^{(m-1)}_{0,1}(\tau_k(z)) p_0(z)}{dx(z)^{m-1}}\\[-5pt]
&= (-1)^{m} \biggl(\sum_{k=0}^{r-1} \frac{Q_{g,n+1}(\tau_k(z); \bmz)}{dx(z)} \,P_m(x(z),y(\tau_k(z)))\\[-5pt]
&\hspace*{3.5cm}+ \sum_{k=0}^{r-1} \frac{Q_{g,n+1}(\tau_k(z); \bmz)}{dx(z)} \,p_{m-1}(z) \biggr),
\end{aligned}
\end{equation}
where the second equality follows from \eqref{eq:up} (recall that $P_m(x,y)$ was defined in Definition \ref{d:Pm}).

The second term vanishes, since
\begin{equation}
\sum_{k=0}^{r-1} Q_{g,n+1}(\tau_k(z); \bmz) = 0.
\end{equation}
Indeed, using the fact that $U_{0,1}^{(0)} = 1$, from Corollary \ref{l:L} we get
\begin{equation}
\sum_{k=0}^{r-1} Q_{g,n+1}(\tau_k(z); \bmz) = \sum_{k=0}^{r-1} Q_{g,n+1}(\tau_k(z); \bmz) U_{0,1}^{(0)} = - Q^{(1)}_{g,n+1}(z; \bmz),
\end{equation}
which vanishes according to Lemma \ref{l:P1}.

As for the first term, we rewrite it as
\begin{equation}
\begin{aligned}
&\hspace*{-1cm}\sum_{k=0}^{r-1} \frac{Q_{g,n+1}(\tau_k(z); \bmz)}{dx(z)} \,P_m(x(z),y(\tau_k(z)))\\[-8pt]
&\hspace*{1cm}=
\sum_{k=0}^{r-1} \underset{z' = \tau_k(z)}{\Res} \frac{Q_{g,n+1}(z'; \bmz)}{x(z')-x(z)} \,P_m(x(z'), y(z'))\\[-5pt]
&\hspace*{1cm}=x(z)^{\lfloor\alpha_{r-m+1} \rfloor} \sum_{k=0}^{r-1} \underset{z' = \tau_k(z)}{\Res} \frac{Q_{g,n+1}(z'; \bmz)}{x(z')-x(z)} \,\frac{P_m(x(z'), y(z'))}{x(z')^{\lfloor \alpha_{r-m+1} \rfloor}},
\end{aligned}
\end{equation}
where $\alpha_m$ was defined in \eqref{eq:sup}.

By the same argument as in the proof of Lemma \ref{l:01}, the only poles of the integrand in $z'$ are at $z' = \tau_k(z)$ and at the poles of $Q_{g,n+1}(z'; \bmz)$. From Lemma \ref{l:PP}, we know that $Q_{g,n+1}(z'; \bmz)$ can only have poles at $z' = \tau_k(z_j)$, $k=0,\ldots,r-1$, $j=1,\ldots,n$. Therefore, we get that the expression above is equal to
\begin{equation}
- x(z)^{\lfloor\alpha_{r-m+1} \rfloor} \sum_{j=1}^n \sum_{k=0}^{r-1} \underset{z' = \tau_k(z_j)}{\Res} \frac{Q_{g,n+1}(z'; \bmz)}{x(z')-x(z)}\, \frac{P_m(x(z'), y(z'))}{x(z')^{\lfloor \alpha_{r-m+1} \rfloor}},
\end{equation}
which is in turn equal to
\begin{equation}
- x(z)^{\lfloor\alpha_{r-m+1} \rfloor}\sum_{j=1}^n \sum_{k=0}^{r-1} \underset{z' = z_j}{\Res} \frac{Q_{g,n+1}(\tau_k(z'); \bmz)}{x(z')-x(z)} \,\frac{P_m(x(z'), y(\tau_k(z')))}{x(z')^{\lfloor \alpha_{r-m+1} \rfloor}}.
\end{equation}
Putting this together, we get
\begin{equation}
\begin{aligned}
&\hspace*{-.8cm}\frac{p_0(z) Q_{g,n+1}^{(m)}(z; \bmz)}{dx(z)^m}\\[-8pt]
&=x(z)^{\lfloor\alpha_{r-m+1} \rfloor} (-1)^{m-1} \sum_{j=1}^n \sum_{k=0}^{r-1} \underset{z' = z_j}{\Res} \frac{Q_{g,n+1}(\tau_k(z'); \bmz)}{x(z')-x(z)} \,\frac{P_m(x(z'), y(\tau_k(z')))}{x(z')^{\lfloor \alpha_{r-m+1} \rfloor}}\\[-3pt]
&= x(z)^{\lfloor\alpha_{r-m+1} \rfloor} \sum_{j=1}^n \sum_{k=0}^{r-1} \underset{z' = z_j}{\Res} \frac{Q_{g,n+1}(\tau_k(z'); \bmz)}{x(z')-x(z)} \,\frac{U_{0,1}^{(m-1)}(\tau_k(z')) p_0(z')}{x(z')^{\lfloor \alpha_{r-m+1} \rfloor} dx(z')^{m-1}}\\[-5pt]
& \hspace*{1cm} + x(z)^{\lfloor\alpha_{r-m+1} \rfloor}(-1)^m \sum_{j=1}^n \sum_{k=0}^{r-1} \underset{z' = z_j}{\Res} \frac{Q_{g,n+1}(\tau_k(z'); \bmz)}{x(z')-x(z)}\, \frac{p_{m-1}(z')}{x(z')^{\lfloor \alpha_{r-m+1} \rfloor}},
\end{aligned}
\end{equation}
where for the second equality we used again \eqref{eq:up}. The second term vanishes because
\begin{equation}
\sum_{k=0}^{r-1} Q_{g,n+1}(\tau_k(z'); \bmz) = 0,
\end{equation}
hence
\begin{equation}
\begin{aligned}
&\hspace*{-1cm}\frac{p_0(z) Q_{g,n+1}^{(m)}(z; \bmz)}{dx(z)^m}\\[-5pt]
&= x(z)^{\lfloor\alpha_{r-m+1} \rfloor}\sum_{j=1}^n \sum_{k=0}^{r-1} \underset{z' = z_j}{\Res} \frac{Q_{g,n+1}(\tau_k(z'); \bmz)}{x(z')-x(z)}\, \frac{U_{0,1}^{(m-1)}(\tau_k(z')) p_0(z')}{x(z')^{\lfloor \alpha_{r-m+1} \rfloor} dx(z')^{m-1}}\\
&= x(z)^{\lfloor\alpha_{r-m+1} \rfloor} \sum_{j=1}^n \underset{z' = z_j}{\Res} \frac{Q^{(m)}_{g,n+1}(z'; \bmz) p_0(z')}{(x(z)-x(z'))x(z')^{\lfloor \alpha_{r-m+1} \rfloor} dx(z')^{m-1}}\\
&= x(z)^{\lfloor\alpha_{r-m+1} \rfloor} \sum_{j=1}^n \underset{z' = z_j}{\Res} \frac{B(z', z_j) U^{(m-1)}_{g,n}(z'; \bmz \setminus \{z_j \}) p_0(z')}{(x(z)-x(z'))x(z')^{\lfloor \alpha_{r-m+1} \rfloor} dx(z')^{m-1}}\\
&= x(z)^{\lfloor\alpha_{r-m+1} \rfloor} \sum_{j=1}^n d_{z_j} \Bigl(\frac{U^{(m-1)}_{g,n}(z_j; \bmz \setminus \{z_j \}) p_0(z_j)}{(x(z)-x(z_j))x(z_j)^{\lfloor \alpha_{r-m+1} \rfloor} dx(z_j)^{m-1}} \Bigr).
\end{aligned}
\vspace*{-2.2\baselineskip}
\end{equation}
\end{proof}

As an immediate corollary, we get:

\begin{lem}\label{l:Uk}
For $m=1,\ldots,r$,
\begin{multline}\label{eq:toint}
\frac{p_0(z)}{x(z)^{\lfloor\alpha_{r-m+1} \rfloor}} \,\frac{U^{(m)}_{g,n+1}(z; \bmz)}{dx(z)^m}\\
= - \frac{p_0(z)}{x(z)^{\lfloor\alpha_{r-m+1} \rfloor}}\,\frac{U^{(m-1)}_{g-1,n+2}(z;\bmz,z)}{dx(z)^{m-1} dx(z)} + \frac{p_1(z)}{x(z)^{\lfloor\alpha_{r-m+1} \rfloor}} \,\frac{U^{(m-1)}_{g, n+1}(z; \bmz)}{dx(z)^{m-1}}\\
\shoveright{+ \frac{p_0(z)}{x(z)^{\lfloor\alpha_{r-m+1} \rfloor}} \sum_{J_1 \uplus J_2 = \bmz} \sum_{g_1 + g_2 = g} \frac{U^{(m-1)}_{g_1, |J_1|+1}(z; J_1)}{dx(z)^{m-1}} \,\frac{U^{(1)}_{g_2,|J_2|+1}(z; J_2)}{dx(z)}}\\
\shoveright{- \sum_{i=1}^n \biggl[\frac{p_0(z)}{x(z)^{\lfloor\alpha_{r-m+1} \rfloor}} \,\frac{dx(z_i)}{(x(z)-x(z_i))^2}\, \frac{U^{(m-1)}_{g,n} (z; \bmz \setminus \{z_i\})}{dx(z)^{m-1}}}\\
\shoveright{- d_{z_i} \Bigl(\frac{p_0(z_i)}{x(z_i)^{\lfloor\alpha_{r-m+1} \rfloor}} \frac{1}{x(z)-x(z_i)} \frac{U^{(m-1)}_{g,n}(z_i; \bmz \setminus \{z_i \})}{dx(z_i)^{m-1}} \Bigr) \biggr]}\\
+ \delta_{g,0} \delta_{n,0} (-1)^m \frac{p_m(z)}{x(z)^{\lfloor\alpha_{r-m+1} \rfloor}}\\
+ \delta_{g,0} \delta_{n,1} (-1)^{m-1} d_{z_1}\left(\frac{1}{x(z) - x(z_1)} \Bigl(\frac{p_{m-1}(z)}{x(z)^{\lfloor\alpha_{r-m+1} \rfloor}} - \frac{p_{m-1}(z_1)}{x(z_1)^{\lfloor\alpha_{r-m+1} \rfloor}} \Bigr)\right).
\end{multline}
\end{lem}

\begin{proof}
This follows directly from Theorem \ref{t:pole} and Corollary \ref{c:Pk}.
\end{proof}

\section{Quantum curves}

We are now ready to prove the existence of the quantum curves. What we need to do is integrate \eqref{eq:toint}.

\Subsection{Integration}

\subsubsection{Integration procedure}

Let us first define the following integration procedure.
\begin{defin}
Let $D = \sum_i \chi_i [p_i]$ be a divisor on $\Sigma$, with $p_i \in \Sigma$.
We define its degree $\deg D =\sum_i \chi_i$. The set of degree 0 divisors of $\Sigma$ is called $\Div_0(\Sigma)$.

For $D\in \Div_0(\Sigma)$ we define integration of a meromorphic one-form $\nu(z)$ on $\Sigma$ as
\begin{equation}
\int_D \nu(z) = \sum_i \chi_i \int^{p_i}_b \nu (z),
\end{equation}
where $b \in \Sigma$ is an arbitrary base point, and the integration contours are the unique homology chains $(b,p_i)$ that do not intersect our basis of non-contractible cycles. Since we assumed that $\deg D=\sum_i \chi_i = 0$, it follows that the integral above does not depend on the choice of base point $b$.
\end{defin}

\begin{rem}
Here and in what follows, we always assume that the integration divisor is chosen such that all integrals converge.
\end{rem}

\begin{defin}
Let $D_1, \ldots, D_n$ be $n$ arbitrary degree 0 divisors on $\Sigma$, with
\begin{equation}
D_i = \sum_j \chi_{i,j} [z_{i,j}].
\end{equation}
We define
\begin{equation}
G^{(k)}_{g,n+1}(z; D_1, \ldots, D_n) = \int_{D_1} \cdots \int_{D_n} U^{(k)}_{g,n+1}(z; z_1,\ldots,z_n).
\end{equation}
Note that we are not integrating in $z$. The $ G^{(k)}_{g,n+1}(z; D_1, \ldots, D_n) $ should be understood as differentials in $z$, and functions of the points $z_{i,j} \in \Sigma$ in the definition of the integration divisors.

We also define the so-called ``principal specialization''
\begin{equation}
G^{(k)}_{g,n+1}(z; D) = \int_D \cdots \int_D U^{(k)}_{g,n+1}(z; z_1,\ldots,z_n),
\end{equation}
where we set all integration divisors to be equal.
\end{defin}

With these definitions, we can integrate the equation in Lemma \ref{l:Uk}:
\begin{lem}\label{l:int1}
For $i=1,\ldots,n$, let $D_i=\sum_j \chi_{i,j} [z_{i,j}]$ be arbitrary degree zero divisors, and introduce the notation $\bmD = \{D_1, \ldots, D_n\}$. Let $D_{n+1} = \chi z' + D'$, with $D'$ an arbitrary divisor of degree $-\chi$, for some $\chi$. Then:
\begin{multline}\label{eq:GknDi}
\frac{p_0(z)}{x(z)^{\lfloor\alpha_{r-k+1} \rfloor}}\, \frac{G^{(k)}_{g,n+1}(z; \bmD)}{dx(z)^k}\\[-5pt]
= - \frac{p_0(z)}{\chi\ x(z)^{\lfloor\alpha_{r-k+1} \rfloor}} \,\frac{d}{dx(z')} \Bigl(\frac{G^{(k-1)}_{g-1,n+2}(z; \bmD, D_{n+1})}{dx(z)^{k-1}} \Bigr)_{z'=z}\\
\shoveright{+ \frac{p_1(z)}{x(z)^{\lfloor\alpha_{r-k+1} \rfloor}} \frac{G^{(k-1)}_{g, n+1}(z; \bmD)}{dx(z)^{k-1}}+ \delta_{g,0} \delta_{n,0} (-1)^k \frac{p_k(z)}{x(z)^{\lfloor\alpha_{r-k+1} \rfloor}}}\\
\shoveright{+ \frac{p_0(z)}{x(z)^{\lfloor\alpha_{r-k+1} \rfloor}} \sum_{J_1 \uplus J_2 = \bmD} \sum_{g_1 + g_2 = g} \frac{G^{(k-1)}_{g_1, |J_1|+1}(z; J_1)}{dx(z)^{k-1}} \frac{G^{(1)}_{g_2,|J_2|+1}(z; J_2)}{dx(z)}}\\
\shoveright{- \sum_{i=1}^n \sum_j \chi_{i,j} \biggl[ \frac{p_0(z)}{x(z)^{\lfloor\alpha_{r-m+1} \rfloor}}\frac{1}{x(z)-x(z_{i,j})} \frac{G^{(k-1)}_{g,n} (z; \bmD \setminus \{D_i\})}{dx(z)^{k-1}}\\[-3pt]
- \frac{p_0(z_{i,j})}{x(z_{i,j})^{\lfloor\alpha_{r-k+1} \rfloor}} \frac{1}{x(z)-x(z_{i,j})} \frac{G^{(k-1)}_{g,n}(z_{i,j}; \bmD \setminus \{D_i \})}{dx(z_{i,j})^{k-1}} \biggr]}\\
+ \delta_{g,0} \delta_{n,1} (-1)^{k-1}\sum_{j} \chi_{1,j} \left(\frac{1}{x(z) - x(z_{1,j})} \Bigl(\frac{p_{k-1}(z)}{x(z)^{\lfloor\alpha_{r-k+1} \rfloor}} - \frac{p_{k-1}(z_{1,j})}{x(z_{1,j})^{\lfloor\alpha_{r-k+1} \rfloor}} \Bigr) \right).
\end{multline}
\end{lem}

\begin{proof}
This is the straightforward integration of Lemma \ref{l:Uk}.
\end{proof}

\subsubsection{Principal specialization}

Let us now principal specialize this equation by setting all divisors equal. We get:

\begin{lem}
When all degree zero divisors are chosen equal, and all containing the point $z$ as:
\begin{equation}
D_i = D = \chi z + \sum_i \chi_i z_i,
\end{equation}
we obtain
\begin{multline}\label{eq:GknD}
\frac{p_0(z)}{x(z)^{\lfloor\alpha_{r-k+1} \rfloor}} \frac{G^{(k)}_{g,n+1}(z; D)}{dx(z)^k}\\
= - \frac{p_0(z)}{\chi (n+1) x(z)^{\lfloor\alpha_{r-k+1} \rfloor}} \frac{d}{dx(z)} \Bigl(\frac{G^{(k-1)}_{g-1,n+2}(z'; D)}{dx(z')^{k-1}} \Bigr)_{z'=z}\\
+ \frac{p_1(z)}{x(z)^{\lfloor\alpha_{r-k+1} \rfloor}} \frac{G^{(k-1)}_{g, n+1}(z; D)}{dx(z)^{k-1}}\\
+ \frac{p_0(z)}{x(z)^{\lfloor\alpha_{r-k+1} \rfloor}} \sum_{m=0}^n \sum_{g_1 + g_2 = g} \frac{n!}{m! (n-m)!} \frac{G^{(k-1)}_{g_1, m+1}(z; D)}{dx(z)^{k-1}} \frac{G^{(1)}_{g_2,n-m+1}(z; D)}{dx(z)}\\
- n \sum_j \chi_{j} \biggl[\frac{p_0(z)}{x(z)^{\lfloor\alpha_{r-k+1} \rfloor}} \frac{1}{x(z)-x(z_{j})} \frac{G^{(k-1)}_{g,n} (z; D)}{dx(z)^{k-1}}\\[-5pt]
\shoveright{- \frac{p_0(z_j)}{x(z_j)^{\lfloor\alpha_{r-k+1} \rfloor}} \frac{1}{x(z)-x(z_{j})} \frac{G^{(k-1)}_{g,n}(z_{j}; D)}{dx(z_{j})^{k-1}} \biggr]}\\
- n \chi \frac{d}{dx(z')} \biggl(\frac{p_0(z')}{x(z')^{\lfloor\alpha_{r-k+1} \rfloor}} \frac{G^{(k-1)}_{g,n} (z'; D)}{dx(z')^{k-1}} \biggr)_{z'=z}
+ \delta_{g,0} \delta_{n,0} (-1)^k \frac{p_k(z)}{x(z)^{\lfloor\alpha_{r-k+1} \rfloor}}\\
+\delta_{g,0} \delta_{n,1} (-1)^{k-1} \biggl[ \sum_{j} \chi_{j} \Bigl(\frac{1}{x(z) - x(z_{j})} \Bigl(\frac{p_{k-1}(z)}{x(z)^{\lfloor\alpha_{r-k+1} \rfloor}} - \frac{p_{k-1}(z_{j})}{x(z_{j})^{\lfloor\alpha_{r-k+1} \rfloor}} \Bigr) \Bigr)\\
+\chi \frac{d}{dx} \Bigl(\frac{p_{k-1}(z)}{x(z)^{\lfloor\alpha_{r-k+1} \rfloor}} \Bigr) \biggr].
\end{multline}
\end{lem}

\begin{proof}
The specialization is straightforward from \eqref{eq:GknDi}. Only the terms involving derivatives require some care.
Indeed as $D_i\to D$, $1/(x(z)-x(z_{i,j}))\to 1/(x(z)-x(z_j))$, as long as $z_j\neq z$.
When $z_{i,j}\to z$, the limit with the denominator $1/(x(z)-x(z_{i,j})$ tends to the derivative, giving these terms.
\end{proof}

\subsubsection{Summing over $g$ and $n$}

We now sum over $g$ and $n$. Let us define:
\begin{defin}
For $m=1,\ldots,r$, we define
\begin{equation}
\xi_m(z;D) = (-1)^m \sum_{g,n=0}^\infty \frac{\hbar^{2g+n}}{n!} \,\frac{G^{(m)}_{g,n+1}(z; D)}{dx(z)^m}.
\end{equation}
We get $\xi_0(z;D) = 1$, and define $\xi_k(z;D) = 0$ for all $k<0$. It is also clear that $\xi_k(z;D) = 0$ for all $k \geq r$.
\end{defin}

\begin{rem}
In the following, we will abuse notation and write $\xi_m(x;D)$, even though it is actually a multivalued function on the base. We will do so to lighten notation and use for instance $d/dx$ rather than $d/dx(z)$.
\end{rem}

Summing over $g$ and $n$, \eqref{eq:GknD} becomes
\begin{multline}
\frac{p_0(x)}{x^{\lfloor\alpha_{r-k+1} \rfloor}} \xi_k(x; D)\\
= \frac{p_k(x)}{x^{\lfloor\alpha_{r-k+1} \rfloor}} - \frac{p_1(x)}{x^{\lfloor\alpha_{r-k+1} \rfloor}} \xi_{k-1}(x; D)
+ \frac{p_0(x)}{x^{\lfloor\alpha_{r-k+1} \rfloor}} \xi_{k-1}(x; D) \xi_1(x; D)\\
\shoveright{+ \hbar \sum_i \chi_{i} \frac{1}{x-x_i} \Bigl(\frac{p_0(x)}{x^{\lfloor\alpha_{r-k+1} \rfloor}} \xi_{k-1}(x; D) - \frac{p_0(x_i)}{x_i^{\lfloor\alpha_{r-k+1} \rfloor}} \xi_{k-1} (x_i; D) \Bigr)}\\
\shoveright{+ \frac{\hbar}{\chi} \frac{d}{dx} \Bigl(\frac{p_0(x')}{x'^{\lfloor\alpha_{r-k+1} \rfloor}} \xi_{k-1}(x'; D) \Bigr)_{x' = x} + \hbar \chi \frac{d}{dx'} \Bigl(\frac{p_0(x')}{x'^{\lfloor\alpha_{r-k+1} \rfloor}} \xi_{k-1} (x'; D) \Bigr)_{x'=x}}\\
- \hbar \sum_{i} \chi_{i} \frac{1}{x-x_i} \Bigl(\frac{p_{k-1}(x)}{x^{\lfloor\alpha_{r-k+1} \rfloor}} - \frac{p_{k-1}(x_i)}{x_i^{\lfloor\alpha_{r-k+1} \rfloor}} \Bigr) - \hbar \chi \frac{d}{dx} \Bigl(\frac{p_{k-1}(x)}{x^{\lfloor\alpha_{r-k+1} \rfloor}} \Bigr).
\end{multline}

We see that something nice happens if $\chi=1/\chi$, \ie $\chi=\pm 1$.
$\chi$ may be called the ``charge'', and in analogy with CFTs, $\chi=\pm 1$ would be called a ``degenerate charge''.

In that case we get:

\begin{lem}\label{lem:eqxikD}
If $D=\chi [z] + \sum_i \chi_i [z_i]\in \Div_0(\Sigma)$ with $\chi=\pm 1$, we obtain the following differential recursion relation for the $\xi_k$'s, $k \geq 1$ (and we recall that $\xi_k=0$ if $k\geq r$ and $k < 0$, and $\xi_0=1$):
\begin{multline}
\frac{p_0(x)}{x^{\lfloor\alpha_{r-k+1} \rfloor}} \,\xi_k(x; D) - \frac{p_k(x)}{x^{\lfloor\alpha_{r-k+1} \rfloor}}\\
= - \frac{p_1(x)}{x^{\lfloor\alpha_{r-k+1} \rfloor}} \,\xi_{k-1}(x; D)
+ \frac{p_0(x)}{x^{\lfloor\alpha_{r-k+1} \rfloor}} \,\xi_{k-1}(x; D) \xi_1(x; D)\\
\shoveright{+ \hbar \sum_i \chi_{i} \frac{1}{x-x_i} \Bigl(\frac{p_0(x)}{x^{\lfloor\alpha_{r-k+1} \rfloor}} \,\xi_{k-1}(x; D) - \frac{p_0(x_i)}{x_i^{\lfloor\alpha_{r-k+1} \rfloor}} \,\xi_{k-1} (x_i; D) \Bigr)}\\
\shoveright{- \hbar \sum_{i} \chi_{i} \frac{1}{x-x_i} \Bigl(\frac{p_{k-1}(x)}{x^{\lfloor\alpha_{r-k+1} \rfloor}} - \frac{p_{k-1}(x_i)}{x_i^{\lfloor\alpha_{r-k+1} \rfloor}} \Bigr)}\\
+ \hbar \chi \frac{d}{dx} \Bigl(\frac{p_0(x)}{x^{\lfloor\alpha_{r-k+1} \rfloor}} \,\xi_{k-1} (x; D) -\frac{p_{k-1}(x)}{x^{\lfloor\alpha_{r-k+1} \rfloor}} \Bigr).
\label{eq:xi}
\end{multline}

\end{lem}

\subsection{Quantum curves}

The non-linear differential system in Lemma \ref{lem:eqxikD} can be linearized as in the Riccati equation. We assume that the $z_i$ are in generic position (in particular they are not in $R$, so that the integrals converge). We let $\chi_0=\chi = \pm 1$ and $z_0 = z$.

\begin{defin}
For $m=1, \ldots, r$, we define
\begin{equation}
\psi_m(x; D) = \psi(D) \frac{1}{x^{\lfloor \alpha_{r-m} \rfloor}} \left[p_0(x) \xi_m(x; D) - p_m(x) \right],
\end{equation}
with
\begin{multline}\label{e:psi}
\psi(D) = \exp\Bigl(\frac{1}{\hbar} \int_D W_{0,1}\Bigr)\\
{}\cdot\exp\biggl(\sum_{(g,n)\neq (0,1)} \frac{\hbar^{2g+n-2}}{n!} \int_D \cdots \int_D \Bigl(W_{g,n}(z_1,\ldots,z_n) - \delta_{g,0}\delta_{n,2} \frac{dx(z_1) dx(z_2)}{(x(z_1) - x(z_2))^2} \Bigr) \biggr).
\end{multline}
Note that the integral $\int_D W_{0,1}$ may need to be regularized. But this will play no role in the following because changing the regularization of $\int_D W_{0,1}$ only amounts to multiplying $\psi(D)$ by a constant.
\end{defin}

Then we can prove:
\begin{lem}\label{l:xi1}
For $m=r$,
\begin{equation}
\psi_r(x;D) = - \frac{p_r(x)}{x^{\lfloor \alpha_0 \rfloor}} \psi(D),
\end{equation}
while for $m=1$,
\begin{equation}
\psi_1(x;D) = \chi \hbar \frac{p_0(x)}{x^{\lfloor \alpha_{r-1} \rfloor}} \frac{d}{dx} \psi(D).
\end{equation}
\end{lem}

\begin{proof}
The $\psi_r(x;D)$ is trivial since $\xi_r = 0$.

For the other case, we calculate:
\begin{multline}
p_0 (x) \hbar \frac{d}{dx} \ln \psi(D)\\
\shoveright{= p_0(x) \sum_{g,n}^\infty \frac{\hbar^{2g + n - 1}}{n!} \frac{d}{dx} \int_D \cdots \int_D \Bigl(W_{g,n}(z_1,\ldots,z_n) - \delta_{g,0}\delta_{n,2} \frac{dx(z_1) dx(z_2)}{(x(z_1) - x(z_2))^2} \Bigr)\phantom{,}}\\
= \chi \frac{p_0(x)}{dx(z)} \sum_{g,n}^\infty \frac{\hbar^{2g + n}}{n!} \int_D \cdots \int_D \Bigl(W_{g,n+1}(z,z_1,\ldots,z_n) - \delta_{g,0}\delta_{n,1} \frac{dx(z) dx(z_1)}{(x(z) - x(z_1))^2} \Bigr),
\end{multline}
where we are now integrating only over the variables $z_1, \ldots, z_n$; we are not integrating over $z$. Then, recall that
\begin{multline}
W_{g,n+1}(z,z_1,\ldots,z_n) + U^{(1)}_{g,n+1}(z; z_1, \ldots, z_n)\\
= - \frac{p_1(x)}{p_0(x)} dx(z) \delta_{g,0} \delta_{n,0} + \frac{dx(z) dx(z_1)}{(x(z) - x(z_1))^2} \delta_{g,0} \delta_{n,1}.
\end{multline}
Therefore,
\begin{align}
p_0(x) \hbar \frac{d}{dx} \ln \psi(D) &= \chi p_0(x) \xi_1(z;D) - \chi p_1(x),
\end{align}
hence
\begin{equation}
\chi p_0(x) \hbar \frac{d}{dx} \psi(D) = \psi_1(x;D) x^{\lfloor \alpha_{r-1} \rfloor}.\qedhere
\end{equation}
\end{proof}

Then, from lemma \ref{lem:eqxikD} we get

\begin{thm}\label{thm:eqpsiD}
For $k=2, \ldots, r$, we have the following system of linear differential equations:
\begin{multline}
\hbar \chi \frac{d}{dx} \left(\psi_{k-1}(x;D)\right) = \frac{x^{\lfloor \alpha_{r-k} \rfloor}}{x^{\lfloor \alpha_{r-k+1} \rfloor}}\,\psi_k(x;D) - \frac{p_{k-1}(x) x^{\lfloor \alpha_{r-1} \rfloor}}{p_0(x)x^{\lfloor\alpha_{r-k+1} \rfloor}}\, \psi_1(x;D)\\
- \hbar \sum_i \chi_i \frac{1}{x-x_i}\,\left(\psi_{k-1}(x;D) - \psi_{k-1}(x_i;D) \right).
\end{multline}
We also get
\begin{equation}
\hbar \chi \frac{d}{dx} \psi_r(x; D) =\hbar \chi \left(\frac{p_r'(x)}{p_r(x)} - \frac{\lfloor \alpha_0 \rfloor}{x} \right) \psi_r(x;D) - \frac{p_r(x) x^{\lfloor \alpha_{r-1} \rfloor}}{p_0(x) x^{\lfloor \alpha_0 \rfloor}} \psi_1(x;D).
\end{equation}

Equivalently, in matrix form:
\begin{multline*}
\hbar \chi \frac{d}{dx}
\begin{pmatrix}
\psi_1(x;D)\\ \vdots\\ \psi_{r-1}(x;D)\\ \psi_{r}(x;D)
\end{pmatrix}
=
\begin{pmatrix}
-\frac{p_{1}(x)}{p_0(x)} & \frac{x^{\lfloor \alpha_{r-2} \rfloor}}{x^{\lfloor \alpha_{r-1} \rfloor}} & & &\\
\vdots & &\ddots & &\\
- \frac{p_{r-1}(x) x^{\lfloor \alpha_{r-1} \rfloor}}{p_0(x)x^{\lfloor\alpha_{1} \rfloor}} & & & & \frac{x^{\lfloor \alpha_{0} \rfloor}}{x^{\lfloor \alpha_{1} \rfloor}}\\
- \frac{p_r(x) x^{\lfloor \alpha_{r-1} \rfloor}}{p_0(x) x^{\lfloor \alpha_0 \rfloor}} & & & & \frac{p_r'(x)}{p_r(x)} - \frac{\lfloor \alpha_0 \rfloor}{x}
\end{pmatrix}
\,
\begin{pmatrix}
\psi_1(x;D)\\ \vdots\\ \psi_{r-1}(x;D)\\ \psi_r(x;D)
\end{pmatrix}\\
- \hbar \sum_i \frac{\chi_i}{x-x_i}\,\begin{pmatrix}
\psi_1(x;D)\\ \vdots\\ \psi_{r-1}(x;D)\\ 0
\end{pmatrix} + \hbar \sum_i \frac{\chi_i}{x-x_i}\,\begin{pmatrix}
\psi_1(x_i;D)\\ \vdots\\ \psi_{r-1}(x_i;D)\\ 0
\end{pmatrix}
\end{multline*}
\end{thm}

\begin{proof}
The special case is straightforward. Since $\psi_r(x;D) = -\frac{p_r(x)}{x^{\lfloor \alpha_0 \rfloor}} \psi(D)$,
\begin{equation}
\begin{aligned}
\hbar \chi \frac{d}{dx} \psi_r(x; D) &= - \hbar \chi \frac{d}{dx} \left(\frac{p_r(x)}{x^{\lfloor \alpha_0 \rfloor}} \psi(D) \right)\\
&=- \hbar \chi \left(\frac{p_r'(x)}{x^{\lfloor \alpha_0 \rfloor}} -\lfloor \alpha_0 \rfloor \frac{p_r(x)}{x^{\lfloor \alpha_0 \rfloor+1}} \right) \psi(D) - \hbar \chi \frac{p_r(x)}{x^{\lfloor \alpha_0 \rfloor}} \frac{d}{dx} \psi(D)\\
&= \hbar \chi \left(\frac{p_r'(x)}{p_r(x)} - \frac{\lfloor \alpha_0 \rfloor}{x} \right) \psi_r(x;D) - \frac{p_r(x) x^{\lfloor \alpha_{r-1} \rfloor}}{p_0(x) x^{\lfloor \alpha_0 \rfloor}} \psi_1(x;D),
\end{aligned}
\end{equation}
where we used Lemma \ref{l:xi1}.

To prove cases $k=2,\ldots,r$, we start with \eqref{eq:xi} and multiply by $\psi(D)$ to get:
\begin{multline}
\frac{x^{\lfloor \alpha_{r-k} \rfloor}}{x^{\lfloor \alpha_{r-k+1} \rfloor}} \psi_k(x;D)= \frac{x^{\lfloor \alpha_{r-1} \rfloor}}{x^{\lfloor \alpha_{r-k+1} \rfloor}} \psi_1(x;D) \xi_{k-1}(x;D) \\
\shoveright{+ \hbar \sum_i \chi_i \frac{1}{x-x_i} \left(\psi_{k-1}(x;D) - \psi_{k-1}(x_i;D) \right)}\\
+ \hbar \chi \psi(D) \frac{d}{dx} \Bigl(\frac{p_0(x)}{x^{\lfloor\alpha_{r-k+1} \rfloor}} \xi_{k-1} (x; D) -\frac{p_{k-1}(x)}{x^{\lfloor\alpha_{r-k+1} \rfloor}} \Bigr).
\end{multline}
We rewrite the last term as
\begin{equation}
\hbar \chi \frac{d}{dx} \left(\psi_{k-1}(x;D)\right) - \hbar \frac{\chi}{x^{\lfloor\alpha_{r-k+1} \rfloor}} \left(p_0(x) \xi_{k-1}(x;D) - p_{k-1}(x) \right) \frac{d}{dx} \psi(D).
\end{equation}
Using Lemma \ref{l:xi1}, this is equal to
\begin{equation}
\hbar \chi \frac{d}{dx} \left(\psi_{k-1}(x;D)\right) - \frac{x^{\lfloor \alpha_{r-1} \rfloor}}{x^{\lfloor\alpha_{r-k+1} \rfloor}} \Bigl(\xi_{k-1}(x;D) - \frac{p_{k-1}(x)}{p_0(x)} \Bigr) \psi_1(x;D).
\end{equation}
Putting everything together, we end up with the differential equation
\begin{multline}
\frac{x^{\lfloor \alpha_{r-k} \rfloor}}{x^{\lfloor \alpha_{r-k+1} \rfloor}} \psi_k(x;D) = \hbar \chi \frac{d}{dx} \left(\psi_{k-1}(x;D)\right) + \frac{p_{k-1}(x) x^{\lfloor \alpha_{r-1} \rfloor}}{p_0(x)x^{\lfloor\alpha_{r-k+1} \rfloor}} \psi_1(x;D)\\
+ \hbar \sum_i \chi_i \frac{1}{x-x_i} \left(\psi_{k-1}(x;D) - \psi_{k-1}(x_i;D) \right).\qedhere
\end{multline}
\end{proof}

In the next subsection, we show that for special choices of integration divisors, this system of differential equations can be turned into an order $r$ differential equation for $\psi$ which is a quantization of the original spectral curve. The key point to notice here is that the quantum curve strongly depends on the choice of integration divisor; different choices of integration divisor give rise to different quantum curves.

\subsection{Special choices of integration divisors}

In this section we study special choices of integration divisors that simplify the system of linear differential equations obtained in Theorem \ref{thm:eqpsiD}.

\subsubsection{Poles of $x$}

The most interesting choice of integration divisor is
\begin{equation}
D = [z] - [\beta],
\end{equation}
that is $\chi=1$, with $\beta$ a simple pole of $x$ (then it is not in $R$).

\begin{rem}
It is also possible to integrate with base point $\beta$ a pole of $x$ of order more than one. But then, $\beta \in R$, hence one needs to check that the integrals converge, since the correlation functions can have poles at $\beta$. We however can do that explicitly in some examples, as we will do in the next section.
\end{rem}

In this case Theorem \ref{thm:eqpsiD} reduces to (for $k=2,\ldots,r$)
\begin{multline}
\hbar \frac{d}{dx} \left(\psi_{k-1}(x;D)\right) = \frac{x^{\lfloor \alpha_{r-k} \rfloor}}{x^{\lfloor \alpha_{r-k+1} \rfloor}} \psi_k(x;D) - \frac{p_{k-1}(x) x^{\lfloor \alpha_{r-1} \rfloor}}{p_0(x)x^{\lfloor\alpha_{r-k+1} \rfloor}} \psi_1(x;D)\\
+ \hbar \lim_{z_1 \to \beta} \frac{1}{x_1(z_1)} \psi_{k-1}(x_1(z_1);D).
\end{multline}

We can say more about the limits. We can prove the following lemma:
\begin{lem}
For $\beta$ a simple pole of $x$ (or a pole of $x$ where the correlation functions are all holomorphic), and $k=1,\ldots,r-1$,
\begin{equation}
\lim_{z_1 \to \beta} \frac{\psi_{k}(x(z_1);D)}{x(z_1)}= \psi(D) \lim_{z_1 \to \beta} \frac{P_{k+1}(x(z_1), y(z_1))}{x(z_1)^{\lfloor \alpha_{r-k} \rfloor+1}},
\end{equation}
where the $P_{k+1}(x,y)$ were defined in \eqref{eq:Pm}. In particular, because of the admissibility condition in Definition \ref{d:admissible}, these limits are finite.
\end{lem}

\begin{proof}
First, we note that
\begin{equation}
\psi_k(x(z_1);D) = \psi(D) \frac{1}{x(z_1)^{\lfloor \alpha_{r-k} \rfloor}}\left(p_0(x(z_1) \xi_k(x_1(z_1);D) - p_k(x(z_1)) \right).
\end{equation}
As $z_1 \to \beta$, since $\beta$ is a pole of $x$ and the correlation functions are holomorphic at $\beta$, it follows that
\begin{equation}
\begin{aligned}
&\hspace*{-1cm}\lim_{z_1 \to \beta} \frac{\psi_{k}(x(z_1);D)}{x(z_1)}\\[-8pt]
&\hspace*{1cm}=\psi(D) \lim_{z_1 \to \beta} \frac{1}{x(z_1)^{\lfloor \alpha_{r-k} \rfloor+1}} \Bigl((-1)^k \frac{U_{0,1}^{(k)}(z_1) p_0(x(z_1))}{dx(z_1)^k} - p_k(x(z_1)) \Bigr)\\
&\hspace*{1cm}= \psi(D) \lim_{z_1 \to \beta} \frac{P_{k+1}(x(z_1), y(z_1))}{x(z_1)^{\lfloor \alpha_{r-k} \rfloor+1}}.
\end{aligned}
\vspace*{-1.7\baselineskip}
\end{equation}
\end{proof}

\medskip
These limits are guaranteed to be finite for admissible spectral curves, by the inequality \eqref{eq:limitex}. Let us denote
\begin{equation}
C_{k} =\lim_{z_1 \to \beta} \frac{P_{k+1}(x(z_1), y(z_1))}{x(z_1)^{\lfloor \alpha_{r-k} \rfloor+1}}, \quad k=1,\ldots,r-1.
\end{equation}
Then
\begin{equation}
\lim_{z_1 \to \beta} \frac{\psi_{k}(x(z_1);D)}{x(z_1)}= C_k \psi(D) = - \frac{x^{\lfloor \alpha_0 \rfloor}}{p_r(x)} C_k \psi_r(x;D),
\end{equation}
and the system of differential equations can be rewritten as
\begin{multline}
\hbar \frac{d}{dx} \left(\psi_{k-1}(x;D)\right) = \frac{x^{\lfloor \alpha_{r-k} \rfloor}}{x^{\lfloor \alpha_{r-k+1} \rfloor}} \psi_k(x;D) - \frac{p_{k-1}(x) x^{\lfloor \alpha_{r-1} \rfloor}}{p_0(x)x^{\lfloor\alpha_{r-k+1} \rfloor}} \psi_1(x;D)\\
- \hbar \frac{x^{\lfloor \alpha_0 \rfloor}}{p_r(x)} C_{k-1} \psi_r(x;D).
\label{e:start}
\end{multline}

We can then rewrite this system of differential equations as an order $r$ differential equation for $\psi(D)$. To this end, we define the differential operators
\begin{equation}
D_i =\hbar \frac{x^{\lfloor \alpha_{i} \rfloor}}{x^{\lfloor \alpha_{i-1} \rfloor}}\frac{d}{dx}, \quad i=1,\ldots,r.
\end{equation}
Then we get the following result:
\begin{lem}\label{l:QC}
The system of differential equations is equivalent to the following order $r$ differential equation for $\psi(D)$:
\begin{multline}
\Bigl[ D_1 D_2 \cdots D_{r-1} \frac{p_{0}(x)}{x^{\lfloor \alpha_r \rfloor}} D_r + D_1 D_2 \cdots D_{r-2} \frac{p_{1}(x)}{x^{\lfloor \alpha_{r-1} \rfloor}} D_{r-1}\\
+\cdots + \frac{p_{r-1}(x)}{x^{\lfloor \alpha_{1} \rfloor}} D_1 + \frac{p_r(x)}{x^{\lfloor \alpha_0 \rfloor}}
- \hbar C_1 D_1 D_2 \cdots D_{r-2} \frac{x^{\lfloor \alpha_{r-1} \rfloor}}{x^{\lfloor \alpha_{r-2} \rfloor}}\\
- \hbar C_2 D_1 D_2 \cdots D_{r-3} \frac{x^{\lfloor \alpha_{r-2} \rfloor}}{x^{\lfloor \alpha_{r-3} \rfloor}}- \cdots - \hbar C_{r-1} \frac{x^{\lfloor \alpha_{1} \rfloor}}{x^{\lfloor \alpha_{0} \rfloor}} \Bigl] \psi(D) = 0.
\end{multline}
After normal ordering, it is easy to see that this is a quantization of the original spectral curve, according to Definition \ref{d:QC}. Moreover, it has only a finite number of $\hbar$ corrections, hence it is simple.
\end{lem}

\begin{proof}
First, we notice that we can rewrite \eqref{e:start} as
\begin{equation}\label{eq:stepone}
\psi_k(x;D) = D_{r-k+1} \psi_{k-1}(x;D) + \frac{p_{k-1}(x)}{x^{\lfloor\alpha_{r-k+1} \rfloor}} D_{r-k+1}\psi(D)- \hbar \frac{x^{\lfloor \alpha_{r-k+1} \rfloor}}{x^{\lfloor \alpha_{r-k} \rfloor}} C_{k-1} \psi(D),
\end{equation}
where we used the fact that
\begin{equation}
\psi_r(x;D) = - \frac{p_r(x)}{x^{\lfloor \alpha_0 \rfloor}} \psi(D) \quad \text{and} \quad
\psi_1(x;D) = \frac{p_0(x)}{x^{\lfloor \alpha_{r-1} \rfloor}} \hbar \frac{d}{d x} \psi(D).
\end{equation}
We start with \eqref{e:start} for $k= r$. We get:
\begin{equation}\label{eq:steptwo}
\psi_r(x;D) = D_{1} \psi_{r-1}(x;D) + \frac{p_{r-1}(x)}{x^{\lfloor\alpha_{1} \rfloor}}D_1 \psi(D)- \hbar \frac{x^{\lfloor \alpha_{1} \rfloor}}{x^{\lfloor \alpha_{0} \rfloor}} C_{r-1} \psi(D).
\end{equation}
At $k=r-1$, we get
\begin{equation}
\psi_{r-1}(x;D) = D_{2} \psi_{r-2}(x;D) + \frac{p_{r-2}(x)}{x^{\lfloor\alpha_{2} \rfloor}} D_2 \psi(D)- \hbar \frac{x^{\lfloor \alpha_{2} \rfloor}}{x^{\lfloor \alpha_{1} \rfloor}} C_{r-2} \psi(D).
\end{equation}
Substituting back in \eqref{eq:steptwo}, we get
\begin{multline}
\psi_r(x;D) = D_1 D_{2} \psi_{r-2}(x;D) + D_1 \frac{p_{r-2}(x)}{x^{\lfloor\alpha_{2} \rfloor}} D_2 \psi(D)+ \frac{p_{r-1}(x)}{x^{\lfloor\alpha_{1} \rfloor}} D_1 \psi(D)\\
- \hbar D_1 \frac{x^{\lfloor \alpha_{2} \rfloor}}{x^{\lfloor \alpha_{1} \rfloor}} C_{r-2} \psi(D) - \hbar \frac{x^{\lfloor \alpha_{1} \rfloor}}{x^{\lfloor \alpha_{0} \rfloor}} C_{r-1} \psi(D).
\end{multline}
We keep going like this until we can rewrite all terms in terms of $\psi(D)$, and we end up with the statement of the Lemma.
\end{proof}

In the special case where $\lfloor \alpha_{i} \rfloor = 0$, $i=0, \ldots, r$, the differential equation simplifies. In this case, $D_i = \hbar \frac{d}{dx}$ for all $i=1,\ldots,r$, hence the differential equation becomes
\begin{multline}
\Bigl[ \hbar^r \frac{d^{r-1}}{dx^{r-1}} p_0(x) \frac{d}{dx} + \hbar^{r-1} \frac{d^{r-2}}{dx^{r-2}}p_1(x) \frac{d}{dx} +\cdots +\hbar p_{r-1}(x) \frac{d}{dx} + p_r(x)\\
- \hbar^{r-1} C_1 \frac{d^{r-2}}{dx^{r-2}} - \hbar^{r-2} C_2 \frac{d^{r-3}}{dx^{r-3}}- \cdots - \hbar C_{r-1} \Bigr] \psi(D) = 0.
\end{multline}
It is even simpler when $C_k = 0$, $k=1,\ldots,r-1$, in which case it becomes
\begin{multline}
\Bigl[ \hbar^r \frac{d^{r-1}}{dx^{r-1}} p_0(x) \frac{d}{dx} + \hbar^{r-1} \frac{d^{r-2}}{dx^{r-2}}p_1(x) \frac{d}{dx}\\[-5pt]
+\cdots +\hbar p_{r-1}(x) \frac{d}{dx} + p_r(x) \Bigr] \psi(D) = 0,
\end{multline}
which is the quantization of the spectral curve obtained through
\begin{equation}
(x,y) \mto \left(\hat{x}, \hat{y} \right) = \bigl(x, \hbar \frac{d}{dx} \bigr),
\end{equation}
and for the following choice of ordering:
\begin{equation}
\left(\hat{y}^{r-1} p_0(\hat{x}) \hat{y} + \hat y^{r-2} p_1(\hat x) \hat{y} +\cdots + p_{r-1}(\hat x) \hat{y} + p_r(\hat x)\right) \psi=0.
\end{equation}

\begin{rem}\label{r:interpol}
We note here that we can generalize this choice of integration divisor slightly. If $x$ has more than one poles, say $\beta_i$, $i=,\ldots,n$, then we can take the integration divisor to be
\begin{equation}
D = [z] - \sum_{i=1}^n \mu_i [ \beta_i],
\end{equation}
for some constants $\mu_1, \ldots, \mu_n$ satisfying $\sum_{i=1}^n \mu_i = 1$. We can go through the same steps as above to get the quantum curve. If we define the constants
\begin{equation}
C_k^{(i)} = \lim_{z_1 \to \beta_i} \frac{P_{k+1}(x(z_1), y(z_1))}{x(z_1)^{\lfloor \alpha_{r-k} \rfloor+1}}, \quad k=1,\ldots,r-1, \quad i = 1, \ldots, n,
\end{equation}
then we get the differential equation in Lemma \ref{l:QC}, but with $C_k$ replaced by $\sum_{i=1}^n \mu_i C_k^{(i)}$. In a sense, this generalized choice of integration divisor interpolates between the different quantizations corresponding to the poles of $x$.
\end{rem}

\subsubsection{Other choices of integration divisors}

For some spectral curves, it may also be interesting to choose an integration divisor of the form
$$
D = [z] - [\gamma],
$$
where $\gamma$ is not a pole of $x$.

For instance, we could choose $\gamma$ to be zero of $p_0(x)$ (if it is not in $R$, so that the integrals converge). In this case, the system becomes
\begin{multline}
\hbar \frac{d}{dx} \left(\psi_{k-1}(x;D)\right) = \frac{x^{\lfloor \alpha_{r-k} \rfloor}}{x^{\lfloor \alpha_{r-k+1} \rfloor}} \psi_k(x;D) - \frac{p_{k-1}(x) x^{\lfloor \alpha_{r-1} \rfloor}}{p_0(x)x^{\lfloor\alpha_{r-k+1} \rfloor}} \psi_1(x;D)\\
+ \hbar \frac{1}{x-x(\gamma)} \bigl(\psi_{k-1}(x;D) - \psi_{k-1}(x(\gamma);D) \bigr).
\end{multline}
To proceed further, we need to evaluate the objects $ \psi_k(x(\gamma);D)$. Recall that
\begin{equation}
\psi_k(x(\gamma); D) = \psi(D) \lim_{z \to \gamma} \Bigl(\frac{1}{x(z)^{\lfloor \alpha_{r-k} \rfloor}} \left(p_0(x(z)) \xi_k(x(z);D) - p_k(x(z)) \right) \Bigr).
\end{equation}
Evaluating these limits may be quite difficult in general. But if $r=2$, then we only need the case $k=1$, which we can evaluate explicitly since
\begin{equation}
\begin{aligned}
\psi_1(x(\gamma); D) &= \psi(D) \lim_{z \to \gamma} \Bigl(\frac{1}{x(z)^{\lfloor \alpha_{1} \rfloor}} \left(p_0(x(z)) \xi_1(x(z);D) - p_1(x(z)) \right) \Bigr)\\
&= \psi(D) \lim_{z \to \gamma} \Bigl(\frac{1}{x(z)^{\lfloor \alpha_{1} \rfloor}}p_0(x(z)) y(z) \Bigr).
\end{aligned}
\end{equation}

We can then evaluate this limit and obtain a quantum curve in this particular case.

\section{Some examples}
\label{s:examples}

In this section we consider many examples of quantum curves. We note that all of these examples have also been verified on Mathematica for small order in $\hbar$.

\subsection{$y^a - x = 0$}

Let us start by considering the spectral curve $y^a - x = 0$ for $ a \geq 2$. This spectral curve arises in the study of intersection numbers of the moduli space of curves with $a$-spin structure. We will come back to this example in the next section to explore its enumerative meaning. Here we only want to apply our Theorem to find the quantum curve.

A parameterization of the curve is $(x,y) = (z^a, z)$. $x$ has a pole at $z=\infty$, which is in $R$. But it is easy to show that the correlation functions do not have poles at $z=\infty$. Therefore we can choose the integration divisor to be $D = [z]-[\infty]$.

The Newton polygon for this curve is a line segment joining $(1,0)$ and $(0,a)$. Then $\lfloor \alpha_i \rfloor = 0$ for $i=1, \ldots, a$, and $\lfloor \alpha_0 \rfloor = 1$.

In this case, we need to evaluate the constants $C_{k}$, $k=1,\ldots,a-1$. We get:
\begin{equation}
\begin{aligned}
C_{k} &= \lim_{z_1 \to \infty} \frac{P_{k+1}(x(z_1), y(z_1))}{x(z_1)}\\
&= \lim_{z_1 \to \infty} \frac{y(z_1)^k}{x(z_1)}= \lim_{z_1 \to \infty} \frac{z_1^k}{z_1^a}= 0.
\end{aligned}
\end{equation}
Therefore, the quantum curve is
\begin{equation}
\Bigl(\hbar^a \frac{d^a}{d x^a} - x \Bigr) \psi = 0.
\end{equation}
This is the straightforward quantization of the spectral curve $y^a - x = 0$ through
\begin{equation}
(x,y) \mto \left(\hat{x}, \hat{y} \right) = \bigl(x, \hbar \frac{d}{dx} \bigr).
\end{equation}

To help the reader's understanding of Section \ref{s:pole}, let us write down explicitly some of the objects under study for $a=3$ and small $g,n$. We use the parameterization $(x,y) = (z^3, z)$. Then $\tau(z) = \{z, \rho z, \rho^2 z \},$ where $\rho = e^{2\pi i/3}$.
The level $2g-2+n = 1$ meromorphic differentials constructed from the topological recursion are:
\begin{align*}
W_{0,3}(z_1,z_2,z_3) &= d z_1 d z_2 d z_3 \Bigl(\frac{2}{3 z_1^3 z_3^2 z_2^2}+\frac{2}{3 z_1^2 z_3^3
z_2^2}+\frac{2}{3 z_1^2 z_3^2 z_2^3} \Bigr),\\
W_{1,1}(z_1)&= d z_1 \Bigl(\frac{1}{9 z_1^5} \Bigr).
\end{align*}

The most important result of section 4 is Theorem \ref{t:pole}. Let us check that it is satisfied for $(g,n)=(1,0)$ and $(g,n)=(0,2)$. First, for $(g,n)=(1,0)$, we calculate that
\begin{equation}
Q_{1,1}^{(1)}(z) = Q_{1,1}^{(2)}(z)=Q_{1,1}^{(3)}(z)=0.
\end{equation}
This is exactly what is required by Theorem \ref{t:pole}, since $n=0$ and the right-hand-side of \eqref{eq:tpole} vanishes. Note that the statement that $Q_{1,1}^{(1)}(z) = 0$ is also expected from Lemma \ref{l:P1}.

Now for $(g,n)=(0,2)$. First, we get that
\begin{equation}
Q_{0,3}^{(1)} (z; z_2,z_3) = 0,
\end{equation}
as expected from both Lemma \ref{l:P1} and Theorem \ref{t:pole}. Then, we get
\begin{multline}
\frac{Q_{0,3}^{(2)} (z; z_2,z_3)}{dx(z)^2} = \frac{dz_2 dz_3}{3 z_2^3 z_3^3 (x(z) - z_2^3)^2 (x(z)-z_3^3)^2}\\
{}\cdot\Bigl(2 x(z)^3 (z_2+z_3)- 4 x(z)^2 (z_2^4 + z_2^3 z_3 + z_2 z_3^3 + z_3^4)\\
\hspace*{3cm}+ 2 x(z) (z_2^7 + z_2^6 z_3 + 3 z_2^4 z_3^3 +3 z_2^3 z_3^4 + z_2 z_3^6 + z_3^7)\\
- z_2^3 z_3^3 (5 z_2^4 + 8 z_2^3 z_3 - 18 z_2^2 z_3^2 + 8 z_2 z_3^3 + 5 z_3^4) \Bigl).
\end{multline}
First, we see that it is indeed a function of $x(z)$, as expected. Then, it is easy to check that it satisfies Theorem \ref{t:pole}, using the fact that
\begin{equation}
U_{0,2}^{(1)}(z;z_2) = - dz dz_2\, \frac{z^2 + 4 z z_2 + z_2^2}{(z^2 + z z_2 + z_2^2)^2}.
\end{equation}

We can also check that Theorem \ref{t:pole} holds for $k=3$. We get:
\begin{multline}
\frac{Q_{0,3}^{(3)} (z; z_2,z_3)}{dx(z)^3} = -\frac{dz_2 dz_3}{3 z_2^2 z_3^2 (x(z) - z_2^3)^2 (x(z)-z_3^3)^2}\\
{}\cdot\Bigl(2 x(z)^3 - 4 x(z)^2 (z_2^3 + z_3^3)+ x(z) (z_2^2+z_3^2)(2 z_2^4 + z_2^2 z_3^2 + 2 z_3^4)\\
- 2 z_2^3 z_3^3 (z_2+z_3)(4 z_2^2 - 7 z_2 z_3 + 4 z_3^2)\Bigr),
\end{multline}
which is again a function of $x(z)$. It satisfies Theorem \ref{t:pole}, since
\begin{equation}
U_{0,2}^{(2)}(z; z_2) = - dz dz_2\, \frac{3 z^3 (z^2 - 2 z z_2 - 2 z_2^2)}{(z^2 + z z_2 + z_2^2)^2}.
\end{equation}

With these expressions we can also calculate the first few orders of the wave-function. With the integration divisor $D = [z]-[\infty]$, we get
\begin{equation}
\psi (z) = \exp\Bigl(\frac{1}{\hbar} S_0(z) + S_1(z) + \hbar S_2(z) + \mathcal{O}(\hbar)^2 \Bigr),
\end{equation}
with
\begin{align}
S_0(z) &= \int_a^z W_{0,1}(z) = \frac{3}{4} z^4 - \frac{3}{4} a^4,\\
S_1(z) &= \frac{1}{2} \int_a^z \int_a^z \Bigl(W_{0,2}(z_1,z_2) - \frac{dx(z_1) dx(z_2)}{(x(z_1)- x(z_2))^2} \Bigr)\\
&= \log(a^2 + a z + z^2) - \frac{1}{2} (\log (3 a^2) + \log(3 z^2)),\notag\\
S_2(z) &= \int_\infty^z W_{1,1}(z) + \frac{1}{3!} \int_\infty^z \int_\infty^z \int_\infty^z W_{0,3}(z_1, z_2, z_3)= \frac{7}{36} z^4.
\end{align}
Here we regularized the first two integrals, but we will send $a \to \infty$ shortly. Then it is easy to check that
\begin{equation*}
\psi(z)^{-1} \Bigl(\hbar^3 \frac{d^3}{dx(z)^3} - x(z) \Bigr) \psi(z) = \hbar\, \frac{a+2z}{a^2 + a z + z^2} - \hbar^2\, \frac{3 a + 4 z}{3 z^4(a^2+a z+z^2)} + \mathcal{O}(\hbar)^3,
\end{equation*}
which indeed goes to zero as $a \to \infty$, as expected.

\subsection{$y^a - x y + 1 = 0$}
Let us now consider the spectral curve $y^a - x y + 1 = 0$, for $a \geq 2$. Its Newton polygon is shown in Figure \ref{f:newex1} for $a=4$. From the Newton polygon we see that $\lfloor \alpha_i \rfloor = 0$ for $i=0, \ldots,a$.

\begin{figure}[htb]
\begin{center}
\includegraphics[width=0.08\textwidth]{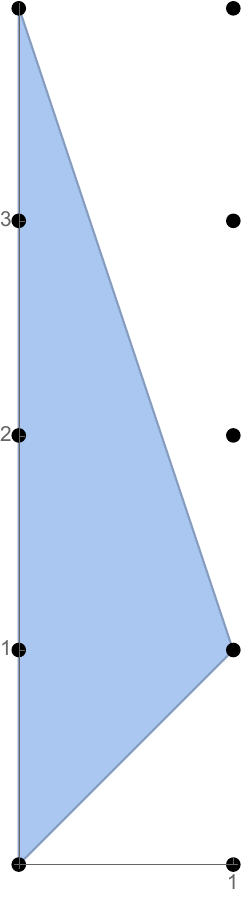}
\caption{The Newton polygon of the curve $y^4 - x y + 1 = 0$.}
\label{f:newex1}
\end{center}
\end{figure}

The case $a=2$ arises in the study of enumeration of ribbon graphs \cite{MS,NS}, while the $a > 2$ case arises in the enumeration of $a$-hypermaps \cite{DM, DBOSS}.

\subsubsection{$a=2$}
Let us first consider the case $a=2$. A parameterization is $(x,y) = \left(z + \sfrac{1}{z}, z \right)$. Then $R$ consists of the two points $z = \pm 1$.

In this case $x$ has two simple poles, at $z=0$ and $z=\infty$, that we can use for integration divisor.

\subsubsection*{Pole at $z=\infty$}
We choose the integration divisor $D = [z] - [\infty]$. We calculate:
\begin{equation}
\begin{aligned}
C_1 &= \lim_{z_1 \to \infty} \frac{p_0(z_1) y(z_1)}{x_1(z_1)}\\
&= \lim_{z_1 \to \infty} \frac{z_1^2}{z_1^2+1}= 1.
\end{aligned}
\end{equation}
Then the quantum curve is
\begin{equation}
\Bigl(\hbar^2 \frac{d^2}{dx^2} - \hbar x \frac{d}{dx} + 1 -\hbar \Bigr) \psi = 0,
\end{equation}
which is equivalent to
\begin{equation}
\Bigl(\hbar^2 \frac{d^2}{dx^2} - \hbar \frac{d}{dx} x + 1 \Bigr) \psi = 0,
\end{equation}
This is the quantization
\begin{equation}
(x,y) \mto \left(\hat{x}, \hat{y} \right) = \bigl(x, \hbar \frac{d}{dx} \bigr),
\end{equation}
with choice of ordering
\begin{equation}
\left(\hat{y}^2 - \hat{y} \hat{x} + 1 \right) \psi = 0.
\end{equation}

\subsubsection*{Pole at $z=0$}
Now suppose that we consider instead the integration divisor $D = [z] - [0]$. In this case,
\begin{equation}
C_1 = \lim_{z_1 \to 0} \frac{z_1^2}{z_1^2 + 1}= 0.
\end{equation}
The quantum curve is then
\begin{equation}
\Bigl(\hbar^2 \frac{d^2}{dx^2} - \hbar x \frac{d}{dx} + 1 \Bigr) \psi = 0,
\end{equation}
which corresponds to the other choice of ordering
\begin{equation}
\left(\hat{y}^2 - \hat{x}\hat{y} + 1 \right) \psi = 0.
\end{equation}
This is the quantum curve that was obtained in \cite{MS}.

\begin{rem}
It is interesting to note that by looking at the two integration divisors of the form $D = [z] - [\beta]$ where $\beta$ is a simple pole of $x$, we recover the two possible choices of ordering for the quantum curve.
\end{rem}

\begin{rem}
As an example of Remark \ref{r:interpol}, we note here that if we choose the integration divisor $D = [z] - \mu [\infty] - (1-\mu) [ 0]$, then we obtain the quantum curve
\begin{equation}
\Bigl(\hbar^2 \frac{d^2}{dx^2} - \hbar x \frac{d}{dx} + 1 - \mu\hbar \Bigr) \psi = 0.
\end{equation}
Different choices of constant $\mu$ interpolate between the two quantizations corresponding to the poles at $0$ and $\infty$.
\end{rem}

\subsubsection{$a>2$}

We can do similar calculations for $a>2$. In this case, a parameterization is $(x,y) = \left(\psfrac{z^a+1}{z}, z \right)$, and $R$ consists of the $a$ solutions of $z^a = \spfrac{1}{a-1}$, plus the point at infinity.

$x$ has a simple pole at $z=0$ and a pole of order $a-1$ at $z=\infty$. The latter is in~$R$, but it is easy to see that the correlation functions have no poles at $z=\infty$. Thus the integrals converge.

\subsubsection*{Pole at $z=\infty$}

We first consider the integration divisor $D =[z]-[\infty]$.

We need to evaluate the limits
\begin{equation}
C_{k} = \lim_{z_1 \to \infty} \sum_{i=1}^k\frac{p_{k-i}(z_1) y(z_1)^i}{x_1(z_1)}, \quad k=1, \ldots, a-1.
\end{equation}
Since $p_0(z_1) = 1$, $p_{a-1}(z_1) = - x(z_1)$, $p_a(z_1) = 1$, with all other $p_i$'s equal to zero, we get that for $k<a-1$,
\begin{equation}
C_{k}= \lim_{z_1 \to \infty} \frac{y(z_1)^k}{x_1(z_1)}= \lim_{z_1 \to \infty} \frac{z_1^{-a+k+1}}{z_1^{-a}+1}= 0,
\end{equation}
since $ k < a-1$. As for the $k=a-1$ case,
\begin{equation}
C_{a-1}= \lim_{z_1 \to \infty}\frac{y(z_1)^{a-1}}{x_1(z_1)}= \lim_{z_1 \to \infty} \Bigl(\frac{1}{z_1^{-a}+1} \Bigr)= 1.
\end{equation}
Therefore, the quantum curve is
\begin{equation}
\Bigl(\hbar^a \frac{d^a}{dx^a} - x \hbar \frac{d}{dx} + 1 - \hbar \Bigr) \psi = 0,
\end{equation}
which is equivalent to
\begin{equation}
\Bigl(\hbar^a \frac{d^a}{dx^a} - \hbar \frac{d}{dx} x + 1 \Bigr) \psi = 0.
\end{equation}
This is the quantization
\begin{equation}
(x,y) \mto \left(\hat{x}, \hat{y} \right) = \bigl(x, \hbar \frac{d}{dx} \bigr),
\end{equation}
with the choice of ordering
\begin{equation}
\left(\hat{y}^a - \hat{y}\hat{x} + 1 \right) \psi = 0.
\end{equation}

\begin{rem}
We note that this includes, as a special case when $a=3$, the quantum curve obtained by Safnuk in \cite{Safnuk} -- see also \cite{Al} -- in the context of intersection numbers on the moduli space of open Riemann surfaces. In this paper, he obtained the quantum curve
\begin{equation}
\Bigl(\hbar^3 \frac{d^3}{dx^3} - 2 x \hbar \frac{d}{dx} + 2 \hbar(Q - 1) \Bigr) \psi = 0,
\end{equation}
which can be rewritten as
\begin{equation}
\Bigl(\hbar^3 \frac{d^3}{dx^3} - 2 \hbar \frac{d}{dx} x + 2 \hbar Q \Bigr) \psi = 0.
\end{equation}
If we let $t = \hbar Q$ be a 't Hooft parameter, that is, it is kept finite as $\hbar$ goes to zero, then the above curve is a quantization of the classical spectral curve
\begin{equation}
y^3 - 2 x y+ 2 t = 0.
\end{equation}
What we have shown is that the standard topological recursion applied to this spectral curve reconstructs the WKB expansion of the above quantum curve, for the choice of integration divisor given by $D = [z] - [\infty]$.

In \cite{Safnuk}, the point of view taken is however different; $t$ is not introduced, and the spectral curve is then taken to be the reducible curve $y^3 - 2 x y = 0$. But then the topological recursion needs to be modified to reconstruct the WKB expansion of the differential operator. What we have shown is that this should be equivalent to using the standard topological recursion, but for the spectral curve with 't Hooft parameter $t= \hbar Q$ kept finite.
\end{rem}

\subsubsection*{Pole at $z=0$}

Let us now consider the integration divisor $D = [z]-[0]$.

We need to evaluate the limits
\begin{equation}
C_{k} = \lim_{z_1 \to 0} \sum_{i=1}^k \frac{p_{k-i}(z_1) y(z_1)^i}{x_1(z_1)}, \quad k=1, \ldots, a-1.
\end{equation}
Since $p_0(z_1) = 1$, $p_{a-1}(z_1) = - x(z_1)$, $p_a(z_1) = 1$, with all other $p_i$'s equal to zero, we get that for $k <a-1$,
\begin{equation}
C_{k}= \lim_{z_1 \to 0} \frac{y(z_1)^k}{x_1(z_1)}=\lim_{z_1 \to \infty} \frac{z_1^{k+1}}{z_1^{a}+1}= 0,
\end{equation}
since $ k < a-1$. As for the $k=a-1$ case,
\begin{equation}
C_{a-1}= \lim_{z_1 \to 0}\frac{y(z_1)^{a-1}}{x_1(z_1)}=\lim_{z_1 \to 0} \Bigl(\frac{z_1^a}{z_1^{a}+1} \Bigl)= 0.
\end{equation}
Therefore, the quantum curve is
\begin{equation}
\Bigl(\hbar^a \frac{d^a}{dx^a} - x \hbar \frac{d}{dx} + 1 \Bigr) \psi = 0,
\end{equation}
which corresponds to the other choice of ordering
\begin{equation}
\left(\hat{y}^a - \hat{x}\hat{y} + 1 \right) \psi = 0.
\end{equation}

Again, choosing the more general integration divisor $D = [z] - \mu [\infty]-(1-\mu)[0]$ gives the family of quantum curves
\begin{equation}
\Bigl(\hbar^a \frac{d^a}{dx^a} - x \hbar \frac{d}{dx} + 1 - \mu \hbar \Bigr) \psi = 0,
\end{equation}
which interpolates between the two quantizations.

\subsection{$x y^a + y + 1 = 0$}
The case $a=2$ arises in the study of monotone Hurwitz numbers \cite{DDM}. The Newton polygon for the case $a=4$ is shown in Figure \ref{f:newex2}. We sill study the general case $a \geq 2$, although we are not aware of an enumerative interpretation for $a>2$.

\begin{figure}[htb]
\begin{center}
\includegraphics[width=0.1\textwidth]{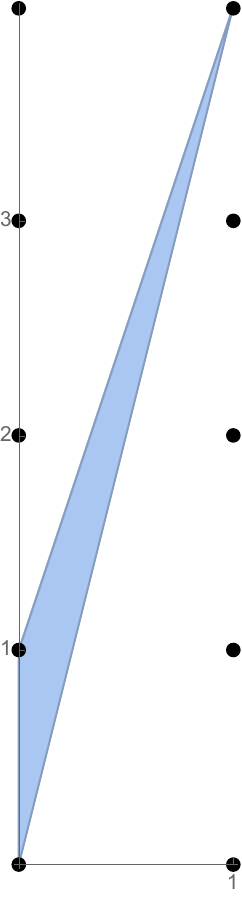}
\caption{The Newton polygon of the curve $x y^4 + y + 1 = 0$.}
\label{f:newex2}
\end{center}
\end{figure}

From the Newton polygon we see that $\lfloor \alpha_i \rfloor = 0$ for $i=0, \ldots, a-1$, while $\lfloor \alpha_a \rfloor = 1$.

\subsubsection{$a=2$}

Let us start with $a=2$. In this case, a parameterization is $(x,y) = \left(- \psfrac{z+1}{z^2}, z \right)$. $R$ consists of the zero of $dx$ at $z=-2$ and the double pole of $x$ at $z=0$.

We can use the double pole at $z=0$ for integration divisor. We can also use the zeros of $p_0(x) = x$, which are at $z=-1$ and $z=\infty$, since those are not in $R$.

\subsubsection*{Pole at $z=0$}
The only pole of $x$ is the double pole at $z=0$, which is in $R$. However, it is easy to see that the correlation functions have no poles at $z=0$, so we can choose the integration divisor $D = [z] - [0]$.\footnote{In fact this is not quite true here. All integrals converge, except as usual the integral of $W_{0,1}$, and, more importantly, the integral $\int_D \int_D \left(W_{0,2}(z_1,z_2) - \frac{dx_1 dx_2}{(x_1 - x_2)^2} \right)$. The latter diverges, hence it needs to be regularized by letting the divisor be $D = [z]-[a]$ for that particular integral. However, its derivatives are finite at $a=0$, hence we can take the limit $a \to 0$ for the WKB expansion.}

Then\vspace*{-3pt}
\begin{equation}
C_1= \lim_{z_1 \to 0} \frac{p_0(z_1) y(z_1)}{x_1(z_1)}= \lim_{z_1 \to 0}z_1= 0.
\end{equation}
The quantum curve is\vspace*{-3pt}
\begin{equation}
\Bigl(\hbar^2 \frac{d}{dx} x \frac{d}{dx} + \hbar \frac{d}{dx} + 1 \Bigl) \psi = 0.
\end{equation}
This is the quantization\vspace*{-3pt}
\begin{equation}
(x,y) \mto \left(\hat{x}, \hat{y} \right) = \bigl(x, \hbar \frac{d}{dx} \bigr),
\end{equation}
with choice of ordering\vspace*{-3pt}
\begin{equation}
\left(\hat{y} \hat{x} \hat{y} + \hat{y} + 1 \right) \psi = 0.
\end{equation}

\subsubsection*{Zero at $z=-1$}
We choose the integration divisor $D = [z]-[-1]$, since $z=-1$ is a zero of $p_0(x)=x$.

This is in fact the choice that is made in \cite{DDM}. In this case, we need to evaluate the limit\vspace*{-3pt}
\begin{equation}
\lim_{z_1 \to -1} \psi_1(x(z_1); D)
=\psi(D) \lim_{z_1 \to -1} p_0(x_1(z_1)) y(z_1)=0.
\end{equation}
Then the system of differential equations becomes
\begin{equation}
\hbar \frac{d}{dx} \psi_1(x;D) = \psi_2(x;D) - \frac{1}{x} \psi_1(x;D) + \hbar \frac{1}{x} \psi_1(x;D).
\end{equation}
Since $\psi_1(x;D) = \hbar x \frac{d}{dx} \psi(D)$ and $\psi_2(x;D) = - \psi(D)$, we get
\begin{equation}
\Bigl[ \hbar^2 \frac{d}{dx} x \frac{d}{dx} + \hbar \frac{d}{dx} - \hbar^2 \frac{d}{dx} + 1 \Bigr] \psi(D)= 0.
\end{equation}
This is equivalent to\vspace*{-3pt}
\begin{equation}
\Bigl[ \hbar^2 x \frac{d^2}{dx^2} + \hbar \frac{d}{dx} + 1 \Bigr] \psi(D)= 0,
\end{equation}
which corresponds to the choice of ordering
\begin{equation}
\left(\hat{x} \hat{y}^2 + \hat{y} + 1 \right) \psi = 0.
\end{equation}
This is the quantum curve that was obtained in \cite{DDM}.

\subsubsection*{Zero at $z = \infty$}
$p_0(x)=x$ has another simple zero at $\infty$. We can choose the divisor $D = [z] - [\infty]$. In this case, we need to evaluate the limit
\begin{equation}
\begin{aligned}
\lim_{z_1 \to \infty} \psi_1(x_1(z_1); D)
&=\psi(D) \lim_{z_1 \to \infty} \left(p_0(x_1(z_1)) y(z_1)\right)\\
&=\psi(D) \lim_{z_1 \to \infty} \Bigl(- \frac{z_1+1}{z_1} \Bigr)= - \psi(D).
\end{aligned}
\end{equation}

Then the system of differential equations becomes
\begin{equation}
\hbar \frac{d}{dx} \psi_1(x;D) = \psi_2(x;D) - \frac{1}{x} \psi_1(x;D) + \hbar \frac{1}{x} \left(\psi_1(x;D) +\psi(D) \right).
\end{equation}
Since $\psi_1(x;D) = \hbar x \frac{d}{dx} \psi(D)$ and $\psi_2(x;D) = - \psi(D)$, we get
\begin{equation}
\Bigl[ \hbar^2 \frac{d}{dx} x \frac{d}{dx} + \hbar \frac{d}{dx} - \hbar^2 \frac{d}{dx} + 1 - \frac{\hbar}{x} \Bigr] \psi(D)= 0.
\end{equation}
This is equivalent to
\begin{equation}
\Bigl[ \hbar^2 x \frac{d^2}{dx^2} + \hbar \frac{d}{dx} + 1 - \frac{\hbar}{x} \Bigr] \psi(D)= 0,
\end{equation}
This is still a quantization of the spectral curve, although it is not as straightforward as the previous ones. However, if we define a new $\tilde{\psi} = \frac{1}{x} \psi$, then the quantum curve for $\tilde{\psi}$ becomes
\begin{equation}
\Bigl(\hbar^2 \frac{d^2}{dx^2} x + \hbar \frac{d}{dx} + 1 \Bigr) \tilde{\psi} = 0,
\end{equation}
which corresponds to the remaining choice of ordering
\begin{equation}
\left(\hat{y}^2 \hat{x} + \hat{y} + 1 \right) \psi = 0.
\end{equation}

\subsubsection{$a>2$}

A parameterization is $(x,y) = \left(-\psfrac{z+1}{z^a}, z \right)$. For $a>2$, the ramification points are at the simple zeros of $dx$ at $z=\spfrac{a}{1-a}$ and $\infty$, and the pole of $x$ at $z=0$.

It is easy to see that the correlation functions do not have poles at $z=0$. Therefore, the only interesting choice of integration divisor here is the pole at $z=0$. The zeros of $p_0(x) = x$ are not so interesting here because the curve has degree $a>2$, hence it is not so straightforward to calculate explicitly a quantum curve for this choice of integration divisor.

\subsubsection*{Pole at $z=0$}

We choose the integration divisor $D=[z]-[0]$. We need to evaluate the limits
\begin{equation}
C_{k} = \lim_{z_1 \to 0}\sum_{i=1}^k \frac{p_{k-i}(z_1) y(z_1)^i}{x_1(z_1)}, \quad k=1, \ldots, a-1.
\end{equation}
Here, $p_0(z_1) = x(z_1)$, $p_{a-1}(z_1) = 1$, $p_a(z_1) = 1$, and all other $p_k$'s are zero. Thus, for $k=1,\ldots, a-1$, we get
\begin{equation}
C_{k}= \lim_{z_1 \to 0} y(z_1)^k= \lim_{z_1 \to 0} z_1^k= 0.
\end{equation}
Therefore, the quantum curve is
\begin{equation}
\Bigl(\hbar^a \frac{d^{a-1}}{dx^{a-1}} x \frac{d^a}{dx^a} + \hbar \frac{d}{dx} + 1 \Bigr) \psi = 0.
\end{equation}
This is the quantization
\begin{equation}
(x,y) \mto \left(\hat{x}, \hat{y} \right) = \bigl(x, \hbar \frac{d}{dx} \bigr),
\end{equation}
with the choice of ordering
\begin{equation}
\left(\hat{y}^{a-1} \hat{x} \hat{y} + \hat{y} + 1 \right) \psi = 0.
\end{equation}

\subsection{$x y^a - 1 =0$}

The case $a=2$ has been called the \emph{Bessel curve}, and was studied in \cite {DN}. The general case has not been studied yet from an enumerative geometric perspective.

The Newton polygon is a line segment joining $(0,0)$ and $(1,a)$. Thus $\lfloor \alpha_i \rfloor = 0 $ for $i=0, \ldots, a-1$, while $\lfloor \alpha_a \rfloor = 1$.

A parameterization is $(x,y) = \left(z^a, \sfrac{1}{z} \right)$. Then $R = \{0, \infty\}$. $x$ has only one pole at~$\infty$. The correlation functions do not have poles at $\infty$ however so we can use it as an integration divisor.

$p_0(x) = x$ has a zero at $z=0$, but it is in $R$, so we cannot use it for the integration divisor since the integrals will not converge.

\subsubsection*{Pole at $z = \infty$}

We choose the integration divisor $D = [z] - [\infty]$. Then we need to evaluate the limits
\begin{equation}
C_{k} = \lim_{z_1 \to \infty} \sum_{i=1}^k\frac{p_{k-i}(z_1) y(z_1)^i}{x_1(z_1)}, \quad k=1, \ldots, a-1.
\end{equation}
Here $p_0(z) = x(z)$, $p_i(z) = 0$ for $i=1,\ldots, a-1$ and $p_a(z) = -1$. Thus
\begin{equation}
C_{k}= \lim_{z_1 \to \infty} \frac{p_0(z_1) y(z_1)^k}{x_1(z_1)}= \lim_{z_1 \to \infty} y(z_1)^k= 0.
\end{equation}
Therefore, the quantum curve is
\begin{equation}
\Bigl(\hbar^a \frac{d^{a-1}}{dx^{a-1}} x \frac{d^a}{dx^a} - 1 \Bigr) \psi = 0,
\end{equation}
which corresponds to the choice of ordering
\begin{equation}
\left(\hat{y}^{a-1} \hat{x} \hat{y}- 1 \right) \psi = 0.
\end{equation}

\subsection{$x y^2 - x y + 1 = 0$}
This spectral curve arises in the enumeration of dessins d'enfants \cite{DN}. It is related to the previous case with $a=2$. Its Newton polygon is shown in Figure \ref{f:newex3}. From the Newton polygon, we see that $\lfloor \alpha_i \rfloor = 0$ for $i=0,1$ while $\lfloor \alpha_2 \rfloor = 1$.

\begin{figure}[htb]
\begin{center}
\includegraphics[width=0.1\textwidth]{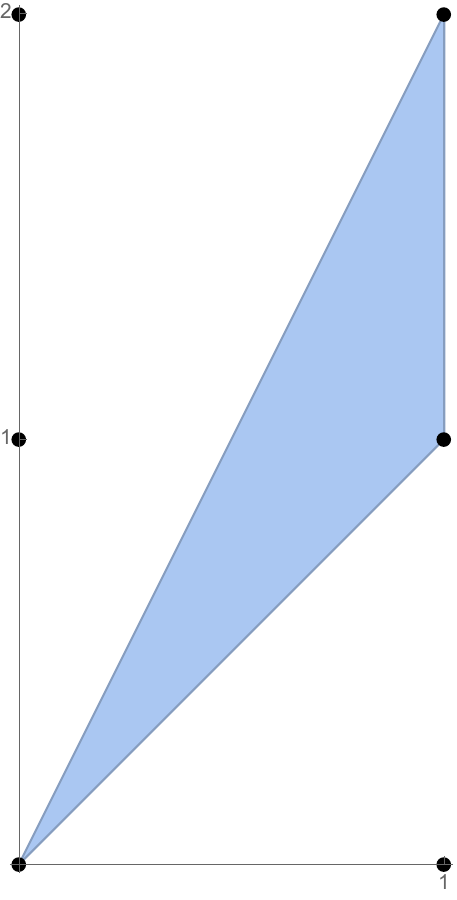}
\caption{The Newton polygon of the curve $x y^2 - x y + 1 = 0$.}
\label{f:newex3}
\end{center}
\end{figure}

A parameterization is $(x,y) = \left(\spfrac{z^2}{z-1}, \sfrac{1}{z} \right)$. Then $R = \{0,2\}$. $x$ has two simple poles at $z=1$ and $z=\infty$. $p_0(x) = x$ has a zero at $z=0$, but it is in $R$, thus we cannot use it for the integration divisor since the integrals will not converge.

\subsubsection*{Simple pole at $z=1$}
We choose the integration divisor $D = [z] - [1]$.

Then
\begin{equation}
C_1= \lim_{z_1 \to 1} \frac{p_0(z_1) y(z_1)}{x_1(z_1)}= \lim_{z_1 \to 1} \Bigl(\frac{1}{z_1} \Bigr)= 1.
\end{equation}
The quantum curve is
\begin{equation}
\Bigl(\hbar^2 \frac{d}{dx} x \frac{d}{dx} - \hbar x \frac{d}{dx} + 1 - \hbar \Bigr) \psi = 0,
\end{equation}
which is equivalent to
\begin{equation}
\Bigl(\hbar^2 \frac{d}{dx} x \frac{d}{dx} - \hbar \frac{d}{dx} x + 1 \Bigr) \psi = 0,
\end{equation}
This is the quantization
\begin{equation}
(x,y) \mto \left(\hat{x}, \hat{y} \right) = \bigl(x, \hbar \frac{d}{dx} \bigr),
\end{equation}
with choice of ordering
\begin{equation}
\left(\hat{y} \hat{x} \hat{y} - \hat{y}\hat{x} + 1 \right) \psi = 0.
\end{equation}
This is the quantum curve that was obtained in \cite{DN}.

\subsubsection*{Simple pole at $z=\infty$}

We choose the integration divisor $D = [z] - [\infty]$.

Then
\begin{equation}
C_1=\lim_{z_1 \to \infty} \frac{p_0(z_1) y(z_1)}{x_1(z_1)}= \lim_{z_1 \to \infty} \Bigl(\frac{1}{z_1} \Bigr)= 0.
\end{equation}
The quantum curve is
\begin{equation}
\Bigl(\hbar^2 \frac{d}{dx} x \frac{d}{dx} - \hbar x \frac{d}{dx} + 1 \Bigr) \psi = 0,
\end{equation}
which corresponds to the choice of ordering
\begin{equation}
\left(\hat{y} \hat{x} \hat{y} -\hat{x} \hat{y} + 1 \right) \psi = 0.
\end{equation}
The ordering is different only for the second term.

Choosing the more general integration divisor $D = [z] - \mu [1] - (1-\mu) [\infty]$ gives the family of quantum curves
\begin{equation}
\Bigl(\hbar^2 \frac{d}{dx} x \frac{d}{dx} - \hbar x \frac{d}{dx} + 1 - \mu \hbar \Bigr) \psi = 0,
\end{equation}
which interpolates between the two quantizations.

\subsection{$ x^2 y^2 + 2 x y - y + 1 = 0 $}
This example is the first that we encounter that is not linear in $x$. The Newton polygon is shown in Figure \ref{f:newex4}. From the Newton polygon, we see that $\lfloor \alpha_i \rfloor = 0$, $i=0,1$, while $\lfloor \alpha_2 \rfloor = 2$.

\begin{figure}[htb]
\begin{center}
\includegraphics[width=0.2\textwidth]{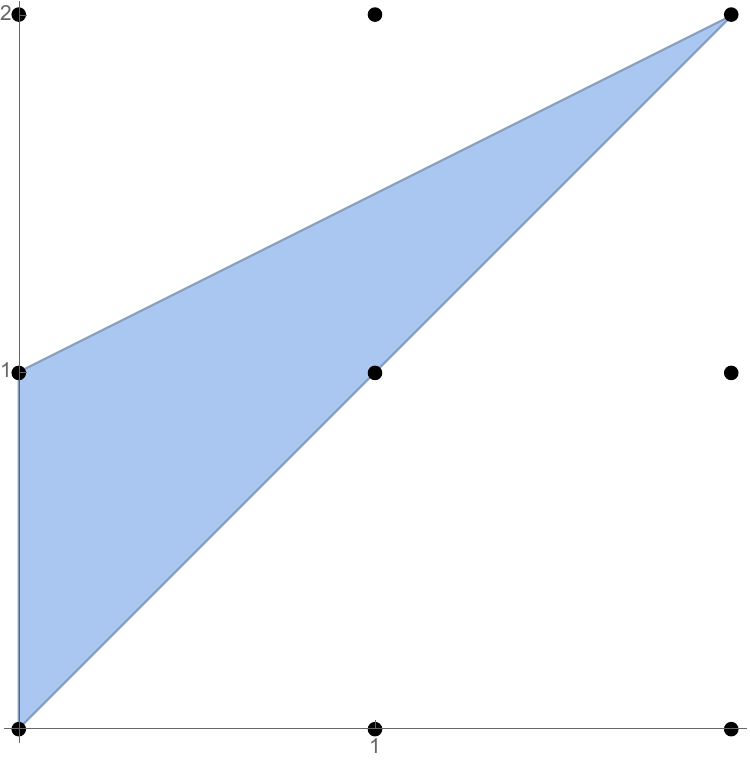}
\caption{The Newton polygon of the curve $x^2 y^2 + 2 x y - y + 1= 0$.}
\label{f:newex4}
\end{center}
\end{figure}

A parameterization is $(x,y) = \left(- \psfrac{1+z}{z^2}, z^2 \right)$. There are two ramification points; $z=-2$, the zero of $dx$, and $z=0$, the double pole of $x$. $x$ also has simple zeros at $z=-1$ and $z = \infty$.

Here the correlation functions do have poles at $z=0$, so we cannot use the double pole of $x$ for the integration divisor. But $p_0(x) = x^2$ has two zeros at $z=-1$ and $z=\infty$ that are not in $R$. We can use those for the integration divisor.

\subsubsection*{Zero at $z=-1$} We choose the integration divisor to be $D = [z] - [-1]$. Then we need to evaluate:
\begin{equation}
\begin{aligned}
\lim_{z_1 \to -1} \psi_1(x(z_1); D)
&=\psi(D) \lim_{z_1 \to -1} \left(p_0(z_1) y(z_1) \right)\\
&=\psi(D) \lim_{z_1 \to -1} \Bigl(\frac{(1+z_1)^2}{z_1^2} \Bigr)=0.
\end{aligned}
\end{equation}
Then the system of differential equations becomes
\begin{equation}
\hbar \frac{d}{dx} \left(\psi_{1}(x;D)\right) =\psi_2(x;D) - \frac{2x-1}{x^2} \psi_1(x;D)
+ \hbar \frac{1}{x} \psi_{1}(x;D).
\end{equation}
But $\psi_1(x;D) = x^2 \hbar \frac{d}{dx} \psi(D)$ and $\psi_2(x;D) = - \psi(D)$, thus the equation becomes
\begin{equation}
\Bigl[ \hbar^2 \frac{d}{dx} x^2 \frac{d}{dx} + (2x-1) \hbar \frac{d}{dx} - \hbar^2 x \frac{d}{dx} + 1 \Bigr] \psi(D) = 0.
\end{equation}
This is equivalent to
\begin{equation}
\Bigl(\hbar^2 x \frac{d}{dx} x \frac{d}{dx} + \hbar (2x-1) \frac{d}{dx} + 1 \Bigr) \psi(D) = 0,
\end{equation}
which is the quantization
\begin{equation}
(x,y) \mto \left(\hat{x}, \hat{y} \right) = \bigl(x, \hbar \frac{d}{dx} \bigr),
\end{equation}
with choice of ordering
\begin{equation}
\left(\hat{x} \hat{y} \hat{x} \hat{y} + (2 \hat{x} - 1) \hat{y} + 1 \right) \psi = 0.
\end{equation}

\subsubsection*{Zero at $z=\infty$} We choose the integration divisor to be $D = [z] - [\infty]$. Then we need to evaluate:
\begin{equation}
\lim_{z_1 \to \infty} \psi_1(x_1(z_1); D)=\psi(D) \lim_{z_1 \to \infty} \Bigl(\frac{(1+z_1)^2}{z_1^2} \Bigr)= \psi(D).
\end{equation}
Then the system of differential equations becomes
\begin{equation}
\hbar \frac{d}{dx} \left(\psi_{1}(x;D)\right) =\psi_2(x;D) - \frac{2x-1}{x^2} \psi_1(x;D)
+ \hbar \frac{1}{x} \left(\psi_{1}(x;D) - \psi(D) \right).
\end{equation}
But $\psi_1(x;D) = x^2 \hbar \frac{d}{dx} \psi(D)$ and $\psi_2(x;D) = - \psi(D)$, thus the equation becomes
\begin{equation}
\Bigl[ \hbar^2 x \frac{d}{dx} x \frac{d}{dx} + (2x-1) \hbar \frac{d}{dx} + 1 + \frac{\hbar}{x} \Bigr] \psi(D) = 0,
\end{equation}
which is a quantization, albeit not straightforward, of the spectral curve.

However, if we rewrite it in terms of the rescaled $\tilde{\psi} = \frac{1}{x} \psi$, then the quantum curve becomes
\begin{equation}
\Bigl(\hbar^2 \frac{d}{dx} x \frac{d}{d x} x + \hbar \frac{d}{d x}(2x-1) + 1 \Bigr) \tilde{\psi} = 0,
\end{equation}
which corresponds to the choice of ordering
\begin{equation}
\left(\hat{y} \hat{x} \hat{y} \hat{x} +\hat{y} (2 \hat{x} - 1) + 1 \right) \psi = 0.
\end{equation}

\subsection{$4y^2-x^2+4=0$}\label{quadratic1}

This curve appears as the spectral curve of the Gaussian matrix model (see \cite{Mehta, Wigner}). We will say more about this application below.

A parameterization is $(x,y) = \left(z + \sfrac{1}{z}, \frac12(z - \sfrac{1}{z}) \right)$.
\begin{figure}[htb]
\begin{center}
\includegraphics[width=0.2\textwidth]{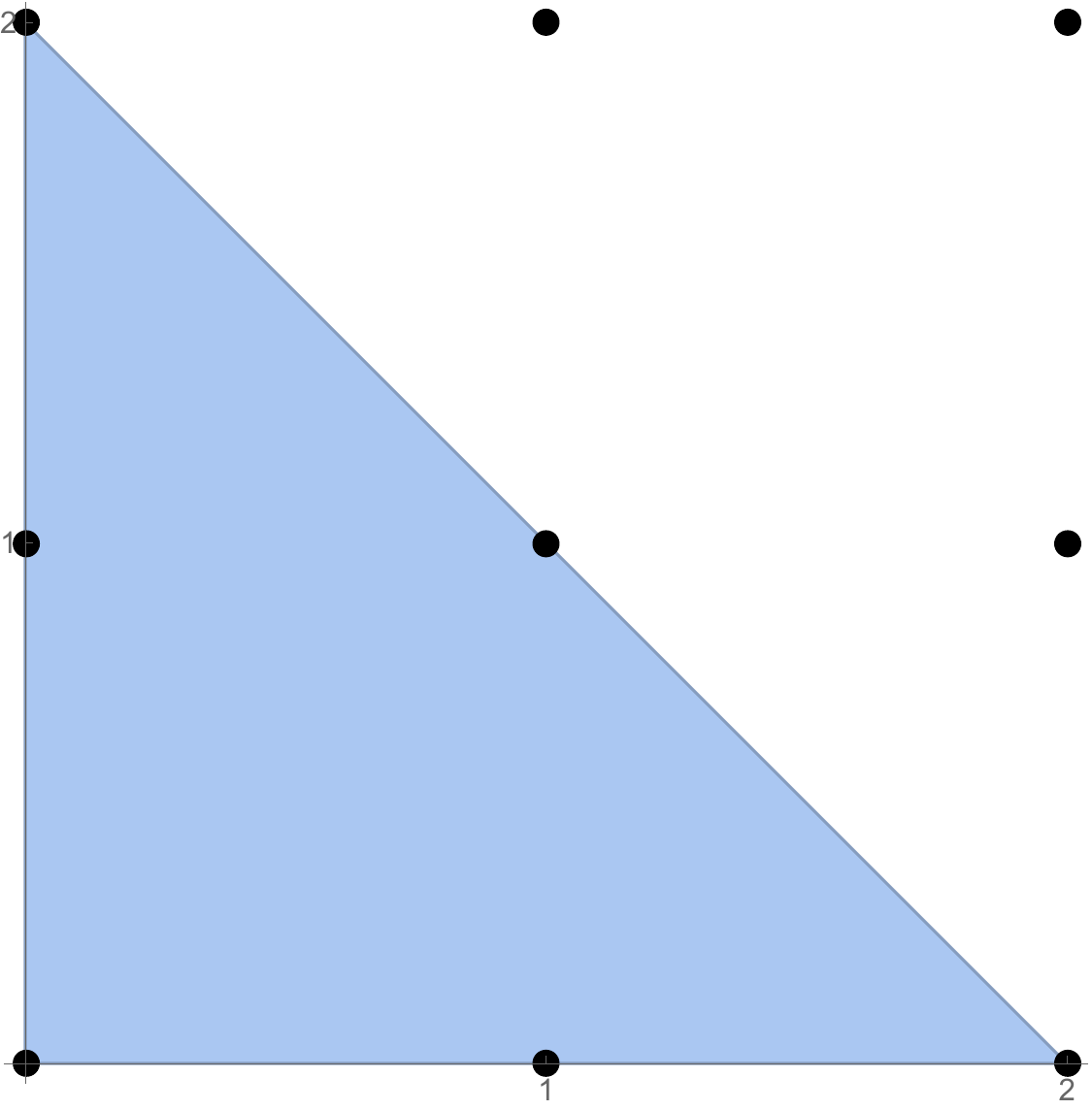}
\caption{The Newton polygon of the curve $4y^2-x^2+4=0$.}
\label{f:circle}
\end{center}
\end{figure}
Its Newton polygon is shown in Figure \ref{f:circle}. We see that $\lfloor \alpha_i \rfloor = 0$, $i=0,1,2$. We can choose the two simples poles of $x$ at $z = 0, \infty$ for the integration divisor.

\subsubsection*{Simple pole at $z = 0$}

We choose the integration divisor to be $D = [z] - [0]$. Then we need to evaluate:
\begin{equation}
C_1= \lim_{z_1 \to 0} \frac{p_0(z_1) y(z_1)}{x(z_1)}=2 \lim_{z_1 \to 0} \Bigl(\frac{z_1^2-1}{z_1^2+1} \Bigr)= - 2.
\end{equation}
The quantum curve is then
\begin{equation}\label{eq:circle}
\Bigl(4 \hbar^2 \frac{d^2}{dx^2} - x^2 + 4 + 2\hbar \Bigr) \psi = 0,
\end{equation}
which is a non-trivial quantization of the original spectral curve.

\subsubsection*{Simple pole at $z=\infty$}

The calculation for this case is similar. We get
\begin{equation}
C_1= \lim_{z_1 \to \infty} \frac{p_0(z_1) y(z_1)}{x(z_1)}=2\lim_{z_1 \to \infty} \Bigl(\frac{z_1^2-1}{z_1^2+1} \Bigr)= 2.
\end{equation}
The quantum curve is then
\begin{equation}
\Bigl(4 \hbar^2 \frac{d^2}{dx^2} -x^2 + 4 - 2\hbar \Bigr) \psi = 0,
\end{equation}
which can be obtained from the other quantum curve by $\hbar \to - \hbar$.

Choosing the more general integration divisor $D= [ z] - \mu [0] - (1-\mu) [\infty]$ gives the family of quantum curves
\begin{equation}
\Bigl(4 \hbar^2 \frac{d^2}{dx^2} -x^2 + 4 + 2 (2\mu-1)\hbar \Bigr) \psi = 0,
\end{equation}
which interpolates between the two quantizations. In particular, for $\mu=1/2$ we get the straightforward quantization
\begin{equation}
\left(4 \hat{y}^2 - \hat{x}^2 + 4 \right) \psi = 0.
\end{equation}

\subsubsection*{Relation to the Gaussian matrix model}

Let us now explain the relation between this curve and the Gaussian matrix model. Consider the Gaussian matrix integral
\begin{equation}
Z_N = \int_{H_N} dM\, e^{-\Psfrac{1}{2\hbar} \operatorname{Tr} M^2}.
\end{equation}
It is known \cite{Mehta, Wigner} that the expectation value of the characteristic polynomial
\begin{equation}
p_N(x) = \langle \det(x-M) \rangle = \frac{1}{Z_N} \int_{H_N} \det(x-M)\, dM\, e^{-\Psfrac{1}{2\hbar} \operatorname{Tr} M^2},
\end{equation}
is the monic orthogonal polynomial with respect to the measure $e^{-\sfrac{x^2}{2\hbar}}$. It is therefore the Hermite polynomial of degree $N$, in the variable $x/\sqrt\hbar$, \ie
\begin{equation}
p_N(x) = \hbar^{N/2} H_N(x/\sqrt\hbar).
\end{equation}
From it we define
\begin{equation}
\psi_N(x) = e^{\sfrac{-x^2}{4\hbar}} p_N(x).
\end{equation}
The differential equation obeyed by Hermite polynomials, $H_N''-x H'_N +N H_N=0$ implies for $\psi_N$ the equation
\begin{equation}
4 \hbar^2 \psi_N'' = (x^2-4\hbar N-2\hbar)\psi_N.
\end{equation}
The choice $N=1/\hbar$ gives
\begin{equation}
4 \hbar^2 \psi_N'' = (x^2-4-2\hbar)\psi_N.
\end{equation}
On the other hand, we have
\begin{equation}
\begin{aligned}
&\hspace*{-1cm}\bigl<\det (x-M)\bigr>\\
&= \bigl <e^{N\ln x + \operatorname{Tr} \ln(1-M/x)} \bigr> \cr
&= x^N \biggl<\exp{\int_{x'=\infty}^x \operatorname{Tr} \frac{dx'}{x'-M}-\frac{dx'}{x'}}\biggr> \cr
&= x^N \sum_{n=0}^\infty \frac{1}{n!} \int_{x'_1=\infty}^{x} \cdots \int_{x'_n=\infty}^{x} \biggl< \prod_{i=1}^n \operatorname{Tr} \frac{dx'_i}{x'_i-M}\biggr>'\cr
&= x^N \exp\biggl(\int_{x'=\infty}^{x} \Bigl< \operatorname{Tr} \frac{dx'}{x'-M}-\frac{dx'}{x'}\Bigr>\\
&\hspace*{3cm}+ \sum_{n=2}^\infty \frac{1}{n!} \int_{x'_1=\infty}^{x} \dots \int_{x'_n=\infty}^{x} \biggl< \prod_{i=1}^n \operatorname{Tr} \frac{dx'_i}{x'_i-M}\biggr>_c \biggr) \cr
&= x^N \exp{\biggl(\int_{x'=\infty}^{x'} \hspace*{-2mm}W_1(x)-\frac{Ndx'}{x'}
+ \sum_{n=2}^\infty \frac{1}{n!} \int_{x'_1=\infty}^{x} \dots \int_{x'_n=\infty}^{x} \hspace*{-2mm}W_n(x'_1,\dots,x'_n) \biggl)},
\end{aligned}
\end{equation}
where in the third line, $\langle\cdot\rangle'$ means that we should replace $dx'_i/(x'_i-M)$ by \hbox{$dx'_i/(x'_i-M)-dx'_i/x'_i$} (\ie we regularize the behavior at $\infty$, which is where the choice of the divisor matters).
In the Gaussian matrix model, all $W_n$'s have a series expansion
\begin{equation}
W_n(x_1,\dots,x_n) = \sum_{g=0}^\infty \hbar^{2g-2+n} W_{g,n}(x_1,\dots,x_n),
\end{equation}
and each coefficient $W_{g,n}$ satisfies the topological recursion with the spectral curve $4y^2=x^2-4$.

In other words, with our method we have recovered that the orthogonal polynomials associated to the Gaussian matrix model are Hermite polynomials, and our quantum spectral curve is nothing but the differential equation satisfied by Hermite polynomials.

\subsection{$xy^2 +(1-x)y-c=0$}

This curve arises in the Laguerre Random Matrix ensemble. Its quantum curve is associated to Laguerre polynomials. A parametrization is $(x,y) =\left ((c-z)/z(z-1),z \right)$.
Then $R=\{(c+ \sqrt{c^2-c}),(c- \sqrt{c^2-c})\}$.

\begin{figure}[htb]
\begin{center}
\includegraphics[width=0.2\textwidth]{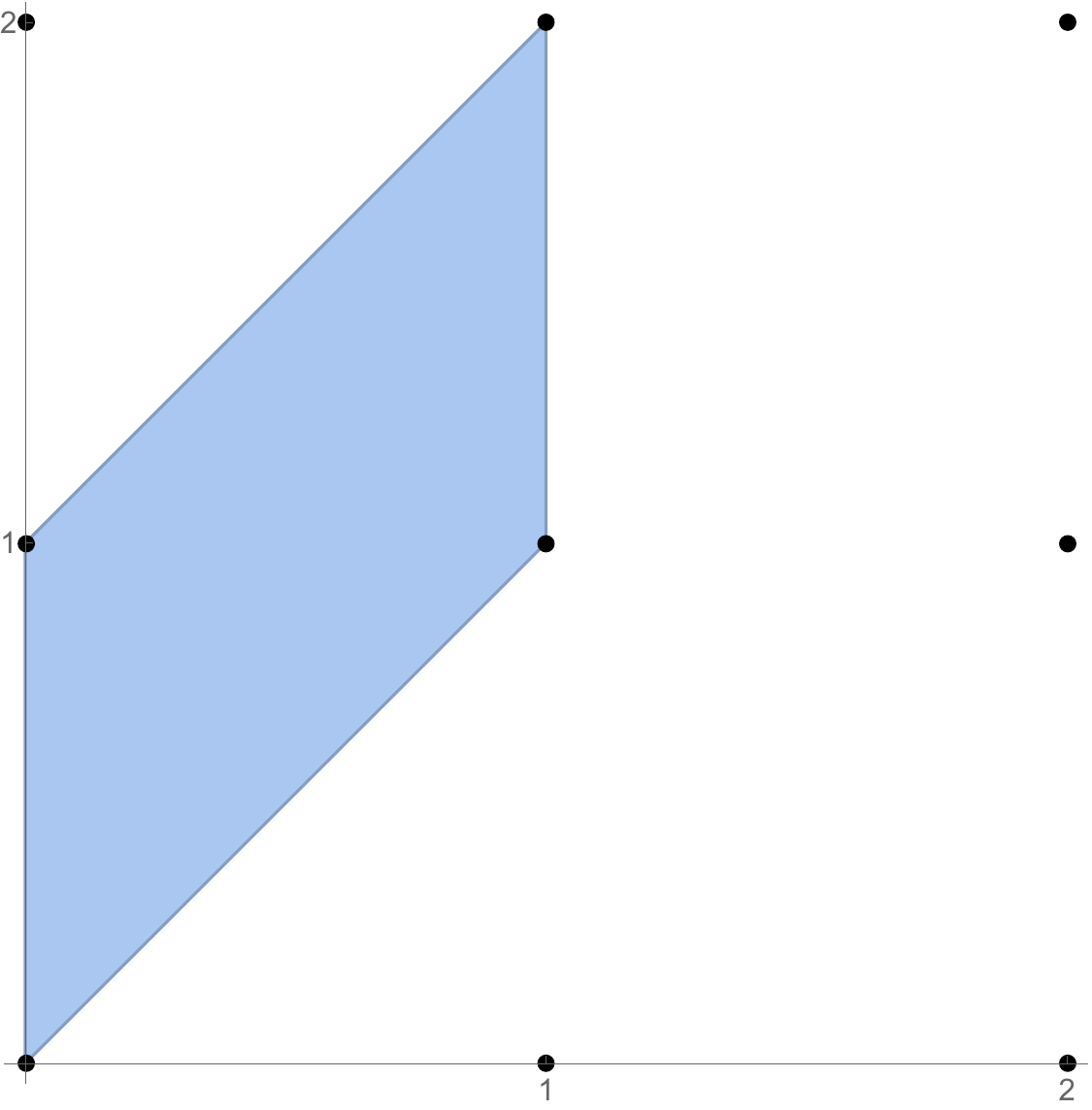}
\caption{The Newton polygon of the curve $xy^2 +(1-x)y-c=0$.}
\label{f:laguerre}
\end{center}
\end{figure}

$x$ has simple poles at $z=0,1$, and $p_0(x)=x$ vanishes at $z=c,\infty$. We can use all four of these points for the integration divisor.

Its Newton polygon is shown in Figure \ref{f:laguerre}. We see that $\lfloor \alpha_i \rfloor = 0$ for $i=0,1$, and $\lfloor \alpha_2 \rfloor = 1$.

\subsubsection*{Simple pole at $z=0$}

We choose the integration divisor $D = [z]-[0]$. We calculate
\begin{equation}
C_1 = \lim_{z_1 \to 0} \frac{p_0(z_1) y(z_1)}{x(z_1)} = \lim_{z_1 \to 0} z_1 = 0.
\end{equation}
The quantum curve is then
\begin{equation}
\Bigl(\hbar^2 \frac{d}{dx} x \frac{d}{dx} + \hbar (1-x) \frac{d}{dx} - c \Bigr) \psi =0,
\end{equation}
which is the quantization
\begin{equation}
(x,y) \mto \left(\hat{x}, \hat{y} \right) = \bigl(x, \hbar \frac{d}{dx} \bigr),
\end{equation}
with choice of ordering
\begin{equation}
\left(\hat{y} \hat{x} \hat{y} + (1 - \hat{x}) \hat{y} - c \right) \psi =0.
\end{equation}
This is indeed the Laguerre equation.

\subsubsection*{Simple pole at $z=1$}

We choose the integration divisor $D=[z]-[1]$. Evaluating the limit at $z_1 \to 1$, we get $C_1 = 1$, thus the quantum curve is
\begin{equation}
\Bigl(\hbar^2 \frac{d}{dx} x \frac{d}{dx} + \hbar (1-x) \frac{d}{dx} - c - \hbar \Bigr) \psi =0,
\end{equation}
which is equivalent to
\begin{equation}
\Bigl(\hbar^2 \frac{d}{dx} x \frac{d}{dx} + \hbar \frac{d}{dx} (1-x) - c \Bigr) \psi =0.
\end{equation}
This corresponds to the choice of ordering
\begin{equation}
\left(\hat{y} \hat{x} \hat{y} + \hat{y} (1 - \hat{x}) - c \right) \psi =0.
\end{equation}

With the more general integration divisor $D = [z] - \mu [1] - (1-\mu) [0]$, we get the family of quantum curves
\begin{equation}
\Bigl(\hbar^2 \frac{d}{dx} x \frac{d}{dx} + \hbar (1-x) \frac{d}{dx} - c - \mu \hbar \Bigr) \psi =0,
\end{equation}
which interpolates between the two quantizations.

\subsubsection*{Zero at $z = c$}

We choose the integration divisor $D = [z]-[c]$. Then we need to evaluate
\begin{equation}
\lim_{z_1 \to 1} \psi_1(x_1(z_1); D) = \psi(D) \lim_{z_1 \to c} x(z_1) y(z_1) = 0.
\end{equation}
The system becomes
\begin{equation}
\hbar \frac{d}{d x} \psi_1(x;D) = \psi_2(x;D) - \frac{1-x}{x} \psi_1(x;D) + \hbar \frac{1}{x} \psi_1(x;D).
\end{equation}
But $\psi_2(x;D) = c \psi(D)$ and $\psi_1(x;D) = x \hbar \frac{d}{dx} \psi(D)$, thus we get the quantum curve
\begin{equation}
\Bigl(\hbar^2 \frac{d}{dx} x \frac{d}{dx} + (1-x) \hbar \frac{d}{dx} - \hbar^2 \frac{d}{dx}- c \Bigr) \psi = 0,
\end{equation}
which is equivalent to
\begin{equation}
\Bigl(\hbar^2 x \frac{d^2}{dx^2} + (1-x) \hbar \frac{d}{dx}- c \Bigr) \psi = 0.
\end{equation}
This corresponds to the choice of ordering
\begin{equation}
\left(\hat{x} \hat{y}^2 +(1 - \hat{x}) \hat{y} - c \right) \psi =0.
\end{equation}

\subsubsection*{Zero at $z=\infty$}

We choose the integration divisor $D = [z]-[\infty]$. In this case we get
\begin{equation}
\lim_{z_1 \to \infty} \psi_1(x_1(z_1); D) = \psi(D) \lim_{z_1 \to \infty} x(z_1) y(z_1) = - \psi(D).
\end{equation}
The system becomes
\begin{equation}
\hbar \frac{d}{d x} \psi_1(x;D) = \psi_2(x;D) - \frac{1-x}{x} \psi_1(x;D) + \hbar \frac{1}{x} \left(\psi_1(x;D) + \psi(D) \right).
\end{equation}
But $\psi_2(x;D) = c \psi(D)$ and $\psi_1(x;D) = x \hbar \frac{d}{dx} \psi(D)$, thus we get the quantum curve
\begin{equation}
\Bigl(\hbar^2 \frac{d}{dx} x \frac{d}{dx} + (1-x) \hbar \frac{d}{dx} - \hbar^2 \frac{d}{dx}- c - \hbar \frac{1}{x} \Bigr) \psi = 0,
\end{equation}
which is equivalent to
\begin{equation}
\Bigl(\hbar^2 x \frac{d^2}{dx^2} + (1-x) \hbar \frac{d}{dx}- c - \hbar \frac{1}{x} \Bigr) \psi = 0.
\end{equation}
This is another quantization of the spectral curve, although not as straightforward.

\subsection{$x^5 y^2 + x^2 y + 1 = 0$}
This curve is an interesting example that is of higher degree in $x$. A parameterization is $(x,y) = \left(\psfrac{z-1}{z^2}, - \sfrac{z^5}{(z-1)^3} \right)$. Then $R =\{0,2\}$. The Newton polygon is shown in Figure \ref{f:deg5}. We see that $\lfloor \alpha_0 \rfloor = 0$, $\lfloor \alpha_1 \rfloor = 2$ and $\lfloor \alpha_2 \rfloor = 5$.

\begin{figure}[htb]
\begin{center}
\includegraphics[width=0.25\textwidth]{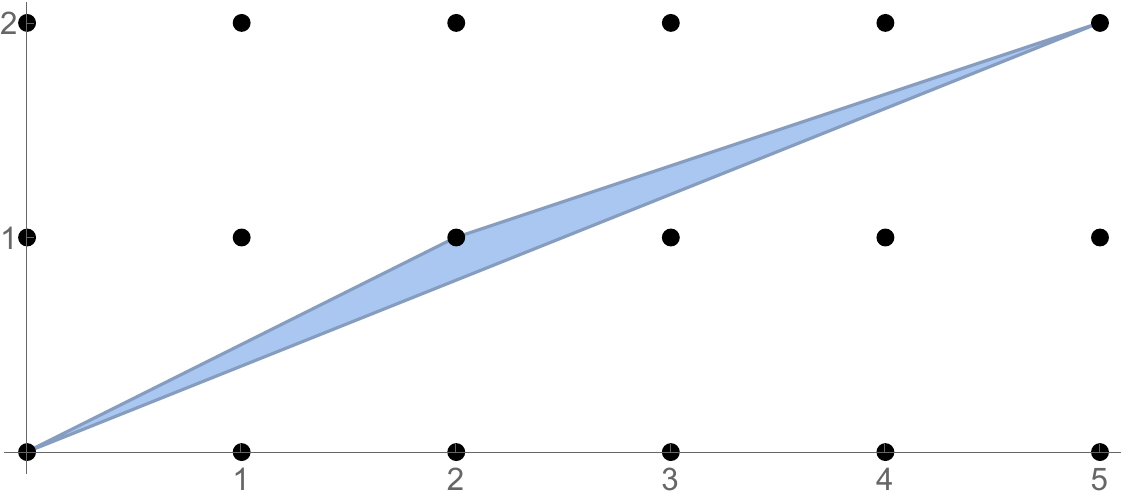}
\caption{The Newton polygon of the curve $x^5 y^2 + x^2 y + 1 = 0$.}
\label{f:deg5}
\end{center}
\end{figure}

The only pole of $x$ is at $z=0$. It is a double pole, hence is in $R$. In fact, the correlation functions have poles at $z=0$, so we cannot use it for the integration divisor since the integrals would not converge.

However, $p_0(x) = x^5$ has zeros at $z=1$ and $z=\infty$ that we can use.

\subsubsection*{Zero at $z=1$}

We choose the integration divisor $D = [z]-[1]$. We need to evaluate
\begin{equation}
\begin{aligned}
\lim_{z_1 \to 1} \psi_1(x(z_1); D)
&=\psi(D) \lim_{z_1 \to 1} \Bigl(\frac{p_0(z_1) y(z_1)}{x(z_1)^{\lfloor \alpha_1 \rfloor}} \Bigr)\\
&=- \psi(D) \lim_{z_1 \to 1} \Bigl(\frac{(z_1-1)^3}{z_1^{6}} \frac{z_1^5}{(z_1-1)^3} \Bigr)= - \psi(D).
\end{aligned}
\end{equation}
Then the differential equation becomes
\begin{equation}
\hbar \frac{d}{dx} \left(\psi_{1}(x;D)\right) = \frac{1}{x^{2}} \psi_2(x;D) - \frac{x^4}{x^7} \psi_1(x;D)
+ \hbar \frac{1}{x} \left(\psi_{1}(x;D) + \psi(D) \right).
\end{equation}
But $\psi_1(x;D) = x^3 \hbar \frac{d}{dx} \psi(D)$ and $\psi_2(x;D) = - \psi(D)$, hence we get
\begin{equation}
\Bigl[ \hbar^2 \frac{d}{dx} x^3 \frac{d}{dx} + \hbar \frac{d}{dx} - x^2 \hbar^2 \frac{d}{dx} + \frac{1}{x^2} - \frac{\hbar}{x} \Bigr] \psi(D)=0.
\end{equation}
This is equivalent to
\begin{equation}
\Bigl[ \hbar^2 x^3 \frac{d}{dx} x^2 \frac{d}{dx} + \hbar x^2 \frac{d}{dx} + 1 - \hbar x \Bigr] \psi(D)=0,
\end{equation}
which is a non-trivial quantization of the original spectral curve.

\subsubsection*{Zero at $z=\infty$}

We choose the integration divisor $D = [z]-[\infty]$. We need to evaluate
\begin{equation}
\lim_{z_1 \to \infty} \psi_1(x(z_1); D)=- \psi(D) \lim_{z_1 \to \infty} \Bigl(\frac{(z_1-1)^3}{z_1^{6}} \frac{z_1^5}{(z_1-1)^3} \Bigl)= 0.
\end{equation}
Then the differential equation becomes
\begin{equation}
\hbar \frac{d}{dx} \left(\psi_{1}(x;D)\right) = \frac{1}{x^{2}} \psi_2(x;D) - \frac{x^4}{x^7} \psi_1(x;D)
+ \hbar \frac{1}{x} \psi_{1}(x;D).
\end{equation}
But $\psi_1(x;D) = x^3 \hbar \frac{d}{dx} \psi(D)$ and $\psi_2(x;D) = - \psi(D)$, hence we get
\begin{equation}
\Bigl[ \hbar^2 \frac{d}{dx} x^3 \frac{d}{dx} + \hbar \frac{d}{dx} - x^2 \hbar^2 \frac{d}{dx} + \frac{1}{x^2} \Bigr] \psi(D)=0.
\end{equation}
This is equivalent to
\begin{equation}
\Bigl[ \hbar^2 x^3 \frac{d}{dx} x^2 \frac{d}{dx} + \hbar x^2 \frac{d}{dx} + 1 \Bigr] \psi(D)=0,
\end{equation}
which is the quantization
\begin{equation}
(x,y) \mto \left(\hat{x}, \hat{y} \right) = \bigl(x, \hbar \frac{d}{dx} \bigr),
\end{equation}
with choice of ordering
\begin{equation}
\left(\hat{x}^3 \hat{y} \hat{x}^2 \hat{y} +\hat{x}^2 \hat{y} + 1 \right) \psi = 0.
\end{equation}

\subsection{$x^4 y^4 - x^3 y^2 + x + 3 = 0$}

This is a fun curve since it is of degree $4$ in $y$ and it has non-zero $\lfloor \alpha_i \rfloor$. A parameterization is $(x,y) = \left(\ppfrac{z^4+3}{z^2-1}, \spfrac{z (z^2-1)}{z^4+3} \right)$. Then $R=\{0, i, -i, \sqrt{3}, - \sqrt{3}, \infty \}$. $x$ has poles at $z = \pm 1$ and at $\infty$. Its Newton polygon is shown in Figure \ref{f:deg4}. We see that $\lfloor \alpha_i \rfloor = i$ for $i=0,\ldots,4$.

\begin{figure}[htb]
\begin{center}
\includegraphics[width=0.2\textwidth]{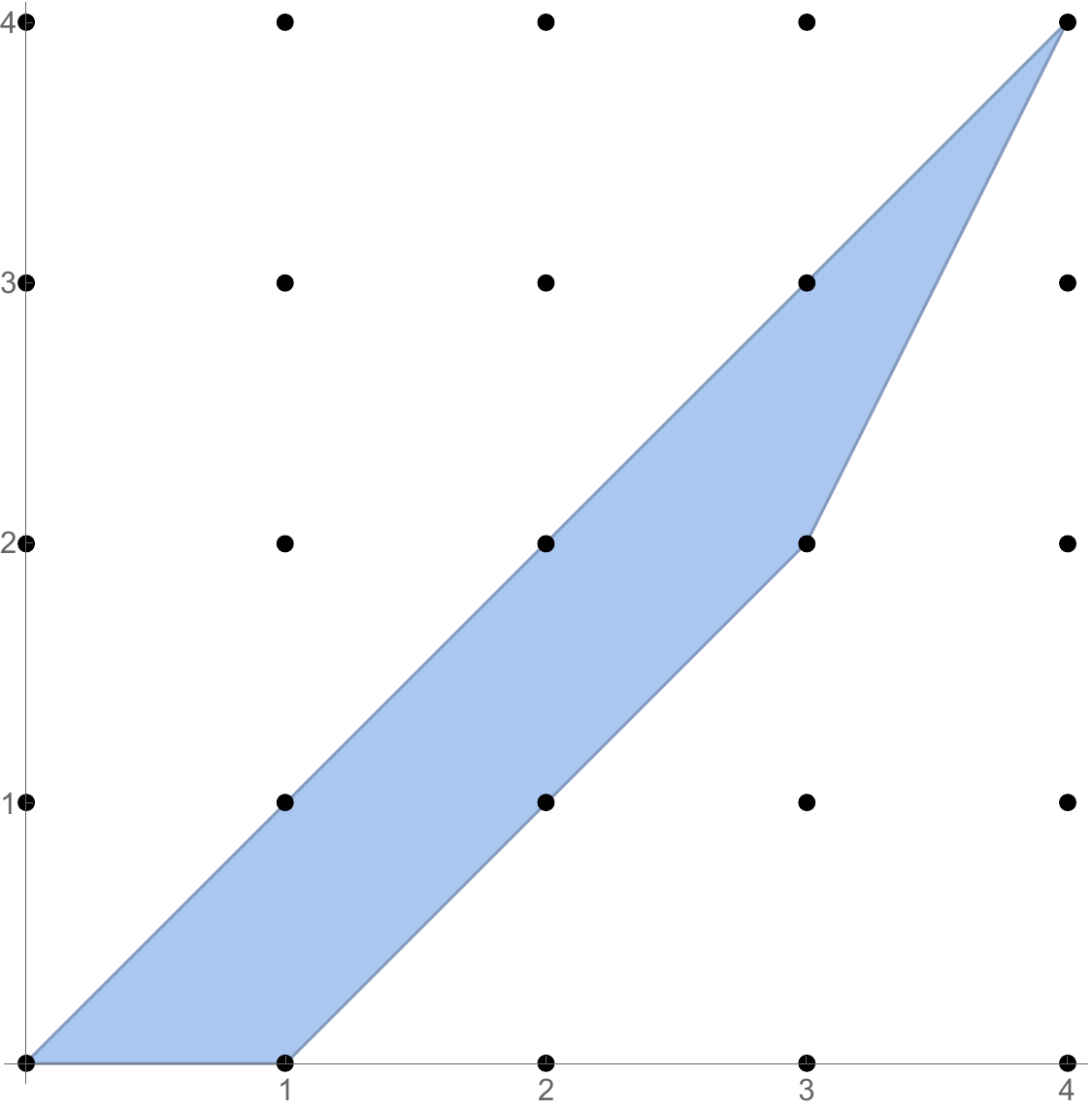}
\caption{The Newton polygon of the curve $x^4 y^4 - x^3 y^2 + x + 3 = 0$.}
\label{f:deg4}
\end{center}
\end{figure}

We can use the two poles of $x$ at $z = \pm 1$ for the integration divisor. As for the pole at $\infty$, it is in $R$, but we can still use it since the correlation functions have no pole at this point, hence the integrals converge.

\subsubsection*{Simple pole at $z=1$}

We choose the integration divisor $D = [z]-[1]$. We calculate:
\begin{equation}
C_1 = \lim_{z_1 \to 1} \frac{p_0(z_1) y(z_1)}{x(z_1)^{\lfloor \alpha_3 \rfloor + 1}}= \lim_{z_1 \to 1} y(z_1) = 0,
\end{equation}
\begin{equation}
C_2 =\lim_{z_1 \to 1} \frac{p_0(z_1) y(z_1)^2 + p_1(z_1) y(z_1)}{x(z_1)^{\lfloor \alpha_2 \rfloor + 1}}= \lim_{z_1 \to 1} x(z_1) y(z_1)^2 = 0,
\end{equation}
and
\begin{equation}
\begin{aligned}
C_3 &=\lim_{z_1 \to 1} \frac{p_0(z_1) y(z_1)^3 + p_1(z_1) y(z_1)^2 + p_2(z_1) y(z_1)}{x(z_1)^{\lfloor \alpha_1 \rfloor + 1}}\\
&= \lim_{z_1 \to 1} \left(x(z_1)^2 y(z_1)^3 - x(z_1) y(z_1) \right) = -1.
\end{aligned}
\end{equation}
The quantum curve is then
\begin{equation}
\Bigl(\hbar^4 x \frac{d}{dx} x \frac{d}{dx} x \frac{d}{dx} x \frac{d}{dx} -\hbar^2 x \frac{d}{dx} x^2 \frac{d}{dx} + x + 3 + \hbar x \Bigr) \psi =0,
\end{equation}
which is a quantization of the spectral curve.

\subsubsection*{Simple pole at $z=-1$}

The calculation is very similar for the integration divisor $D = [z]-[-1]$. Evaluating the limits at $z_1 \to -1$, we get $C_1=0$, $C_2 = 0$ and $C_3 = 1$. Therefore, the quantum curve is
\begin{equation}
\Bigl(\hbar^4 x \frac{d}{dx} x \frac{d}{dx} x \frac{d}{dx} x \frac{d}{dx} -\hbar^2 x \frac{d}{dx} x^2 \frac{d}{dx} + x + 3 - \hbar x \Bigr) \psi =0,
\end{equation}
which can be obtained from the previous quantum curve by $\hbar \to - \hbar$.

\subsubsection*{Pole at $\infty$}

This case is a little different. We choose the integration divisor $D = [z]-[\infty]$. Calculating the limits at $z_1 \to \infty$, we get $C_1 = 0$, $C_2 = 1$ and $C_3 = 0$. Therefore, in this case the quantum curve is
\begin{equation}
\Bigl(\hbar^4 x \frac{d}{dx} x \frac{d}{dx} x \frac{d}{dx} x \frac{d}{dx} -\hbar^2 x \frac{d}{dx} x^2 \frac{d}{dx} - \hbar^2 x \frac{d}{dx} x + x + 3 \Bigr) \psi =0.
\end{equation}
This can be rewritten as
\begin{equation}
\left(\hbar^4 x \frac{d}{dx} x \frac{d}{dx} x \frac{d}{dx} x \frac{d}{dx} -\hbar^2 x \frac{d}{dx} x \frac{d}{dx} x + x + 3 \right) \psi =0,
\end{equation}
which is the quantization
\begin{equation}
(x,y) \mto \left(\hat{x}, \hat{y} \right) = \bigl(x, \hbar \frac{d}{dx} \bigr),
\end{equation}
with choice of ordering
\begin{equation}
\left(\hat{x}\hat{y}\hat{x}\hat{y}\hat{x}\hat{y}\hat{x}\hat{y} - \hat{x} \hat{y}\hat{x}\hat{y}\hat{x} + \hat{x} + 3 \right) \psi =0.
\end{equation}

We can choose the more general integration divisor $D = [z] - \mu_1 [1] - \mu_2[-1] - (1-\nobreak\mu_1-\nobreak\mu_2)[\infty]$. We get the two-parameter family of quantum curves
\begin{multline}
\Bigl(\hbar^4 x \frac{d}{dx} x \frac{d}{dx} x \frac{d}{dx} x \frac{d}{dx} -\hbar^2 x \frac{d}{dx} x^2 \frac{d}{dx}\\
- (1-\mu_1-\mu_2) \hbar^2 x \frac{d}{dx} x+ x + 3 + (\mu_1 -\mu_2) \hbar x \Bigr) \psi =0,
\end{multline}
which interpolates between the three quantizations above.

\section{\texorpdfstring{$r$}{r}-spin intersection numbers and \texorpdfstring{$r$}{r}-Airy curve}

In this section we study in more detail a particularly interesting example of a spectral curve, namely:
\begin{equation}
y^r-x = 0,
\end{equation}
which we call the \emph{$r$-Airy curve}. We proved that $\psi$ as constructed in \eqref{e:psi} for the $r$-Airy curve is the WKB asymptotic solution to the differential equation
\begin{equation}
\Bigl(\hbar^r \frac{d^r}{dx^r} - x \Bigr) \psi = 0.
\end{equation}

On the one hand, this differential equation is intimately related to the $r$-KdV integrable hierarchy. On the other hand, Witten's conjecture \cite{Witten}, which was proven in \cite{FSZ}, states that certain generating functions for $r$-spin intersection numbers satisfy the $r$-KdV integrable hierarchy. From these two statements we can deduce that the meromorphic differentials constructed from the topological recursion on the $r$-Airy curve should be generating functions for $r$-spin intersection numbers. More precisely,
\begin{thm}\label{t:rspin}
Let $W_{g,n}(z_1,\ldots,z_n)$ be the meromorphic differentials constructed from the topological recursion applied to the spectral curve $y^r-x=0$, with parameterization $y=z, x=z^r$. Then
\begin{multline}
W_{g,n}(z_1,\ldots,z_n)\\[-8pt]
= (-r)^{g-1} (-1)^n \sum_{\substack{0 \leq j_1, \ldots, j_n \leq r-2\\0 \leq m_1, \ldots, m_n}} \prod_{i=1}^n \frac{c_{m_i, j_i} (r m_i + j_i + 1) d z_i}{z_i^{m_i r + j_i + 2}} \biggl \langle \prod_{i=1}^n \tau_{m_i, j_i} \biggr \rangle_g,
\end{multline}
where the $\left \langle \prod_{i=1}^n \tau_{m_i, j_i} \right \rangle_g$ are intersection numbers over the moduli space of $r$-spin curves (in standard notation) and
\begin{equation}
c_{m,j} = (-1)^m \frac{\Gamma\left(m + \psfrac{j+1}{r} \right)}{\Gamma\left(\psfrac{j+1}{r} \right)}, \quad j=0, \ldots, r-2.
\end{equation}
We note that the intersection numbers are non-vanishing only if
\begin{equation}
\sum_{i=1}^n (r m_i + j_i + 1) = (r+1) (2g-2+n).
\end{equation}
\end{thm}

This theorem was announced in September 2014 \cite{BEAIM}. A detailed proof of this theorem based on matrix model analysis will be provided in a future publication \cite{BEf}. Meanwhile, an alternative proof was presented recently in the preprint \cite{DBNOPS} by relating the topological recursion to Dubrovin's superpotential.

\subsection{Calculations}

We can use this result to calculate $r$-spin intersection numbers. Here are a few example calculations.

\subsubsection{$r=3$}

Let us list some results for $r=3$. The first few meromorphic differentials are:
\begin{align*}
W_{0,3}(z_1,z_2,z_3) &= d z_1 d z_2 d z_3 \left(\frac{2}{3 z_1^3 z_3^2 z_2^2}+\frac{2}{3 z_1^2 z_3^3
z_2^2}+\frac{2}{3 z_1^2 z_3^2 z_2^3}\right),\\
W_{0,4}(z_1,z_2,z_3,z_4) &= d z_1 d z_2 d z_3 d z_4 \Big(\frac{10}{9 z_1^6 z_3^2 z_4^2z_2^2}+\frac{8}{9 z_1^5 z_3^3
z_4^2 z_2^2}+\frac{8}{9 z_1^3
z_3^5 z_4^2 z_2^2}\\
&+\frac{10}{9z_1^2 z_3^6 z_4^2
z_2^2}+\frac{8}{9 z_1^5 z_3^2
z_4^3 z_2^2}+\frac{8}{9 z_1^2
z_3^5 z_4^3 z_2^2}+\frac{8}{9
z_1^3 z_3^2 z_4^5
z_2^2}+\frac{8}{9 z_1^2 z_3^3
z_4^5 z_2^2}\\
&+\frac{10}{9 z_1^2
z_3^2 z_4^6 z_2^2}+\frac{8}{9
z_1^5 z_3^2 z_4^2
z_2^3}+\frac{8}{9 z_1^2 z_3^5
z_4^2 z_2^3}-\frac{16}{9 z_1^3
z_3^3 z_4^3 z_2^3}+\frac{8}{9
z_1^2 z_3^2 z_4^5
z_2^3}\\
& +\frac{8}{9 z_1^3 z_3^2
z_4^2 z_2^5}+\frac{8}{9 z_1^2
z_3^3 z_4^2 z_2^5}+\frac{8}{9
z_1^2 z_3^2 z_4^3
z_2^5}+\frac{10}{9 z_1^2 z_3^2
z_4^2 z_2^6}\Big),\\
W_{1,1}(z_1)&= d z_1 \Bigl(\frac{1}{9 z_1^5} \Bigr),\\
W_{1,2}(z_1,z_2) &= d z_1 d z_2 \Bigl(\frac{7}{27 z_2^8 z_1^2}+\frac{4}{27 z_2^5
z_1^5}+\frac{7}{27 z_2^2 z_1^8} \Bigr)\\
W_{2,1} (z_1) &= 0,
\end{align*}
\begin{align*}
W_{2,2}(z_1,z_2) &= d z_1 d z_2 \Bigl(-\frac{770}{729 z_2^{15} z_1^3}-\frac{605}{729 z_2^{12}
z_1^6}-\frac{680}{729 z_2^9 z_1^9}\\
&\hspace*{5.3cm}-\frac{605}{729 z_2^6
z_1^{12}}-\frac{770}{729 z_2^3 z_1^{15}} \Bigr),\\[-5pt]
W_{3,1} (z_1) &= d z_1 \Bigl(-\frac{32725}{19683 z_1^{21}} \Bigr).
\end{align*}

From these we can extract intersection numbers for $r=3$. We get, at genus 0:
\begin{gather}
\langle \tau_{0,0}^2 \tau_{0,1} \rangle_0 = 1, \quad \langle \tau_{0,0}^3 \tau_{1,1} \rangle_0 = 1,\quad
\langle \tau_{0,0}^2 \tau_{0,1} \tau_{1,0} \rangle_0 = 1, \quad \langle \tau_{0,1}^4 \rangle_0 = \frac{1}{3}.
\end{gather}
At genus 1,
\begin{equation}
\begin{gathered}
\langle \tau_{1,0} \rangle_1 = \frac{1}{12}, \quad \langle \tau_{0,0} \tau_{2,0} \rangle_1 = \frac{1}{12}, \quad
\langle \tau_{1,0}^2 \rangle_1 = \frac{1}{12}, \quad \langle \tau_{0,0}^2 \tau_{3,0} \rangle_1 = \frac{1}{12},\\
\langle \tau_{0,1}^2 \tau_{2,1} \rangle_1 = \frac{1}{36}, \quad
\langle \tau_{0,0} \tau_{1,0} \tau_{2,0} \rangle_1 = \frac{1}{6}, \quad \langle \tau_{0,1} \tau_{1,1}^2 \rangle_1 = \frac{1}{36}, \quad \langle \tau_{1,0}^3 \rangle_1 = \frac{1}{6}.
\end{gathered}
\end{equation}
At genus 2,
\begin{gather}
\langle \tau_{0,1} \tau_{4,1} \rangle_2 = \frac{1}{864}, \quad \langle \tau_{1,1} \tau_{3,1} \rangle = \frac{11}{4320}, \quad \langle \tau_{2,1}^2 \rangle_2 = \frac{17}{4320}.
\end{gather}
And finally, at genus 3,
\begin{equation}
\langle \tau_{6,1} \rangle_3 = \frac{1}{31104}.
\end{equation}
Those results agree with known r-spin intersection numbers (see for instance \cite{BH}).

\subsubsection{$r=4$}

The first few meromorphic differentials are
\begin{align*}
W_{0,3} (z_1, z_2, z_3) &= d z_1 d z_2 d z_3 \Bigl(\frac{3}{4 z_2^4 z_3^2
z_1^2}+\frac{1}{z_2^3 z_3^3
z_1^2}+\frac{3}{4 z_2^2 z_3^4
z_1^2}+\frac{1}{z_2^3 z_3^2
z_1^3}\\
&\hspace*{6.7cm}+\frac{1}{z_2^2 z_3^3
z_1^3}+\frac{3}{4 z_2^2 z_3^2
z_1^4} \Bigr),\\[-8pt]
W_{1,1}(z_1) &= d z_1 \Bigl(\frac{5}{32 z_1^6} \Bigr),\\
W_{1,2}(z_1, z_2) &= d z_1 d z_2 \Bigl(\frac{45}{128 z_2^{10}
z_1^2}-\frac{21}{128 z_2^8
z_1^4}+\frac{25}{128 z_2^6
z_1^6}-\frac{21}{128 z_2^4
z_1^8}+\frac{45}{128 z_2^2
z_1^{10}} \Bigr),\\
W_{2,1}(z_1) &= d z_1 \Bigl(-\frac{2079}{8192 z_1^{16}} \Bigr),
\end{align*}
from which we extract the following intersection numbers.

At genus $0$,
\begin{equation}
\langle \tau_{0,0}^2 \tau_{0,2} \rangle_0 = 1, \quad \langle \tau_{0,0} \tau_{0,1}^2 \rangle_0 = 1.
\end{equation}
At genus $1$,
\begin{equation}
\langle \tau_{1,0} \rangle_1 = \frac{1}{8}, \quad \langle \tau_{0,0} \tau_{2,0} \rangle_1= \frac{1}{8}, \quad \langle \tau_{0,2} \tau_{1,2} \rangle_1 = \frac{1}{96}, \quad \langle \tau_{1,0}^2 \rangle_1= \frac{1}{8},
\end{equation}
and at genus $2$,
\begin{equation}
\langle \tau_{3,2} \rangle_2 = \frac{3}{2560}.
\end{equation}
These again match known results, such as in \cite{BH}.

We also calculated higher order meromorphic differentials and intersection numbers, and they match with known results. Calculations and results are available upon request.

\subsubsection{$r=5$}

For $r=5$, we get the following meromorphic differentials:
\begin{align*}
W_{0,3} (z_1, z_2, z_3) &= d z_1 d z_2 d z_3 \Bigl(\frac{4}{5 z_2^5 z_3^2
z_1^2}+\frac{6}{5 z_2^4 z_3^3
z_1^2}+\frac{6}{5 z_2^3 z_3^4
z_1^2}+\frac{4}{5 z_2^2 z_3^5
z_1^2}+\frac{6}{5 z_2^4 z_3^2
z_1^3}\\
&\hspace*{1.5cm}+\frac{8}{5 z_2^3 z_3^3
z_1^3}+\frac{6}{5 z_2^2 z_3^4
z_1^3}+\frac{6}{5 z_2^3 z_3^2
z_1^4}+\frac{6}{5 z_2^2 z_3^3
z_1^4}+\frac{4}{5 z_2^2 z_3^2
z_1^5}\Bigr),\\
W_{1,1}(z_1) &= d z_1 \Bigl(\frac{1}{5 z_1^7} \Bigr),\\
W_{1,2}(z_1, z_2) &= d z_1 d z_2 \Bigl(\frac{11}{25 z_2^{12}
z_1^2}-\frac{9}{25 z_2^{10}
z_1^4}-\frac{8}{25 z_2^9
z_1^5}+\frac{6}{25 z_2^7
z_1^7}\\
&\hspace*{5cm}-\frac{8}{25 z_2^5
z_1^9}-\frac{9}{25 z_2^4
z_1^{10}}+\frac{11}{25 z_2^2
z_1^{12}} \Bigr),\\
W_{2,1}(z_1) &= d z_1 \Bigl(-\frac{429}{625 z_1^{19}} \Bigr).
\end{align*}

The corresponding intersection numbers are as follows.
At genus $0$,
\begin{equation}
\langle \tau_{0,0}^2 \tau_{0,3} \rangle_0 = 1, \quad \langle \tau_{0,0} \tau_{0,1} \tau_{0,2} \rangle_0 = 1, \quad \langle \tau_{0,1}^3 \rangle_0 = 1.
\end{equation}
At genus $1$,
\begin{equation}
\begin{gathered}
\langle \tau_{1,0} \rangle_1 = \frac{1}{6}, \quad \langle \tau_{0,0} \tau_{2,0} \rangle_1= \frac{1}{6}, \\
\langle \tau_{0,2} \tau_{1,3} \rangle_1 = \frac{1}{60}, \quad
\langle \tau_{0,3} \tau_{1,2} \rangle_1= \frac{1}{60}, \quad \langle \tau_{1,0}^2 \rangle_1 = \frac{1}{6},
\end{gathered}
\end{equation}
and at genus $2$,
\begin{equation}
\langle \tau_{3,2} \rangle_2 = \frac{11}{3600}.
\end{equation}
These also match known results, such as in \cite{BH}.

We also calculated higher order meromorphic differentials and intersection numbers, and they match with known results. Calculations and results are available upon request.

\backmatter
\providecommand{\bysame}{\leavevmode\hbox to3em{\hrulefill}\thinspace}
\providecommand{\MR}{\relax\ifhmode\unskip\space\fi MR }
\providecommand{\MRhref}[2]{%
  \href{http://www.ams.org/mathscinet-getitem?mr=#1}{#2}
}
\providecommand{\href}[2]{#2}

\end{document}